\documentclass[a4paper,10pt]{article}
\pdfoutput=1

\usepackage[left=2.16cm,right=2.16cm]{geometry}

\usepackage[affil-it]{authblk}
\usepackage{enumitem}
\usepackage{fontawesome}

\usepackage{graphicx}
\usepackage{subfigure}
\usepackage[table, dvipsnames]{xcolor}
\usepackage{amsmath,amssymb,amsthm}
\usepackage{mathtools}
\usepackage{wrapfig}
\usepackage{longtable}
\usepackage{booktabs}
\usepackage{pythonhighlight}
\usepackage{array}
\usepackage{pgffor}
\usepackage{etoolbox}

\newtheorem{assumption}{Assumption}
\newtheorem{corollary}{Corollary}
\newtheorem{lemma}{Lemma}
\newtheorem{theorem}{Theorem}

\newtheorem{definition}{Definition}
\newtheorem{example}{Example}
\newtheorem{remark}{Remark}

\makeatletter
\def\bstr{b}
\def\bfstr{bf}
\def\cstr{c}
\def\fstr{f}
\def\strLst{A,B,C,D,d,E,F,G,H,I,J,K,L,M,N,O,P,Q,R,S,T,U,V,W,X,Y,Z}

\newcommand{\MkB}[1]{\expandafter\def\csname\bstr#1\endcsname{\mathbb{#1}}}
  \@for\i:=\strLst\do{%
    \expandafter\MkB \i     }

\newcommand{\MkBF}[1]{\expandafter\def\csname\bfstr#1\endcsname{\mathbf{#1}}}
  \@for\i:=\strLst\do{%
    \expandafter\MkBF \i     }

\newcommand{\MkCal}[1]{\expandafter\def\csname\cstr#1\endcsname{\mathcal{#1}}}
  \@for\i:=\strLst\do{%
    \expandafter\MkCal \i     }

\newcommand{\MkFrak}[1]{\expandafter\def\csname\fstr#1\endcsname{\mathfrak{#1}}}
  \@for\i:=\strLst\do{%
    \expandafter\MkFrak \i     }
    
\makeatother

\newcommand{\Lin}[1]{\mathop{\mathsf{Lin}}(#1)}

\newcommand{\LinAc}[1]{\overline{\mathsf{Lin}}(#1)}

\newcommand{\LinEq}[1]{\overline{\mathsf{Lin}}(#1)_{\sim}}

\newcommand{\ac}[1]{\mathsf{#1}} %

\newcommand{\Shift}{\mathsf{Shift}}

\newcommand{\Trans}{\mathsf{Trans}}

\newcommand{\MatchGT}[3]{\mathsf{M}^{{\text{\tiny $#1$}}}_{#2}(#3)}

\newcommand{\RMatchGT}[3]{\mathcal{M}^{{\text{\tiny $#1$}}}_{#2}(#3)}

\newcommand{\obs}[1]{\mathsf{obs}(#1)}

\newcommand{\pB}[1]{\mathsf{PB}(#1)}
\newcommand{\pO}[1]{\mathsf{PO}(#1)}

\newcommand{\mono}[1]{\mathsf{mono}(#1)}

\newcommand{\epi}[1]{\mathsf{epi}(#1)}

\newcommand{\mor}[1]{\mathsf{mor}(#1)}

\newcommand{\iso}[1]{\mathsf{iso}(#1)}

\newcommand{\obj}[1]{\mathsf{obj}(#1)}

\newcommand{\mIO}{\mathop{\varnothing}}

\newcommand{\canRep}[2]{\rho^{#1}_{\bfC}\left(#2\right)}

\newcommand{\bra}[1]{\left\langle #1\right\vert}
\newcommand{\ket}[1]{\left\vert #1\right\rangle}
\newcommand{\braket}[2]{\left\langle \left. #1 \right\vert #2\right\rangle}

\newcommand{\MOD}{\textsc{M\O{}D}}
\newcommand{\KAP}{\textsc{Kappa}}
\newcommand{\CHEM}{\textsc{Chem}}

\newcommand{\compGT}[4]{#2 {}^{#3}\!{\triangleleft}_{#1} #4}

\newcommand{\rap}[3]{#2 \star_{#1}{#3}}

\newcommand{\rrap}[3]{#2 \,\overline{\star}_{#1}\,{#3}}

\renewcommand{\vec}[1]{\underline{#1}}

\newcommand{\cond}[1]{\mathsf{cond}(#1)}

\newcommand{\jcOp}[1]{\hat{\bO}(#1)}

\newcommand{\Prob}[1]{\mathsf{Prob}(#1)}

\def\Ag{\Sigma_{\mathsf{ag}}}
\def\St{\Sigma_{\mathsf{ste}}}
\def\AgSt{\Sigma_{\mathsf{ag-ste}}}
\def\Prp{\Sigma_{\mathsf{prop}}}
\def\mc#1{\mathcal{#1}}
\def\tAg{\mathsf{type}}
\def\set#1{\{#1\}}
\def\pto{\mathrel{\rightharpoonup}}
\def\StSt{\Sigma_{\mathsf{ste}-\mathsf{ste}}}
\def\StP{\Sigma_{\mathsf{ste}-\mathsf{prop}}}

\newcommand{\ti}[1]{%
 \ensuremath{\vcenter{\hbox{\includegraphics{diagrams/#1.pdf}}}}%
}

\newcommand{\mi}[1]{%
 \includegraphics{diagrams/mod-diagrams/#1.pdf}%
}

\newcommand{\miScaled}[2]{%
 \includegraphics[scale=#2]{diagrams/mod-diagrams/#1.pdf}%
}

\newcommand{\aligntemp}{}

\colorlet{h1color}{blue!70!black} %
\colorlet{h2color}{orange!90!black} %
\colorlet{h3color}{blue!40!white} %
\colorlet{h4color}{green!40!black} %
\usepackage{hyperref}
\begin{document}

\title{Rewriting Theory for the Life Sciences:\\A Unifying Theory of CTMC Semantics\\
(Long version)\thanks{This is a long version (including additional results and background materials) of the \href{https://staf2020.hvl.no/events/icgt2020/}{ICGT~2020} conference paper~\cite{BK2020} (cf.\ Appendix~\ref{sec:confToExt} for further details).}}

\author[1, \faEnvelopeO]{Nicolas Behr}

\author[1]{Jean Krivine}

\author[2]{Jakob L.\ Andersen}

\author[2]{Daniel Merkle}

\affil[1]{Université de Paris, CNRS, IRIF, F-75006, Paris, France}
	
\affil[2]{Department of Mathematics and Computer Science, University of Southern Denmark, Odense M DK-5230, Denmark}

\affil[\faEnvelopeO]{Corresponding author; nicolas.behr@irif.fr}

\maketitle

\begin{abstract}
The \KAP{} biochemistry and the M\O{}D organic chemistry frameworks are amongst the most intensely developed applications of rewriting-based methods in the life sciences to date. A typical feature of these types of rewriting theories is the necessity to implement certain structural constraints on the objects to be rewritten (a protein is empirically found to have a certain signature of sites, a carbon atom can form at most four bonds, ...). In this paper, we contribute a number of original developments that permit to implement a universal theory of continuous-time Markov chains (CTMCs) for stochastic rewriting systems. Our core mathematical concepts are a novel rule algebra construction for the relevant setting of rewriting rules with conditions, both in Double- and in Sesqui-Pushout semantics, augmented by a suitable stochastic mechanics formalism extension that permits to derive dynamical evolution equations for pattern-counting statistics. A second main contribution of our paper is a novel framework of restricted rewriting theories, which comprises a rule-algebra calculus under the restriction to so-called constraint-preserving completions of application conditions (for rules considered to act only upon objects of the underlying category satisfying a globally fixed set of structural constraints). This novel framework in turn renders a faithful encoding of bio- and organo-chemical rewriting in the sense of \KAP{} and M\O{}D possible, which allows us to derive a rewriting-based formulation of reaction systems including a full-fledged CTMC semantics as instances of our universal CTMC framework. While offering an interesting new perspective and conceptual simplification of this semantics in the setting of \KAP{}, both the formal encoding and the CTMC semantics of organo-chemical reaction systems as motivated by the M\O{}D framework  are the first such results of their kind.
\end{abstract}

\section{Motivation}

One of the key applications that rewriting theory may be considered for in the life sciences is the theory of continuous-time Markov chains (CTMCs) modeling complex systems. %
In fact, since Delbr\"{u}ck's seminal work on autocatalytic reaction systems in the 1940s~\cite{Delbr_ck_1940}, %
the mathematical theory of chemical reaction systems has effectively been formulated as a rewriting theory in disguise, %
namely via the rule algebra of discrete graph rewriting~\cite{bp2019-ext}. %
In the present paper, we provide the necessary technical constructions in order to consider the CTMCs and analysis methods of relevance for more general types of compositional rewriting theories with conditions, %
with key examples provided in the form of  \emph{biochemical graph rewriting} in the sense of the \KAP{} framework (\url{https://kappalanguage.org})~\cite{Boutillier:2018aa}, %
and \emph{(organo-) chemical graph rewriting} in the sense of the \MOD{} framework (\url{https://cheminf.imada.sdu.dk/mod/})~\cite{Andersen_2016}. %

The present paper aims to serve two main purposes: the first consists in providing an extension of the existing category-theoretical rule-algebra frameworks~\cite{bp2018,bp2019-ext,nbSqPO2019} by the rewriting theoretical design feature of incorporating rules with conditions as well as constraints on objects (Section~\ref{sec:ra}). %
These technical developments then form the basis for a novel stochastic mechanics framework, providing a universal semantics of CTMCs based upon stochastic rewriting theory for rules with conditions (Section~\ref{sec:stochMech}). %

The second main theme of this paper concerns the practical implementation of rewriting-based CTMCs, with a particular focus on the application scenarios of bio- and organo-chemical reaction systems. A crucial prerequisite for efficiently implementing the enormously intricate structural constraints imposed by the chemical theories within rewriting theory, we introduce a novel restricted rule-algebraic rewriting framework (Section~\ref{sec:rrt}; cf.\ Section~\ref{sec:confToExt} for an overview). We then proceed to introduce the first-of-their-kind fully faithful encodings of the semantics of biochemical reaction systems (Section~\ref{sec:bcgr}) and of organo-chemical reaction systems (Section~\ref{sec:ocgr}), both within the rule-algebraic restricted rewriting framework.

\section{High-level overview of key concepts and results}\label{sec:hlo} 
\begin{figure}
\centering
 \begin{minipage}[t]{.475\linewidth}
 \subfigure[\label{fig:overview-rules}(Linear) rewriting rules.]{%
	$
	\begin{array}{c}
	\hphantom{XXXX}r=(O\hookleftarrow K\hookrightarrow I)\hphantom{XXXX}\\[1em]
	\end{array}
	$
}%
\\[1.5em]
\subfigure[\label{fig:overview-rar}$\bT$-type rule algebra representation $\rho^{\bT}$.]{%
	$\displaystyle{\rho^{\bT}(\delta(r))\ket{X}:=\underset{\begin{array}{c}
	\text{{\footnotesize\emph{{\color{blue}``sum over all rewritings''}}}}\\[1em]\end{array}}{\sum_{{\color{blue}m}\in \MatchGT{\bT}{r}{X}} \ket{r_{\color{blue}m}(X)}}}$
}%
\\[1.5em]
\subfigure[\label{fig:overview-rap}$\bT$-type rule algebra product.]{%
	$\displaystyle{\rap{\bT}{\delta(r_2)}{\delta(r_1)}
		:=\underset{\begin{array}{c}\text{{\footnotesize \emph{{\color{h1color}``sum over all rule compositions''}}}}\\[1em]\end{array}}{\sum_{{\color{h1color}\mu}\in \RMatchGT{\bT}{r_2}{r_1}}
			\delta\left(\compGT{\bT}{r_2}{{\color{h1color}\mu}}{r_1}\right)}}$
}%
\\[1.5em]
\subfigure[\label{fig:overview-rp}Fundamental property of representations.]{%
	$\displaystyle{
	\begin{array}{c}
	\rho^{\bT}(\delta(r_2))\rho^{\bT}(\delta(r_1))
	=\rho^{\bT}(\rap{\bT}{\delta(r_2)}{\delta(r_1)})\\[1em]
	\end{array}}$
}%
 \end{minipage}
  \begin{minipage}[t]{.5\linewidth}
  \subfigure[\label{fig:overview-dd}$\bT$-type direct derivations.]{%
	$\begin{array}{c}
		\ti{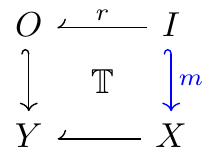}\quad :=\quad 
		\ti{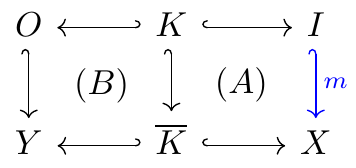}
		\\
		\\
		{\footnotesize
		\begin{array}{c|ccc}
		\bT & DPO & DPO^{\dag} & SqPO\\
		\hline
		(A) & \mathsf{POC} & \mathsf{PO} & \mathsf{FPC}\\
		(B) & \mathsf{PO} & \mathsf{POC} & \mathsf{PO}
		\end{array}}\\[0.75em]
		\text{{\footnotesize $\mathsf{PO}$ -- pushout; $\mathsf{POC}$ -- pushout complement}}\\[-0.25em]
		 \text{{\footnotesize $\mathsf{FPC}$ -- final pullback complement (cf.\ Definition~\ref{def:FPC})}}\\[0.5em]
		\end{array}$
}%
\\[1.5em]
\subfigure[\label{fig:overview-rc}$\bT$-type rule compositions.]{%
	$\begin{array}{c}
	\ti{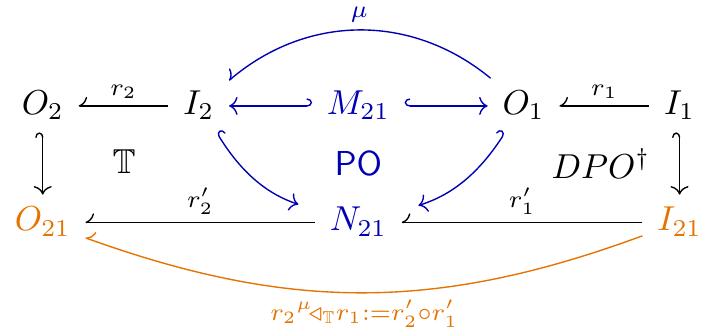}\\\vphantom{X}
		\end{array}$
}%
  \end{minipage}
\caption{Overview of the rule algebra framework for the setting of rewriting with ``plain'' rules in Double Pushout (DPO)~\cite{bp2018} and Sesqui-Pushout (SqPO) semantics~\cite{nbSqPO2019}.}
\label{fig:overviewRAtheory}
\end{figure}

The core principle behind the materials presented in this paper is the so-called \emph{rule algebra framework} for categorical rewriting theories in the sense of~\cite{bp2018,nbSqPO2019}. The essential mathematical ideas of the original form of this framework for the setting of rewriting theories without constraints or conditions on rules (i.e., for theories based upon ``plain'' rules) are collected in Figure~\ref{fig:overviewRAtheory}.

\paragraph{Basic categorical rewriting theory} \emph{(Linear) rewriting rules} (Figure~\ref{fig:overview-rules}) are encoded as spans of monomorphisms; a given rewriting rule $r=(O\hookleftarrow K\hookrightarrow I)$ consists of an \emph{input object} $I$, an \emph{output object} $O$, and a \emph{kontext object} $K$, together with embeddings of $K$ into $O$ and $I$. With the precise details depending on the chosen semantics, such a rule loosely speaking encodes rewriting operations where a copy of $I$ is to be picked in a target graph, of which only the image of $K\hookrightarrow I$ is retained, followed by extending $K$ into $O$ as encoded in $K\hookrightarrow O$. The mathematically precise definition of the action of some rule $r$ onto a target object, typically referred to as a \emph{direct derivation}, is provided in Figure~\ref{fig:overview-dd} (for Double Pushout (DPO), its ``reverse'' variant $DPO^{\dag}$, and Sesqui-Pushout (SqPO) semantics). For instance, when considering the rewriting of directed multigraphs, the choice of SqPO- or DPO-semantics controls in effect whether or not vertex deletion operations can implicitly delete incident edges, respectively. Notationally, it is customary to write $r_{{\color{blue}m}}(X):=Y$ for the object that results from a DPO- or SqPO-type direct derivation as in Figure~\ref{fig:overview-dd}.

\paragraph{Non-determinism in rule applications} Since in general an object $X$ may be rewritten via some rewriting rule $r$ in multiple ways (each of which is given via a so-called \emph{admissible match} ${\color{blue}m}\in \MatchGT{\bT}{r}{X}$, for $\bT\in \{DPO,SqPO\}$), this allows for an interesting mathematical operation as depicted in Figure~\ref{fig:overview-rar}: letting $\ket{X}$ denote a basis vector of some $\bR$-vector space ``over'' objects, one may define a linear operator $\rho^{\bT}(\delta(r))$ that maps $\ket{X}$ to the ``sum over all (outcomes of) rewritings''. While there would in principle be many possible choices for the concrete semantics of this intuitive operation feasible, the particular choice employed in the rule algebra framework (cf.\ Section~\ref{sec:ra} for further details) amounts to considering ``states'' $\ket{X}$ as indexed by objects up to isomorphisms, and rule algebra elements $\delta(r)$ as indexed by rules up to isomorphisms. In this sense, the operation depicted in Figure~\ref{fig:overview-rar} amounts to a form of ``book-keeping'' of possible outcomes of applying rule $r$ to object $X$, with outcomes classified by isomorphisms (so that the coefficients when evaluating the sum over outcomes amount to non-negative integers encoding the numbers of ways a given isomorphism class of objects can be obtained via application of $r$ to $X$).

\paragraph{Non-determinism in rule compositions} A fundamentally new aspect of rule algebra theory as compared to conventional categorical rewriting theory is centered upon the intuitively evident observation that rewriting rules may \emph{interact} with each other within \emph{sequences} of direct derivations. Concretely, as depicted in Figure~\ref{fig:overview-rc}, one may in a certain sense ``classify'' the interaction of two consecutively applied rewriting rules via a \emph{partial overlap} of the input object $I_2$ of the second with the output object $O_1$ of the first rule (encoded as a span of monomorphisms ${\color{h1color}\mu}=(I_2{\color{h1color}\hookleftarrow M_{21}\hookrightarrow }O_1)$). Reversing the argument, one may determine certain technical conditions (which depend on the type $\bT$ of the rewriting semantics) under which a given partial overlap ${\color{h1color}\mu}$ is causally possible (then denoted ${\color{h1color}\mu}\in \MatchGT{\bT}{r_2}{r_1}$), so that one may compute the \emph{($\bT$-type) composite rule} (denoted $\compGT{\bT }{r_2}{{\color{h1color}\mu}}{r_1}$) along the overlap. Since for two given rewriting rules there may in general be many choices for admissible partial overlaps ${\color{h1color}\mu}$ possible, one may once again rely upon the idea of ``book-keeping'' these possible choices in the form of a ``sum over all possible compositions''. The concrete choice of semantics for this operation realized in the rule algebra framework is depicted in Figure~\ref{fig:overview-rap}: an $\bR$-vector space with basis vectors $\delta(r)$ indexed by isomorphism classes of rules is introduced, upon which the so-called \emph{($\bT$-type) rule algebra product} $\rap{\bT}{\delta(r_2)}{\delta(r_1)}$ of two basis vectors $\delta(r_2)$ and $\delta(r_1)$ is defined via a sum over all admissible compositions of the two rules (considered up to isomorphisms). Endowing an $\bR$-vector space with a bi-linear binary operation such as $\rap{\bT}{.}{.}$ yields what is referred to in the general mathematics literature as an algebra, hence the moniker \emph{rule algebras}.

\paragraph{The algorithmic essence of rule algebra theory} Under certain technical assumptions upon the base categories over which the rewriting theories are defined (cf.\ Section~\ref{ssc:crtwc} for further details), the aforementioned rule algebra products $\rap{\bT}{}{}$ and the linear operators $\rho^{\bT}(\delta(r))$ as depicted in Figure~\ref{fig:overview-rar} are guaranteed to satisfy certain mathematical properties that are quintessential in view of algorithmic developments. The property most important in view of the formal definition and static analysis of stochastic rewriting systems is depicted in Figure~\ref{fig:overview-rp}: for any computation that requires computing all possible two-step derivation sequences along rule $r_2$ after rule $r_1$, the \emph{representation property} entails that one may instead first compute all possible rule compositions of $r_2$ with $r_1$ (encoded in the rule algebra product $\rap{\bT}{\delta(r_2)}{\delta(r_1)}$), followed by determining the ways in which the constituent composite rules may be applied. The crucial practical value of the latter type of computation consists in the type of static analysis of rewriting-based continuous-time Markov chains as introduced in Theorem~\ref{thm:CTMCmev}, which relies upon so-called \emph{commutators}; letting $[\delta(r_2),\delta(r_1)]_{\rap{\bT}{}{}}:=\rap{\bT}{\delta(r_2)}{\delta(r_1)}-\rap{\bT}{\delta(r_1)}{\delta(r_2)}$, such a commutator computes the \emph{difference} in all possible ways of composing rule $r_2$ with $r_1$ minus the composition of $r_1$ with $r_2$. The true computational gain achieved by the rule algebra framework consists in the fact that in practice often only very few terms remain in computing commutators, encoding essentially those compositions only possible in a given order. We refer the interested readers to Example~\ref{ex:HWex} for a concise illustration of this phenomenon for the case of the rule algebra arising from rewriting systems over vertex-only graphs.\\

\paragraph{Rewriting theory in the life sciences} A major obstacle for applying rewriting-based and in particular rule-algebraic techniques to the modeling of complex bio- and organo-chemical reaction systems is the intricate nature of the encodings of data structures and rewriting semantics in these theories. Concretely, while \emph{molecules} in these settings are formalizable as typed (undirected) graphs, and with reactions thus formulated as rewriting rules of such graphs, the concrete encoding requires certain \emph{structural constraints} on the graphs; the preservation of such constraints in turn entails that rewriting rules must be endowed with so-called \emph{application conditions}. Referring to Section~\ref{sec:cond} for the precise details, a formalism of constraints and conditions had been available in the categorical rewriting literature since the pioneering work of Habel and Pennemann~\cite{habel2009correctness}, yet it was only recently demonstrated in~\cite{behrRaSiR} that (under certain technical conditions) this extended categorical rewriting framework possesses the \emph{compositionality properties} that are necessary in order to formulate suitable rule algebra frameworks. In this paper, we finally assemble a complete rule algebra framework for both DPO- and SqPO-type semantics and for fully general rewriting rules with conditions. Our novel theoretical contributions are organized into two main parts:
\begin{itemize}
\item Sections~\ref{sec:ra} and~\ref{sec:stochMech} contain a complete account of the general rule algebra theory, including a theory of continuous-time Markov chains for stochastic rewriting systems based upon rewriting rules with conditions (and possibly constraints on objects). 
\item Section~\ref{sec:rrt} contains a (technically rather intricate) refinement of the general theory to the special cases of rewriting theories in which the conditions on rules arise from the requirement of preserving (globally fixed) constraints on objects. Such types of \emph{restricted rewriting theories} are demonstrated to admit algorithmic implementations of the rule-algebraic operations that are considerably more tractable than in the general setting, and that are in fact the relevant setting in particular for the aforementioned modeling applications in the life sciences.
\end{itemize}

Finally, the second part of this paper is devoted to applications of our novel rule algebra formalism in its restricted rewriting theory variant to biochemistry (Section~\ref{sec:bcgr}) and to organo-chemistry (Section~\ref{sec:ocgr}). To the best of our knowledge, this is the first-of-its-kind rewriting-theoretical formalization of the syntactic definitions of the relevant \KAP{}~\cite{Boutillier:2018aa} and \MOD{}~\cite{Andersen_2016} frameworks.

\section{Compositional rewriting theories with conditions}\label{ssc:crtwc}

The well-established \emph{Double-Pushout (DPO)}~\cite{ehrig:2006fund} and \emph{Sesqui-Pushout (SqPO)}~\cite{Corradini_2006} frameworks for rewriting systems over categories with suitable adhesivity properties~\cite{lack2005adhesive,ehrig2004adhesive,GABRIEL_2014,ehrig2014mathcal} provide a principled and very general  foundation for rewriting theories. %
However, in practice many applications require the rewriting of objects that may not be interpreted directly as the objects of some adhesive category, but which instead may be obtained from a suitable ``ambient'' category via the notion of \emph{constraints} on objects. %
Together with a corresponding notion of \emph{constraint-preserving application conditions} on rewriting rules, this approach yields a versatile extension of rewriting theory. In the DPO setting, this modification had been well-known~\cite{habel2009correctness,ehrig:2006fund,ehrig2014mathcal,ehrig2012m}, while it has been only very recently introduced for the SqPO setting~\cite{behrRaSiR}. For the \emph{rule algebra} constructions presented in the main part of this contribution, we require in addition a certain \emph{compositionality} property of our rewriting theories (established for the DPO case in~\cite{bp2018,bp2019-ext}, for the SqPO case in~\cite{nbSqPO2019}, and for both settings augmented with conditions in~\cite{behrRaSiR}). 

\subsection{Category-theoretical prerequisites}

We collect in~\ref{sec:MACapp} some of the salient concepts on $\cM$-adhesive categories and the relevant notational conventions. Throughout this paper, we will make the following assumptions:
\begin{assumption}\label{as:main}
    $\bfC\equiv(\bfC,\cM)$ is a finitary $\cM$-adhesive category with $\cM$-initial object, $\cM$-effective unions and epi-$\cM$-factorization. In the setting of \emph{Sesqui-Pushout (SqPO) rewriting}, we assume in addition that all final pullback complements (FPCs) along composable pairs of $\cM$-morphisms exist, and that $\cM$-morphisms are stable under FPCs. For $\bT\in \{DPO,SqPO\}$, we use the convenient \emph{shorthand notation} $\bfC\in\cM-\mathbf{CAT}_{\bT}$ as a shorthand for $\bfC\equiv(\bfC,\cM)$ satisfying the version of the assumption relevant to type-$\bT$ rewriting.
\end{assumption}

Many applications of practical interest, including both of the main application examples presented within this paper, are formulated in terms of (typed variants of) undirected multigraphs, giving rise to one of the principal examples of a category which is $\cM$-adhesive, but \emph{not} adhesive. It is precisely this latter fact which emphasizes the need of the modern standard formulation of rewriting theory in terms of the level of generality offered by the framework of $\cM$-adhesive categories.

\begin{definition}\label{def:ugraph}
    Let $\cP^{(1,2)}:\mathbf{Set}\rightarrow \mathbf{Set}$ be the restricted powerset functor (mapping a set $S$ to the set of its subsets $P\subset S$ with $1\leq |P|\leq 2$). The category $\mathbf{uGraph}$~\cite{bp2019-ext} of \emph{finite undirected multigraphs} is defined as the finitary restriction of the comma category $(ID_{\mathbf{Set}},\cP^{(1,2)})$. Thus an undirected multigraph is specified as a triple of data $G=(E_G,V_G,i_G)$, where $E_G$ and $V_G$ are (finite) sets of edges and vertices, respectively, and where $i_G:E_G\rightarrow \cP^{(1,2)}(V_G)$ is the edge-incidence map.
\end{definition}

\begin{theorem}
    $\mathbf{uGraph}$ satisfies Assumption~\ref{as:main}, both for the DPO- and for the extended SqPO-variant.
\end{theorem}
\begin{proof}
	As demonstrated in~\cite{bp2019-ext}, $\mathbf{uGraph}$ is indeed a finitary $\cM$-adhesive category with $\cM$-initial object and $\cM$-effective unions, for $\cM$ the class of component-wise monic $\mathbf{uGraph}$-morphisms. It thus remains to prove the existence of an epi-$\cM$-factorization as well as the properties related to FPCs. To this end, utilizing the fact that the category $\mathbf{Set}$ upon which the comma category $\mathbf{uGraph}$ is based possesses an epi-mono-factorization, we may construct the following diagram from a $\mathbf{uGraph}$-morphism $\varphi=(\varphi_E,\cP^{(1,2)}(\varphi_V))$ (for component morphisms $\varphi_E:E\rightarrow E'$ and $\varphi_V:V\rightarrow V'$):
\begin{equation}
    \ti{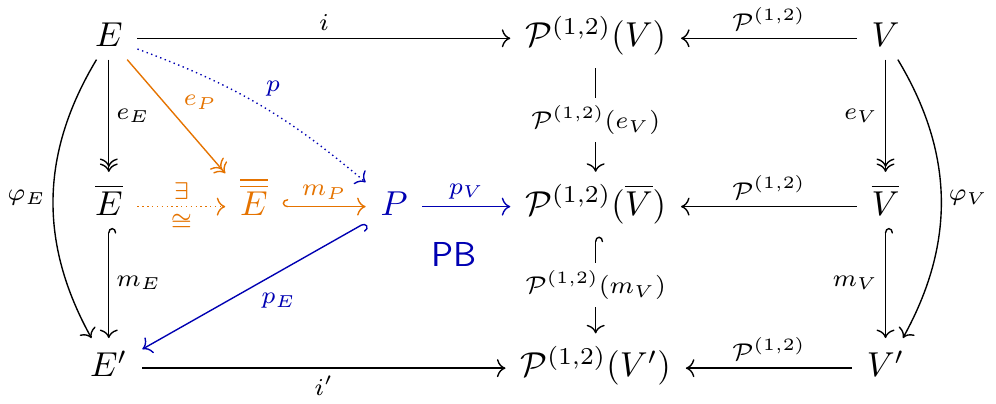}
\end{equation}
The diagram is constructed as follows:
\begin{enumerate}
\item Perform the epi-mono-factorizations $\varphi_E=m_E\circ e_E$ and $\varphi_V=m_V\circ e_V$, and apply the functor $\cP^{(1,2)}$ in order to obtain the morphisms $\cP^{(1,2)}(e_V)$ and $\cP^{(1,2)}(m_V)$; since  the functor $\cP^{(1,2)}$ preserves monomorphisms~\cite{padberg2017towards}, $\cP^{(1,2)}(m_V)\in \mono{\mathbf{Set}}$.
\item Construct the pullback 
\[(E'{\color{h1color}\leftarrow P\rightarrow}\cP^{(1,2)}(\overline{V})):=\pB{E'\rightarrow \cP^{(1,2)}(V')\leftarrow \cP^{(1,2)}(\overline{V})}\,,
\]
Since monomorphisms are stable under pullback in $\mathbf{Set}$, having proved that $\cP^{(1,2)}(m_V)\in \mono{\mathbf{Set}}$ implies ${\color{h1color}(p_E:P\rightarrow E')}\in \mono{\mathbf{Set}}$.
\item By the universal property of pullbacks, there exists a morphism ${\color{h1color}(p:E\rightarrow P)}$. Let $p={\color{h2color}m_P\circ e_P}$ be the epi-mono-factorization of this morphism. 
\item By stability of monomorphisms under composition in $\mathbf{Set}$, we find that ${\color{h1color}p_E}\circ {\color{h2color}m_P}\in \mono{\mathbf{Set}}$, and consequently $\varphi_E=({\color{h1color}p_E}\circ {\color{h2color}m_P})\circ {\color{h2color}e_P}$ yields an alternative epi-mono-factorization of $\varphi_E$. Then by uniqueness of epi-mono-factorizations up to isomorphism, there must exist an isomorphism $(\overline{E}{\color{h2color}\rightarrow \overline{\overline{E}}})\in \iso{\mathbf{Set}}$.
\end{enumerate}
We have thus demonstrated that both $(e_E,\cP^{(1,2)}(e_V))$ and $(m_E,\cP^{(1,2)}(m_V))$ are morphisms in $\mathbf{uGraph}$. Since morphisms in comma categories are mono-, epi- or iso-morphisms if they are so component-wise~\cite{ehrig:2006fund}, we conclude that 
\[
(e_E,\cP^{(1,2)}(e_V))\in\epi{\mathbf{uGraph}}\,,\quad (m_E,\cP^{(1,2)}(m_V))\in\mono{\mathbf{uGraph}}\,,
\]
which finally entails that we have explicitly constructed an epi-mono-factorization of the $\mathbf{uGraph}$-morphism $(\varphi_E,\cP^{(1,2)}(\varphi_V))$. 

In order to demonstrate that FPCs along pairs of composable $\cM$-morphisms $\varphi_A,\varphi_B\in\cM$ in $\mathbf{uGraph}$ exist (for $\cM$ the class of component-wise monomomophic $\mathbf{uGraph}$ morphisms), we provide the following explicit construction:
\begin{equation}
\begin{array}{c|c}
\ti{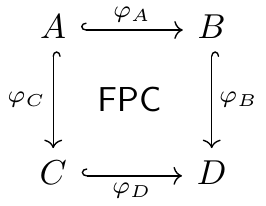}\hphantom{X} & \hphantom{X}
    \begin{aligned}
        V_C&=V_D\setminus(V_B\setminus V_A)\\
        E_C&=\{e\in E_D\setminus (E_B\setminus E_A)\mid
        u_D(e)\in\cP^{(1,2)}(V_C)\}\\
        u_C&=u_D\vert_{E_C}\\
        \varphi_C&=(E_A\hookrightarrow E_C,\cP^{(1,2)}(V_A\hookrightarrow V_C))\\
        \varphi_D&=(E_C\hookrightarrow E_D,\cP^{(1,2)}(V_C\hookrightarrow V_D))
    \end{aligned}
    \end{array}
\end{equation}
\end{proof}

\subsection{Conditions}\label{sec:cond}

\begin{definition}
    \emph{Conditions}\cite{habel2009correctness,ehrig2014mathcal} in an $\cM$-adhesive category $(\bfC,\cM)$ satisfying Assumption~\ref{as:main} are recursively defined for every object $X\in \obj{\bfC}$ as follows:
    \begin{enumerate}
        \item $\ac{true}_X$ is a condition.
        \item Given $(f:X\hookrightarrow Y)\in \cM$ and a condition $\ac{c}_Y$, $\exists(f,\ac{c}_Y)$ is a condition.
        \item If $\ac{c}_X$ is a condition, so is $\neg \ac{c}_X$.
        \item If $\ac{c}_X^{(1)},\ac{c}_X^{(2)}$ are conditions, so is $\ac{c}_X^{(1)}\land \ac{c}_X^{(2)}$.
    \end{enumerate}
    The \emph{satisfaction} of a condition $\ac{c}_X$ by a $\cM$-morphism $(h:X\hookrightarrow Z)\in \cM$, denoted $h\vDash \ac{c}_X$, is recursively defined (with notations as above) as follows:
    \begin{enumerate}
        \item $h\vDash\ac{true}_X$.
        \item $h\vDash \exists(f,\ac{c}_Y)$ iff there exists an $\cM$-morphism $(g:Y\hookrightarrow Z)\in \cM$ such that $h=g\circ f$ and $g\vDash Y$.
        \item $h\vDash \neg\ac{c}_X$ iff $h\not{\vDash} \ac{c}_X$.
        \item $h\vDash (\ac{c}_X^{(1)}\land \ac{c}_X^{(2)})$ iff $h\vDash \ac{c}_X^{(1)}$ and $h\vDash \ac{c}_X^{(2)}$.
    \end{enumerate}
   \begin{equation}
\ti{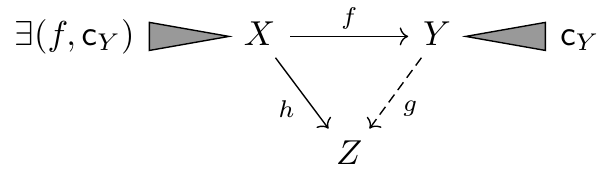}
    \end{equation}
    Two conditions $\ac{c}_X$ and $\ac{c}_X'$ are \emph{equivalent}, denoted $\ac{c}_X\equiv \ac{c}_X'$, iff for every $\cM$-morphism $(h:X\hookrightarrow Z)\in \cM$, $h\vDash \ac{c}_X$ if and only if $h\vDash\ac{c}_X'$.
    \medskip

    Finally, a condition $\ac{c}_{\mIO}$ over the $\cM$-initial object $\mIO$ is called a \emph{constraint}, and we define for every object $Z\in \obj{\bfC}$
    \begin{equation}
        Z\vDash \ac{c}_{\mIO} \quad :\Leftrightarrow \quad (\mIO\hookrightarrow Z)\vDash \ac{c}_{\mIO}\,.
    \end{equation} 
\end{definition}

We will utilize as a \textbf{notational convention} the standard shorthand notations
\begin{equation}
    \exists(X\hookrightarrow Y):=\exists(X\hookrightarrow Y,\ac{true}_Y)\,,\quad
    \forall(X\hookrightarrow Y,\ac{c}_Y):=\neg\exists(X\hookrightarrow Y,\neg\ac{c}_Y)\,.
\end{equation}
It is conventional to refer to a condition as a \emph{constraint} if it is formulated over the $\cM$-initial object $\mIO$. As for arbitrary objects $X\in \obj{\bfC}$ by definition of $\cM$-initiality there exists precisely one morphism $(\mIO\hookrightarrow X)\in \cM$ from the $\cM$-initial object $\mIO$, it is customary to employ the additional simplification of notations
\begin{equation}\label{eq:constrNotation}
\exists(X,\ac{c}_X) := \exists(\mIO\hookrightarrow X,\ac{c}_X)\,.
\end{equation}
For example, the constraints
\begin{equation*}
    \ac{c}_{\mIO}^{(1)}=\exists(\ti{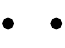})\,,\quad \ac{c}_{\mIO}^{(2)}=\not \exists(\ti{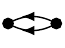})\,,\quad
\ac{c}_{\mIO}^{(3)}=\forall(\ti{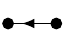},\exists(
\ti{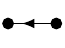}\hookrightarrow 
\ti{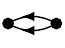}))
\end{equation*}
express for a given object $Z\in\obj{\bfC}$ that $Z$ contains at least two vertices (if $Z\vDash\ac{c}_{\mIO}^{(1)}$), that $Z$ does not contain parallel pairs of directed edges (if $Z\vDash\ac{c}_{\mIO}^{(2)}$), and that for every directed edge in $Z$ there also exists a directed edge between the same endpoints with opposite direction (if $Z\vDash\ac{c}_{\mIO}^{(3)}$), respectively.\\

In practical applications of the calculus of constraints and application conditions, it is of key importance to be able to extend conditions into larger contexts (via the so-called $\Shift$ operation), and to ``transport'' conditions across rules (via the so-called $\Trans$ operation). We refer the interested readers to Definition~\ref{def:Rcomp} for the key application of the two operations (i.e., in defining the operation of sequential composition of rewriting rules with conditions), as well as to Example~\ref{ex:sgraph-cpc} for a detailed illustration of the uses of the two operations in computations of application conditions for rewriting rules.
\begin{definition}[\cite{habel2009correctness,behrRaSiR}]\label{def:shiftTrans}
In an $\cM$-adhesive category satisfying the version of Assumption~\ref{as:main} appropriating for the type $\bT\in\{DPO,SqPO\}$ of rewriting, we define the \emph{shift operation}, denoted $\Shift$, and the \emph{transport operation}, denoted $\Trans$, via their effect on conditions:
    \begin{itemize}
    \item For all conditions $\ac{c}_X$ and for all $\cM$-morphisms $(f:X\hookrightarrow Y)\in \cM$, $(g:Y\hookrightarrow Z)\in \cM$ and $(h:X\hookrightarrow Z)\in \cM$ with $h=g\circ f$,
    \begin{equation}
        h\vDash \ac{c}_X\quad \Leftrightarrow \quad g\vDash \Shift(f,\ac{c}_X)\,.
    \end{equation}
    \item For all linear rules $r=(O\leftharpoonup I)\in \Lin{\bfC}$, for all conditions $\ac{c}_O\in \cond{\bfC}$ and for all $\bT$-admissible matches $(m:X\hookleftarrow I)\in \MatchGT{\bT}{r}{X}$,
    \begin{equation}
    	(r_m(X)\hookleftarrow O)\vDash \ac{c}_O\quad  \Leftrightarrow \quad m\vDash \Trans(r,\ac{c}_O)\,. \label{eq:TransProps}
    \end{equation}
\end{itemize}
\end{definition}

Crucially, there exist algorithmic implementations for both of these constructions in DPO-rewriting~\cite{habel2009correctness} as well as in SqPO-rewriting~\cite{behrRaSiR}:

\begin{theorem}[\cite{habel2009correctness,behrRaSiR}]\label{thm:STdefns}
With notations as in Definition~\ref{def:shiftTrans}, the action of the operation $\Shift$ on an application condition $\ac{c}_X\in \cond{\bfC}$ over some object $X$ along an $\cM$-morphism $(y:X\hookrightarrow Y)$ is implemented concretely via the following recursive algorithm:
\begin{itemize}
\begin{subequations}
\item \emph{Case $\ac{c}_X=\ac{true}$:} 
\begin{equation}\label{eq:ShiftAlgo-triv}
\Shift(f,\ac{true}):=\ac{true}
\end{equation}
\item \emph{Case $\ac{c}_X=\exists(a:X\hookrightarrow A,\ac{c}_A)$ (for $a\in \cM$):} 
\begin{equation}\label{eq:ShiftAlgo-iter}
\vcenter{\hbox{\ti{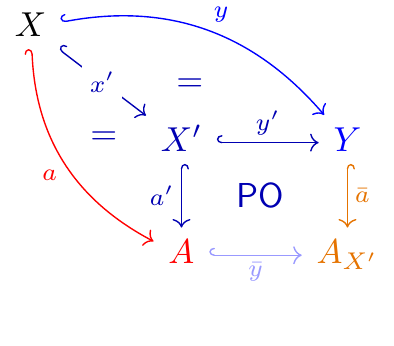}}}\qquad 
\vcenter{\hbox{$\begin{array}{rl}
\Shift({\color{blue}y}, \exists ({\color{red} a, \ac{c}_A}))
&:= 
\bigwedge\limits_{\substack{{\color{h1color}(a',x',y')}\in \cM^{\times\:3}\\ 
{\color{h1color}a'}\circ {\color{h1color}x'}= {\color{red}a}
\land {\color{h1color}y'}\circ {\color{h1color}x'}= {\color{blue}y}}}
 \exists({\color{h2color}\bar{a}}, 
 \Shift({\color{h3color}\bar{y}},{\color{red}\ac{c}_A}))\\
&\quad\text{with }{\color{h2color}A_{X'}}:=
 \pO{{\color{red}A}{\color{h1color}\xhookleftarrow{a'}X'\xhookrightarrow{y'}}{\color{blue}Y}}
\end{array}$}}
\end{equation}
\item \emph{Case $\ac{c}_X = \ac{c}_X^{(1)}\land \ac{c}_X^{(2)}$:} 
\begin{equation}\label{eq:ShiftAlgo-lor}
\Shift(y,\ac{c}_X^{(1)}\land \ac{c}_X^{(2)}) := \Shift(y,\ac{c}_X^{(1)})\land \Shift(y,\ac{c}_X^{(2)})
\end{equation}
\item \emph{Case $\ac{c}_X = \neg\ac{c}_X'$:} 
\begin{equation}\label{eq:ShiftAlgo-neg}
\Shift(y,\neg\ac{c}_X') := \neg \Shift(y,\ac{c}_X')
\end{equation}
\end{subequations}
\end{itemize}
Moreover, the action of the $\Trans$ operation on a condition $\ac{c}_O\in \cond{\bfC}$ along a linear rule $r=(O\hookleftarrow K\hookrightarrow I)\in \Lin{\bfC}$ is algorithmically implementable as follows:
\begin{itemize}
\begin{subequations}
\item \emph{Case $\ac{c}_O=\ac{true}$:} 
\begin{equation}\label{eq:TransAlgo-triv}
\Trans(r,\ac{true}):=\ac{true}
\end{equation}
\item \emph{Case $\ac{c}_O=\exists(a:O\hookrightarrow A,\ac{c}_A)$ (for $a\in \cM$):} 
\begin{equation}\label{eq:TransAlgo-iter}
\begin{gathered}
\vcenter{\hbox{\ti{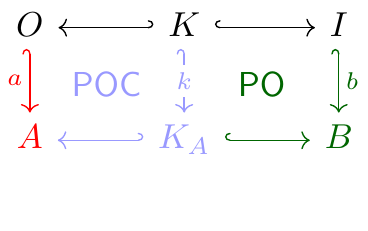}}}\\ 
\vcenter{\hbox{$
\Trans({\color{blue}r}, \exists ({\color{red} a, \ac{c}_A}))
:= \begin{cases}
\ac{false} \quad &\text{if ${\color{h3color}\mathsf{POC}}$ does not exist}\\
\exists({\color{h4color}b},\Trans({\color{red}A}{\color{h3color}\hookleftarrow K_A}{\color{h4color}\hookrightarrow B},{\color{red}\ac{c}_A})) &\text{otherwise.}
\end{cases}
$}}
\end{gathered}
\end{equation}
\item \emph{Case $\ac{c}_O = \ac{c}_O^{(1)}\land \ac{c}_O^{(2)}$:} 
\begin{equation}\label{eq:TransAlgo-lor}
\Trans(r,\ac{c}_O^{(1)}\land \ac{c}_O^{(2)}) := \Trans(r,\ac{c}_O^{(1)})\land \Trans(r,\ac{c}_O^{(2)})
\end{equation}
\item \emph{Case $\ac{c}_O = \neg\ac{c}_O'$:} 
\begin{equation}\label{eq:TransAlgo-neg}
\Trans(r,\neg\ac{c}_O') := \neg \Trans(r,\ac{c}_O')
\end{equation}
\end{subequations}
\end{itemize}
\end{theorem}

Finally, we will take advantage of the following results as part of the rule-algebraic calculus for rules with conditions (specifically when analyzing the properties of sequential rule compositions):
\begin{theorem}\label{thm:STcompComp}
For an $\cM$-adhesive category satisfying the suitable version of Assumption~\ref{as:main} for $\bT$-type rewriting (with $\bT\in \{DPO,SqPO\}$), the following \emph{compositionality} and \emph{compatibility} properties for $\Shift$ and $\Trans$ hold:
\begin{itemize}
\item \emph{Compositionality of $\Shift$:} $\forall (\beta:C\hookleftarrow B), (\alpha:B\hookleftarrow A)\in \cM, \ac{c}_A\in \cond{\bfC}$, 
\begin{equation}
\Shift(\beta,\Shift(\alpha,\ac{c}_A)) \equiv \Shift(\beta\circ\alpha,\ac{c}_A)\,. \label{eq:compShift}
\end{equation}
\item \emph{Compositionality of $\Trans$:} $\forall r_2=(C\leftharpoonup B), r_1=(B\leftharpoonup A)\in \Lin{\bfC}, \ac{c}_C\in \cond{\bfC}$,
\begin{equation}
\Trans(r_1, \Trans(r_2, \ac{c}_E)) \dot{\equiv} \Trans( r_2\circ r_1, \ac{c}_C)\,.
\label{eq:compTrans}
\end{equation}
\item \emph{Compatibility of $\Shift$ and $\Trans$:} $\forall r=(O\leftharpoonup I)\in \Lin{\bfC}\,,\; X\in \obj{\bfC}, m\in \MatchGT{\bT}{r}{X}, \ac{c}_O\in \cond{\bfC}$,
\begin{equation}
\Shift(m, \Trans(r,\ac{c}_O))\dot{\equiv} \Trans(r_m(X)\leftharpoonup X, \Shift(r_m(X)\hookleftarrow O,\ac{c}_O)\,.
\label{eq:ShiftTransCompatibility}
\end{equation}
\end{itemize}
Here we used the notation $\dot{\equiv}$ for \emph{equivalence of conditions modulo admissibility} (i.e., equivalence is only required to hold for admissible matches of the rule $r$). 
\end{theorem}

\subsection{Compositional rewriting with conditions}\label{sec:crc}

Throughout this section, we assume that we are given a type $\bT\in\{DPO,SqPO\}$ of rewriting semantics and an $\cM$-adhesive category $\bfC$ satisfying the respective variant of Assumption~\ref{as:main} (i.e., $\bfC\in \cM-\mathbf{CAT}_{\bT}$). In categorical rewriting theories, the universal constructions utilized such as pushouts, pullbacks, pushout complements and final pullback complements are unique only up to universal isomorphisms. This motivates specifying a suitable notion of equivalence classes of rules with conditions:

\begin{definition}[Rules with conditions]\label{def:rwc}
    Let $\LinAc{\bfC}$ denote the class of \emph{(linear) rules with conditions}, defined as\footnote{Note that our ``input-to-output'', i.e., ``right-to-left'' notational convention for rules differs from the traditional ``left-to-right'' convention; this is deeply motivated from the theory of rule algebras, as explained in detail in Remark~\ref{rem:rNC}.}
    \begin{equation}
        \LinAc{\bfC}:=\{(O\xleftarrow{o}K\xrightarrow{i}I,\ac{c}_I)\mid o,i\in\cM,\; \ac{c}_I\in\cond{\bfC}\}\,.
    \end{equation}
{\makeatletter
\let\par\@@par
\par\parshape0
\everypar{}
\begin{wrapfigure}[4]{r}{0.29\linewidth}
\vspace{-1.9em}
  \begin{equation}
       \!\!\!\!\!\!\!\!\!\!\!\!\ti{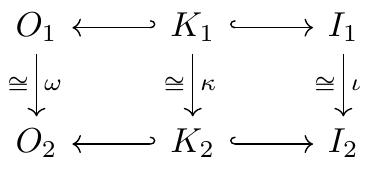}
  \end{equation}
\end{wrapfigure}
\noindent We define two rules with conditions $R_j=(r_j,\ac{c}_{I_j})$ ($j=1,2$) \emph{equivalent}, denoted $R_2\sim R_1$, iff $\ac{c}_{I_1}\equiv\ac{c}_{I_2}$ and if there exist isomorphisms $\omega,\kappa,\iota\in\iso{\bfC}$ such that the diagram on the right commutes. We denote by $\LinEq{\bfC}$ the set of equivalence classes under $\sim$ of rules with conditions.\par}
\end{definition}

\begin{definition}[Direct derivations]\label{def:dd}
    Let $r=(O\hookleftarrow K\hookrightarrow I)\in\Lin{\bfC}$ and $\ac{c}_I\in\cond{\bfC}$ be concrete representatives of some equivalence class $R\in\LinEq{\bfC}$, and let $X,Y\in\obj{\bfC}$ be objects. Then a \emph{type $\bT$ direct derivation} is defined as a commutative diagram such as below right, where all morphism are in $\cM$ (and with the left representation a shorthand notation)
\begin{equation}\label{eq:DD}
\ti{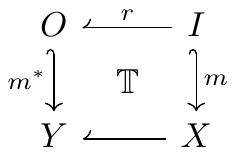}\quad :=\quad 
\ti{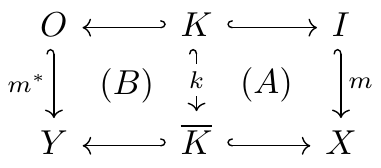}\,.
\end{equation}
with the following pieces of information required relative to the type:
\begin{enumerate}
    \item $\mathbf{\bT=DPO}$: given $(m:I\hookrightarrow X)\in\cM$, $m$ is a \emph{DPO-admissible match of $R$ into $X$}, denoted $m\in\MatchGT{DPO}{R}{X}$, if $m\vDash \ac{c}_I$ and $(A)$ is constructable as a \emph{pushout complement}, in which case $(B)$ is constructed as a \emph{pushout}.
    \item $\mathbf{\bT=SqPO}$: given $(m:I\hookrightarrow X)\in\cM$, $m$ is a \emph{SqPO-admissible match of $R$ into $X$}, denoted $m\in\MatchGT{SqPO}{R}{X}$, if $m\vDash\ac{c}_I$, in which case $(A)$ is constructed as a \emph{final pullback complement} and $(B)$ as a \emph{pushout}.
    \item $\mathbf{\bT=DPO^{\dag}}$: given just the ``plain rule'' $r$ and $(m^{*}:O\hookrightarrow Y)\in\cM$, $m^{*}$ is a \emph{DPO${}^{\dag}$-admissible match of $r$ into $X$}, denoted $m\in\MatchGT{DPO^{\dag}}{r}{Y}$, if $(B)$ is constructable as a \emph{pushout complement}, in which case $(B)$ is constructed as a \emph{pushout}.
\end{enumerate}
For types $\bT\in\{DPO,SqPO\}$, we will sometimes employ the notation $R_m(X)$ for the object $Y$.
\end{definition}
Note that at this point, we have not yet resolved a conceptual issue that arises from the non-uniqueness of a direct derivation given a rule and an admissible match. Concretely, the pushout complement, pushout and FPC constructions are only unique up to isomorphisms. This issue will ultimately be resolved as part of the rule algebraic theory. We next consider a certain \emph{composition operation} on rules with conditions that is quintessential to our main constructions:

\begin{definition}[Rule compositions]\label{def:Rcomp}
 Let $R_1,R_2\in \LinEq{\bfC}$ be two equivalence classes of rules with conditions, and let $(r_j,\ac{c}_{I_j})\in \LinAc{\bfC}$ (for $r_j\in\Lin{\bfC}$ and $\ac{c}_{I_j}\in \cond{\bfC}$) be concrete representatives of $R_j$ (for $j=1,2$). For $\bT\in\{DPO,SqPO\}$, an $\cM$-span $\mu=(I_2\hookleftarrow M_{21}\hookrightarrow O_1)$ (i.e., with $(M_{21}\hookrightarrow O_1),(M_{21}\hookrightarrow I_2)\in\cM$) is a \emph{$\bT$-admissible match of $R_2$ into $R_1$}, denoted $\mu\in\RMatchGT{\bT}{R_2}{R_1}$, if the diagram below is constructable (with $N_{21}$ constructed by taking pushout)
 \begin{equation}\label{eq:defRcomp}
 \ti{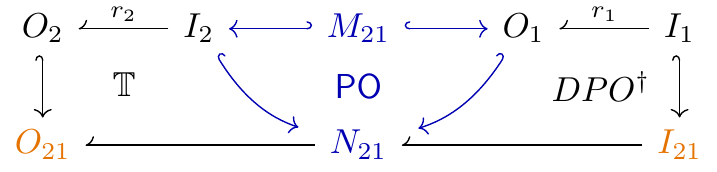}
\end{equation}
and if $\ac{c}_{I_{21}}\not{\!\!\dot{\equiv}}\,\,\ac{false}$. Here, the condition $\ac{c}_{I_{21}}$ is computed as
\begin{equation}\label{eq:acRcomp}
    \ac{c}_{I_{21}}:=\Shift(I_1\hookrightarrow I_{21},\ac{c}_{I_1})\;\land\; \Trans(N_{21}\leftharpoonup I_{21},\Shift(I_2\hookrightarrow N_{21},\ac{c}_{I_2}))\,.
\end{equation}
In this case, we define the \emph{type $\bT$ composition of $R_2$ with $R_1$ along $\mu$}, denoted $\compGT{\bT}{R_2}{\mu}{R_1}$, as
\begin{equation}
\compGT{\bT}{R_2}{\mu}{R_1}:=[(O_{21}\leftharpoonup I_{21},\ac{c}_{I_{21}})]_{\sim}\,,
\end{equation}
where $(O_{21}\leftharpoonup I_{21}):=(O_{21}\leftharpoonup N_{21})\circ(N_{21}\leftharpoonup I_{21})$ (with $\circ$ the \emph{span composition} operation).
\end{definition}

We recall in~\ref{app:ACthms} two important technical results on the notions of direct derivations and rule compositions that have been derived in~\cite{behrRaSiR} (where however the DPO-type concurrency theorem is of course classical, cf.\ e.g.\ \cite{ehrig:2006fund}).

\section{Rule algebras for compositional rewriting with conditions}\label{sec:ra}

The associativity property of rule compositions in both DPO- and SqPO-type semantics for rewriting with conditions as proved in~\cite{behrRaSiR} may be fruitfully exploited within rule algebra theory. One possibility to encode the non-determinism in sequential applications of rules to objects is given by lifting each possible configuration $X\in\obj{\bfC}_{\cong}$ (i.e., isomorphism class of \emph{objects}) to a basis vector $\ket{X}$ of a vector space $\hat{\bfC}$; then a rule $r$ is lifted to a linear operator acting on $\hat{\bfC}$, with the idea that this operator acting upon a basis vector $\ket{X}$ should evaluate to the ``sum over all possibilities to act with $r$ on $X$''. We will extend here the general rule algebra framework~\cite{bdg2016,bp2018,nbSqPO2019} to the present setting of rewriting rules with conditions.\\

We will first lift the notion of rule composition into the setting of a composition operation on a certain abstract vector space over rules, thus realizing the heuristic concept of ``summing over all possibilities to compose rules''.
\begin{definition}\label{def:RA}
    Let $\bT\in\{DPO,SqPO\}$ be the rewriting type, and let $\bfC$ be a category satisfying the relevant variant of Assumption~\ref{as:main}. %
For a field $\bK$ of characteristic $0$ (such as $\bK=\bR$ or $\bK=\bC$), let $\cR_{\bfC}$ be a $\bK$-vector space, defined via a bijection $\delta:\LinEq{\bfC}\xrightarrow{\cong}\mathsf{basis}(\cR_{\bfC})$ from the set of equivalence classes of linear rules with conditions to the set of basis vectors of $\cR_{\bfC}$. Let $\rap{\bT}{}{}$ denote the \emph{type $\bT$ rule algebra product}, a binary operation defined  via its action on basis elements $\delta(R_1),\delta(R_1)\in \cR_{\bfC}$ (for $R_1,R_2\in \LinEq{\bfC}$) as
    \begin{equation}
        \rap{\bT}{\delta(R_2)}{\delta(R_1)}:=\sum_{\mu\in \RMatchGT{\bT}{R_2}{R_1}}\delta\left(\compGT{\bT}{R_2}{\mu}{R_1}\right)\,.
    \end{equation}
    We refer to $\cR_{\bfC}^{\bT}:=(\cR_{\bfC},\rap{\bT}{}{})$ as the \emph{$\bT$-type rule algebra over $\bfC$}.
\end{definition}
\begin{theorem}\label{thm:RAmain}
    For type $\bT\in\{DPO,SqPO\}$ rewriting over a category $\bfC$ satisfying Assumption~\ref{as:main}, the rule algebra $\cR_{\bfC}^{\bT}$ is an \emph{associative unital algebra}, with unit element $\delta(R_{\mIO})$, where $R_{\mIO}:=(\mIO\hookleftarrow\mIO\hookrightarrow \mIO;\ac{true})$.
\end{theorem}
\begin{proof}
    \emph{Associativity} follows from Theorem~\ref{thm:assocR}, while \emph{unitality}, i.e., that 
    \[
      \forall R\in \LinEq{\bfC}:\quad  \rap{\bT}{\delta(R_{\mIO})}{\delta(R)}
      =\rap{\bT}{\delta(R)}{\delta(R_{\mIO})}=\delta(R)
    \]
    follows directly from an explicit computation of the relevant rule compositions.
\end{proof}

As alluded to in the introduction, the prototypical example of rule algebras are those of DPO- or (in this case by coincidence equivalently) of SqPO-type over discrete graphs, giving rise as a special case to the famous Heisenberg-Weyl algebra of key importance in mathematical chemistry, combinatorics and quantum physics (see~\cite{bp2019-ext} for further details). We will now illustrate the rule algebra concept in an example involving a more general base category.

\begin{example}
    For the category $\mathbf{uGraph}$ and DPO-type rewriting semantics, consider as an example the following two rules with conditions:
    \begin{equation}
        R_C:=\left(\ti{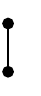}\hookleftarrow \ti{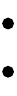}\hookrightarrow \ti{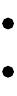}\,,\neg \exists \left(\ti{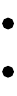}\hookrightarrow \ti{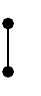}\right)\right)\,,\quad R_V:=(\ti{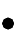}\hookleftarrow\mIO\hookrightarrow\ti{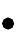}\,,\ac{true})\,.
    \end{equation}
The first rule is a typical example of a rule with application conditions, i.e., here stating that the rule may only link two vertices if they were previously not already linked to each other. The second rule, owing to DPO semantics, can in effect only be applied to vertices without any incident edges. The utility of the rule-algebraic composition operation then consists in reasoning about sequential compositions of these rules, for example (letting $*:=\rap{DPO}{}{}$):
\begin{equation}
\begin{aligned}
    \delta(R_C)*\delta(R_V)&=\delta(R_C\uplus R_V)+2\delta(R_C')\,,\; 
    R_C':=\left(\ti{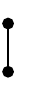}\hookleftarrow 
		 \ti{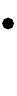}\hookrightarrow 
		 \ti{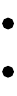}\,,\ac{true}\right)\\
    \delta(R_V)*\delta(R_C)&=\delta(R_C\uplus R_V)\,.
\end{aligned}
\end{equation}
To provide some intuition: the first computation encodes the causal information that the two rules may either be composed along a trivial overlap, or rule $R_C$ may overlap on one of the vertices in the output of $R_V$; in the latter case, any vertex to which first $R_V$ and then $R_C$ applies must not have had any incident edges, i.e., in particular no edge violating the constraint of $R_C$, which is why the composite rule $R_C'$ does not feature any non-trivial constraint. In the other order of composition, the two vertices in the output of $R_C$ are linked by an edge, so $R_V$ cannot be applied to any of these two vertices (leaving just the trivial overlap contribution).
\end{example}

Just as the rule algebra construction encodes the compositional associativity property of rule compositions, the following \emph{representation} construction encodes in a certain sense the properties described by the concurrency theorem:
\begin{definition}\label{def:RAR}
    Let $\bfC\in \cM-\mathbf{CAT}_{\bT}$ be an $\cM$-adhesive category suitable for type $\bT\in \{DPO,SqPO\}$ rewriting, %
    and denote by $\cR^{\bT}_{\bfC}$ the $\bT$-type rule algebra defined over a field $\bK$ of characteristic $0$. %
    Let $\hat{\bfC}$ be defined as the \emph{$\bK$-vector space} whose set of basis vectors is isomorphic to the set\footnote{We assume here that the isomorphism classes of objects of $\bfC$ form a \emph{set} (i.e., not a proper class).} of isomorphism classes of objects of $\bfC$ %
    via a bijection $\ket{.}:\obj{\bfC}_{\cong}\rightarrow \mathsf{basis}(\hat{\bfC})$. %
    Then the \emph{canonical representation of the $\bT$-type rule algebra $\cR_{\bfC}^{\bT}$ over $\bfC$}, denoted $\rho^{\bT}_{\bfC}$, %
    is defined as the morphism $\rho^{\bT}_{\bfC}:\cR_{\bfC}^{\bT}\rightarrow End_{\bK}(\hat{\bfC})$ specified via 
    \begin{equation}\label{eq:defRRA}
    \forall R\in\LinEq{\bfC},X\in\obj{\bfC}_{\cong}:
    \quad \canRep{\bT}{\delta(R)}\ket{X}:=\sum_{m\in\MatchGT{\bT}{R}{X}}\ket{R_m(X)}\,.
    \end{equation}
\end{definition}

\begin{theorem}\label{thm:canrep}
    $\rho^{\bT}_{\bfC}$ as defined above is an \emph{algebra homomorphism} (and thus in particular a well defined representation).
\end{theorem}
\begin{proof}
The proof is closely analogous to the one for the case without application conditions~\cite{bp2018,nbSqPO2019} (cf.\ \ref{app:proofCanrep}).
\end{proof}

\begin{remark}\label{rem:rNC}
It is worthwhile emphasizing that while our notational convention for linear rules (i.e., ``input-to-output'' as opposed to the traditional ``left-to-right'' convention) is non-standard, it is well-adapted to our setting of rule-algebraic computations. Concretely, by standard mathematical convention, \emph{matrix multiplication} and more generally the \emph{composition or linear operators} is performed as \emph{left composition}, hence in rule algebra theory, considering formulas such as the \emph{representation property} (cf.~\ref{app:proofCanrep})
\begin{equation}\label{eq:RRAproperty}
\rho^{\bT}_{\bfC}(\delta(R_2))\rho^{\bT}_{\bfC}(\delta(R_1))\ket{X}
=\rho^{\bT}_{\bfC}\left(\rap{\bT}{\delta(R_2)}{\delta(R_1)}\right)\ket{X}\,,
\end{equation}
it is natural to make explicit the ``direction'' or ordering of the rule applications (i.e., $R_2$ after $R_1$). As an additional benefit, the contributions to $\rap{\bT}{\delta(R_2)}{\delta(R_1)}$ arise from sequential rule composition diagrams as in~\eqref{eq:defRcomp}, which again due to its ``right-to-left'' notational convention is directly compatible with the encoded application order of the rules. In particular, this permits to directly interpret~\eqref{eq:RRAproperty} as an instantiation of the concurrency theorem (Theorem~\ref{thm:concur}), in the sense that the left hand side of~\eqref{eq:defRcomp} is evaluated as a sum over all two-step  derivation sequences starting from an object $X$, and with $R_2$ applied after $R_1$, while the right hand side of~\eqref{eq:defRcomp} evaluates to a sum over all one-step (i.e., direct) derivations along all possible sequential compositions of $R_2$ with $R_1$, with each composite rule applied in all possible ways to $X$.
\end{remark}

\begin{example}\label{ex:HWex}
In order to provide some high-level intuitions for the concept of representations for rule algebras, let us consider the particularly simple, yet in a certain sense paradigmatic example of a rule algebra and its representation, namely the famous \emph{Heisenberg-Weyl algebra} and its representation in the ``number vector'' basis (following~\cite[Sec.~3.1]{bp2019-ext}). In rewriting-theoretic terms, this amounts to considering rewriting systems over vertex-only graphs, and considering the linear rules of \emph{vertex creation} $r_{+}:=(\bullet \leftharpoonup \mIO)\equiv (\bullet\hookleftarrow \mIO\hookrightarrow \mIO)$ and of \emph{vertex deletion} $r_{-}:=(\mIO \leftharpoonup \bullet)\equiv (\mIO\hookleftarrow \mIO\hookrightarrow \bullet)$. %
Letting $\ket{n}:=\ket{\bullet^{\uplus\:n}}$ denote the basis vector encoding the isomorphism class of an $n$-vertex graph (for $n\geq0$), one may specialize the definition of the rule algebra representation as given in~\eqref{eq:defRRA} to the following formulas, where $R_{\pm}:=(r_{\pm},\ac{true})$, $\rho^{\bT}:=\rho^{\bT}_{\mathbf{uGraph}}$ for concreteness, and with $\bT\in \{DPO,SqPO\}$:
\begin{equation}
\rho^{\bT}(\delta(R_{+}))\ket{n}=\ket{n+1}\,,\qquad \rho^{\bT}(\delta(R_{-}))\ket{n} = \begin{cases}
0\quad &\text{if $n=0$}\\
n\ket{n-1} & \text{otherwise}
\end{cases}
\end{equation}
These equations thus encode the combinatorial facts that there is one way up to isomorphisms to add a vertex to a graph with $n$ vertices (obtaining $n+1$ vertices in total), while there is no way possible to delete a vertex from an empty graph, and precisely $n$ ways to delete a vertex from a graph with $n\geq 1$ vertices (each resulting in a graph with $n-1$ vertices). More importantly, the linear operators $\rho^{\bT}(\delta(R_{+}))$ and $\rho^{\bT}(\delta(R_{-}))$ in the basis of states $\ket{n}$ for $n\geq 0$ have precisely the same ``matrix elements'' as the linear operators $\tfrac{d}{dx}$ and $\hat{x}$ in the basis of monomials $x^n$ for $n\geq 0$; here, $x$ is a formal variable, $\tfrac{d}{dx}$ is the derivative operator, and $\hat{x}$ is the operator of multiplication with the formal variable (i.e., $\hat{x}x^n:= x^{n+1}$). In this sense, one may ``emulate'' or ``explain'' the calculus of operations on formal power series from within rewriting theory, with the following example of a \emph{commutation relation} amongst the key technical concepts:
\begin{equation}\label{eq:HWAccr}
[\rho^{\bT}(\delta(R_{-})),\rho^{\bT}(\delta(R_{+}))]:=\rho^{\bT}(\delta(R_{-}))\rho^{\bT}(\delta(R_{+})) - \rho^{\bT}(\delta(R_{+}))\rho^{\bT}(\delta(R_{-}))
\end{equation}
In order to evaluate the above equation, one could try to establish a formula for the ``matrix elements'' (i.e., by acting with the commutator on an arbitrary state $\ket{n}$), but one may also directly compute this expression by taking advantage of the \emph{representation property} enjoyed by $\rho^{\bT}$ (cf.\ Theorem~\ref{thm:canrep} and \ref{app:proofCanrep}):
\begin{equation}\label{eq:exHWA}
\begin{aligned}
&\rho^{\bT}(\delta(R_{-}))\rho^{\bT}(\delta(R_{+})) - \rho^{\bT}(\delta(R_{+}))\rho^{\bT}(\delta(R_{-})) = \rho^{\bT}\left(
\rap{\bT}{\delta(R_{-})}{\delta(R_{+})} - \rap{\bT}{\delta(R_{+})}{\delta(R_{-})}
\right)\\
&\qquad\qquad= \rho^{\bT}\left(
\delta(\bullet \hookleftarrow \mIO\hookrightarrow \bullet, \ac{true})
+ {\color{blue}\delta(\mIO \hookleftarrow \mIO\hookrightarrow \mIO, \ac{true})}
- \delta(\bullet \hookleftarrow \mIO\hookrightarrow \bullet, \ac{true})
\right)= \rho^{\bT}({\color{blue}\delta(R_{\mIO})})\,.
\end{aligned}
\end{equation}
Here, we have highlighted in {\color{blue}blue} the only contribution that ``survives'' the subtraction, which in this case amounts to $\rho^{\bT}({\color{blue}\delta(R_{\mIO})})$. This equation has a very intuitive rewriting-theoretic explanation: when first creating a vertex and then deleting a vertex, one may either delete a vertex that is different from the previously created vertex (yielding a composite rule $\bullet \hookleftarrow \mIO\hookrightarrow \bullet$), or instead one might delete precisely the vertex that was created in the first step (yielding in effect ${\color{blue}\mIO \hookleftarrow \mIO\hookrightarrow \mIO}$, i.e., the ``trivial'' rule). In the other order of operations, first deleting and then creating a vertex permits no causal interaction, i.e., the only possible composite rule in this order is $\bullet \hookleftarrow \mIO\hookrightarrow \bullet$. In this sense, \eqref{eq:exHWA} exhibits precisely the contribution only possible in one of the orders of applying the two rules, which is a typical feature of such \emph{commutation relations}, and which is at the heart of \emph{static analysis techniques} for rewriting-based continuous-time Markov chains according to Theorem~\ref{thm:CTMCmev}.\\

To provide a further illustration of the utility of rule-algebraic commutators, let us recall that we had defined $\rho^{\bT}:=\rho^{\bT}_{\mathbf{uGraph}}$, so that we may in particular consider acting with $\rho^{\bT}(\delta(R_{\pm}))$ on graph states other than those of the form $\ket{n}$. Note first that the DPO- and SqPO-semantics yield drastically different actions of $\rho^{\bT}(\delta(R_{-}))$ on generic graph states $\ket{G}$, i.e., $\rho^{DPO}(\delta(R_{-}))\ket{G}=0$ for any $G$ containing edges, while $\rho^{DPO}(\delta(R_{-}))\ket{G}$ evaluates to a linear combination of graph states obtained from $G$ by removing a vertex and all incident edges in all possible ways. However, the result presented in~\eqref{eq:exHWA} was in fact obtained \emph{independently} of concrete applications of the underlying rules to basis states, which entails that~\eqref{eq:exHWA} encodes a form of \emph{invariant} property of the rewriting systems arising from $R_{+}$ and $R_{-}$. This example is indicative of the general empirical observation that the combination of rule-algebraic computations with the representation property~\eqref{eq:RRAproperty} offers a fundamentally new approach to the study of categorical rewriting theories in many practical applications.
\end{example}

\section{Stochastic mechanics formalism}\label{sec:stochMech}

Referring to~\cite{bdg2016,bp2019-ext,bdg2019} for further details and derivations, suffice it here to highlight the key role played by the algebraic concept of \emph{commutators} in stochastic mechanics. Let us first provide the constructions of continuous-time Markov chains (CTMCs) and observables in stochastic rewriting systems.
\begin{remark}
Throughout this section, we fix the base field $\bK$ in all constructions to $\bK=\bR$.
\end{remark}

\begin{definition}
    Let $\bra{}:\hat{\bfC}\rightarrow \bR$ (referred to as \emph{dual projection vector}) be defined via its action on basis vectors of $\hat{\bfC}$ as $\braket{}{X}:=1_{\bR}$.
\end{definition}

\begin{theorem}\label{thm:CTMCs}
    Let $\bfC\in \cM-\mathbf{CAT}_{\bT}$ be an $\cM$-category suitable for type $\bT\in \{DPO,SqPO\}$ rewriting, and let $\cR_{\bfC}^{\bT}$ be the $\bT$-type rule algebra of linear rules with conditions over $\bfC$. Let $\rho\equiv \rho^{\bT}_{\bfC}$ denote the $\bT$-type canonical representation of $\cR_{\bfC}^{\bT}$. Then the following results hold:
\begin{enumerate}
    \item The basis elements of the space $\obs{\bfC}_{\bT}$ of \textbf{$\bT$-type observables}, i.e., the diagonal linear operators that arise as (linear combinations of) $\bT$-type canonical representations of rewriting rules with conditions, have the following structure ($ \hat{\cO}_{P,q}^{\ac{c}_P}$ in the DPO case, $\hat{\cO}_P^{\ac{c}_P}$ in the SqPO case):
\begin{equation}\label{eq:defObs}
\begin{aligned}
    \hat{\cO}_{P,q}^{\ac{c}_P}&:=\rho(\delta(P\xleftarrow{q}Q\xrightarrow{q}P,\ac{c}_P))
    \quad (P\in \obj{\bfC}_{\cong},q\in\cM,\ac{c}_P\in \cond{\bfC}_{\sim})\\
    \hat{\cO}_P^{\ac{c}_P}&:=\rho(\delta(P\xleftarrow{id}P\xrightarrow{id}P,\ac{c}_P))\quad (P\in \obj{\bfC}_{\cong},\ac{c}_P\in \cond{\bfC}_{\sim})\,.
\end{aligned}
\end{equation}
\item \textbf{DPO-type jump closure property:} for every linear rule with condition $R\equiv(O\hookleftarrow K\hookrightarrow I,\ac{c}_{I})\in \LinAc{\bfC}$, we find that
    \begin{equation}
        \bra{}\rho(\delta(R))=\bra{}\jcOp{\delta(R)}\,,
    \end{equation}
    where $\hat{\bO}:\cR^{\text{DPO}}_{\bfC}\rightarrow End_{\bR}(\hat{\bfC})$ is the homomorphism defined via its action on basis elements $\delta(R)$ for $R=(O\hookleftarrow K\hookrightarrow I, \ac{c}_{I})\in\LinEq{\bfC}$ as
    \begin{equation}
        \jcOp{\delta(R)}:=\rho(\delta(I\hookleftarrow K\hookrightarrow I,\ac{c}_{I}))\in \obs{\bfC}\,.
    \end{equation}
    \item \textbf{SqPO-type jump closure property:} for every linear rule with condition $R\equiv(O\hookleftarrow K\hookrightarrow I,\ac{c}_{I})\in \LinAc{\bfC}$, we find that
    \begin{equation}
        \bra{}\rho(\delta(R))=\bra{}\jcOp{\delta(R)}\,,
    \end{equation}
    where\footnote{Since in applications we will always fix the type of rewriting to either DPO or SqPO, we will use the same symbol for the jump-closure operator in both cases.} $\hat{\bO}:\cR^{\text{SqPO}}_{\bfC}\rightarrow End_{\bR}(\hat{\bfC})$ is the homomorphism defined via
    \begin{equation}
        \jcOp{\delta(R)}:=\rho(\delta(I\xleftarrow{id}I\xrightarrow{id}I,\ac{c}_{I}))\in \obs{\bfC}\,.
    \end{equation}
    \item \textbf{CTMCs via stochastic rewriting systems:} Let $\Prob{\bfC}$ be the space of \emph{(sub-)probability distributions over $\hat{\bfC}$} (i.e., $\ket{\Psi}=\sum_{X\in\obj{\bfC}_{\cong}}\psi_X\ket{X}$). Let $\cT$ be a collection of $N$ pairs of positive real-valued parameters $\kappa_j$ (referred to as \emph{base rates}) and linear rules $R_j$ with application conditions,
\begin{equation}
    \cT:=\{(\kappa_j,R_j)\}_{1\leq j\leq N}\qquad (\kappa_j\in \bR_{\geq 0}\,,\;R_j\equiv(r_j,\ac{c}_{I_j})\in \LinAc{\bfC})\,.
\end{equation} 
Then given an \emph{initial state} $\ket{\Psi_0}\in \Prob{\bfC}$, the $\bT$-type stochastic rewriting system based upon the transitions $\cT$ gives rise to the CTMC $(\cH,\ket{\Psi(0)})$ with time-dependent state $\ket{\Psi(t)}\in \Prob{\bfC}$ (for $t\geq0$) and evolution equation
\begin{equation}    
 \forall t\geq 0:\quad   \tfrac{d}{dt}\ket{\Psi(t)}=\cH\ket{\Psi(t)}\,,\quad \ket{\Psi(0)}=\ket{\Psi_0}\,.
\end{equation}
Here, the \emph{infinitesimal generator} $\cH$ of the CTMC is given by
\begin{equation}\label{eq:H}
    \cH=\hat{H}-\jcOp{\hat{H}}\,,\quad \hat{H}=\sum_{j=1}^N \kappa_j\,\rho(\delta(R_j))\,.
\end{equation}
\end{enumerate}
\end{theorem} 
\begin{proof}
	See~\ref{sec:CTMCproofsApp}.
\end{proof}

\begin{remark}
    The operation $\hat{\bO}$ featuring in the DPO- and SqPO-type jump-closure properties has a very intuitive interpretation: given a linear rule with condition $R\equiv(r,\ac{c}_{I})\in \LinAc{\bfC}$, the linear operator $\jcOp{\delta(R)}$ is an observable that evaluates on a basis vector $\ket{X}\in \hat{\bfC}$ as $\jcOp{\delta(R)}\ket{X}=(\text{\# of ways to apply $R$ to $X$})\cdot\ket{X}$.
\end{remark}

As for the concrete computational techniques offered by the stochastic mechanics formalism, one of the key advantages of this rule-algebraic framework is the possibility to reason about \emph{expectation values} (and higher moments) of pattern-count observables in a principled and universal manner. The precise formulation is given by the following generalization of results from~\cite{bdg2019} to the setting of DPO- and SqPO-type rewriting for rules with conditions:
\begin{theorem}\label{thm:CTMCmev}
    Given a CTMC $(\ket{\Psi_0},\cH)$ with time-dependent state $\ket{\Psi(t)}$ (for $t\geq 0$), a set of observables $O_1,\dotsc O_n\in\obs{\bfC}$ and $n$ \emph{formal variables} $\lambda_1,\dotsc,\lambda_n$, define the \emph{exponential moment-generating function (EMGF)} $M(t;\vec{\lambda})$ as
    \begin{equation}
        M(t;\vec{\lambda}):=\bra{}e^{\vec{\lambda}\cdot\vec{O}}\ket{\Psi(t)}\,,\quad  \vec{\lambda}\cdot\vec{O}:=\sum_{j=1}^n \lambda_jO_j\,.
    \end{equation}
    Then $M(t;\vec{\lambda})$ satisfies the following \emph{formal evolution equation} (for $t\geq0$):
    \begin{equation}\label{eq:MEGFevo}
    \begin{aligned}
	\tfrac{d}{dt}M(t;\vec{\lambda})&=\sum_{q\geq 1}\tfrac{1}{q!}\bra{}\left(ad_{\vec{\lambda}\cdot\vec{O}}^{\circ q}(\hat{H})\right)e^{\vec{\lambda}\cdot\vec{O}}\ket{\Psi(t)}\,,\quad M(0;\vec{\lambda})=\bra{}e^{\vec{\lambda}\cdot\vec{O}}\ket{\Psi_0}\,.
    \end{aligned}
    \end{equation}
\end{theorem}
\begin{proof}
	In full analogy to the case of rules without conditions~\cite{bdg2019}, the proof  follows from the BCH formula $e^{\lambda A}Be^{-\lambda A}=e^{ ad_{\lambda A}}(B)$ (for $A,B\in End_{\bR}(\hat{\bfC})$). Here, $ad_A^{\circ 0}(B):=B$, $ad_A(B):=AB-BA$  (also referred to as the \emph{commutator} $[A,B]$ of $A$ and $B$), and $ad_A^{\circ(q+1)}(B):=ad_A(ad_A^{\circ q}(B))$ for $q\geq 1$. Finally, the $q=0$ term in the above expression evaluates identically to $0$ due to $\bra{}\cH=0$.
\end{proof}

Combining this theorem with the notion of $\bT$-type jump-closure, one can in favorable cases express the EMGF evolution equation as a PDE on formal power series in $\lambda_1,\dotsc,\lambda_n$ and with $t$-dependent real-valued coefficients. Referring the interested readers to~\cite{bdg2019} for further details on this technique, let us provide here a simple non-trivial example of such a calculation.

\begin{example}\label{ex:ugModel}
    Let us consider a stochastic rewriting system over the category $\bfC=\mathbf{uGraph}$ of finite undirected multigraphs, with objects constrained by the structure constraint $\ac{c}^S_{\mIO}:=\neg \exists(\mIO\hookrightarrow \ti{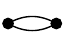})\in\cond{\mathbf{uGraph}}$ that prohibits multiedges. %
    Let us consider for type $\bT=SqPO$ the four rule algebra elements based upon rules with conditions which implement edge-creation/-deletion and vertex creation/deletion, respectively, defined as 
\begin{equation*}
\begin{aligned}
        E_{+}&:=\tfrac{1}{2}\delta\left(\ti{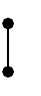}\hookleftarrow \ti{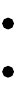}\hookrightarrow \ti{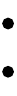}\,,\neg \exists \left(\ti{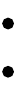}\hookrightarrow \ti{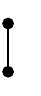}\right)\right)\,, &
V_{+}&:=\delta(\ti{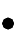}\hookleftarrow\mIO\hookrightarrow\mIO;\ac{true})\\
E_{-}&:=\tfrac{1}{2}\delta\left(\ti{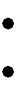}\hookleftarrow \ti{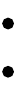}\hookrightarrow \ti{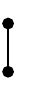};\ac{true}\right)\,,\quad &
V_{-}&:=\delta(\mIO\hookleftarrow\mIO\hookrightarrow
\ti{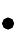};\ac{true})\,.
\end{aligned}
    \end{equation*}
Here, the prefactors $\tfrac{1}{2}$ in the definition of $E_{\pm}$ are chosen purely for convenience. Note that the rule underlying $E_{+}$ is the only rule requiring a non-trivial application condition, since linking two vertices with an edge might create a multiedge (precisely when the two vertices were already linked). Introducing base rates $\nu_{\pm},\varepsilon_{\pm}\in \bR_{>0}$ and letting $\hat{X}:=\rho(X)$ (for $\rho:=\rho_{\mathbf{uGraph}}^{SqPO}$), we may assemble the infinitesimal generator $\cH$ of a CTMC as
\begin{equation}
\cH=\hat{H}-\jcOp{\hat{H}}\,,\quad
\hat{H}:=\nu_{+}\hat{V}_{+}+\nu_{-}\hat{V}_{-}+\varepsilon_{+}\hat{E}_{+}+\varepsilon_{-}\hat{E}_{-}\,.
\end{equation}
One might now ask whether there is any interesting dynamical structure e.g.\ in the evolution of the moments of the observables that count the number of times each of the transitions of this system is applicable,
\begin{equation}
    O_{\bullet\vert \bullet}:=\jcOp{E_{+}}\,,\;
    O_{\bullet\!-\! \bullet}:=\jcOp{E_{-}}\,,\; O_{\bullet}:=\jcOp{V_{-}}\,.
\end{equation}
The algebraic data necessary in order to formulate EMGF evolution equations are all 
\textbf{\emph{commutators}} of the observables with the contributions $\hat{X}:=\rho(X)$  to the ``off-diagonal part'' $\hat{H}$ of the infinitesimal generator $\cH$. We will present here for brevity just those commutators necessary in order to compute the evolution equations for the averages of the three observables:
\begin{equation}
\begin{aligned}
	[O_{\bullet},\hat{V}_{\pm}]&=\pm \hat{V}_{\pm}\,,\; & 	
	[O_{\bullet},\hat{E}_{\pm}]&=0 &&\\
	[O_{\bullet\vert \bullet},\hat{V}_{+}]&= \hat{A}\,, &
	[O_{\bullet\vert \bullet},\hat{V}_{-}]&= -\hat{B}\,,\; &
	[O_{\bullet\vert \bullet},\hat{E}_{\pm}]&= \mp \hat{E}_{\pm}\\
	[O_{\bullet\!-\! \bullet},\hat{V}_{+}]&=0\,, &
	[O_{\bullet\!-\! \bullet},\hat{V}_{-}]&= -\hat{C}\,,\; &
	[O_{\bullet\!-\! \bullet},\hat{E}_{\pm}]&= \pm \hat{E}_{\pm}
\end{aligned}
\end{equation}
As typical in these types of commutator computations, we find a number of contributions (here $\hat{A}$, $\hat{B}$ and $\hat{C}$) that were neither observables nor contributions to the off-diagonal part of the infinitesimal generator $\cH$:
\begin{equation*}
\begin{aligned}
\hat{A}&:=\rho\left(\delta\left(\ti{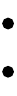}\hookleftarrow \ti{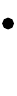}\hookrightarrow \ti{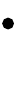}\,,\ac{true}\right)\right),\;
\hat{B}:=\rho\left(\delta\left(\ti{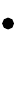}\hookleftarrow \ti{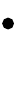}\hookrightarrow \ti{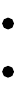}\,,\neg \exists \left(\ti{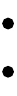}\hookrightarrow \ti{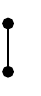}\right)\right)\right)\\
\hat{C}&:=\rho\left(\delta\left(\ti{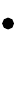}\hookleftarrow \ti{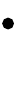}\hookrightarrow \ti{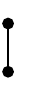}\,,\ac{true}\right)\right)\,.
\end{aligned}
\end{equation*}
For the computation of the moment evolution equations, applying the SqPO-type jump-closure operator to these additional contributions,
\begin{equation}
\jcOp{\hat{A}}=O_{\bullet}\,,\quad
\jcOp{\hat{B}}=2O_{\bullet\vert\bullet}\,,\quad
\jcOp{\hat{C}}=2O_{\bullet\!-\!\bullet}\,,
\end{equation}
we discover that all the resulting diagonal linear operators are linear combinations of the observables already encountered in the diagonal part of the infinitesimal generator $\cH$. Picking for simplicity as an initial state $\ket{\Psi(0)}=\ket{\mIO}$ just the empty graph, and invoking the SqPO-type jump-closure property (cf.\ Theorem~\ref{thm:CTMCs}) repeatedly in order to evaluate $\langle [O_P,\hat{H}]\rangle(t)=\langle \jcOp{[O_P,\hat{H}]}\rangle(t)$, the moment EGF evolution equation~\eqref{eq:MEGFevo} specializes to the following ``Ehrenfest-like''~\cite{bdg2019} ODE system:
\begin{equation*}
\begin{aligned}
\tfrac{d}{dt}\langle O_{\bullet}\rangle(t)&=\langle [O_{\bullet},H]\rangle(t)=\nu_{+}-\nu_{-}\langle O_{\bullet}\rangle(t)\\
\tfrac{d}{dt}\langle O_{\bullet\vert\bullet}\rangle(t)&=\langle [O_{\bullet\vert\bullet},H]\rangle(t)
=\nu_{+}\langle O_{\bullet}\rangle(t)
-(2\nu_{-}+\varepsilon_{+})\langle O_{\bullet\vert\bullet}\rangle(t)
+\varepsilon_{-}\langle O_{\bullet\!-\!\bullet}\rangle(t)\\
\tfrac{d}{dt}\langle O_{\bullet\!-\!\bullet}\rangle(t)&=\langle [\langle O_{\bullet\!-\!\bullet}\rangle(t),H]\rangle(t)
=\varepsilon_{+}\langle O_{\bullet\vert\bullet}\rangle(t)
-(2\nu_{-}+\varepsilon_{-})\langle O_{\bullet\!-\!\bullet}\rangle(t)\\
\langle O_{\bullet}\rangle(0)&=\langle O_{\bullet\vert\bullet}\rangle(t)=\langle O_{\bullet\!-\!\bullet}\rangle(t)=0\,.
\end{aligned}
\end{equation*}
This ODE system may be solved exactly (see~\ref{app:se}). We depict in Figure~\ref{fig:evoEx} two exemplary evolutions of the three average pattern counts for different choices of parameters. Since due to SqPO-semantics the vertex deletion and creation transitions are entirely independent of the edge creation and deletion transitions, the vertex counts stabilize on a Poisson distribution of parameter $\nu_{+}/\nu_{-}$ (where we only present the average vertex count value here). As for the non-linked vertex pair and edge patter counts, the precise average values are sensitive to the parameter choices (i.e., whether or not vertices tend to be linked by an edge or not may be freely tuned in this model via adjusting the parameters).
\end{example}

\begin{figure}[t]\label{fig:evoEx}
 \centering
    \subfigure[\label{fig:a}Vertices tend to be linked.]{\includegraphics[width=0.45\textwidth]{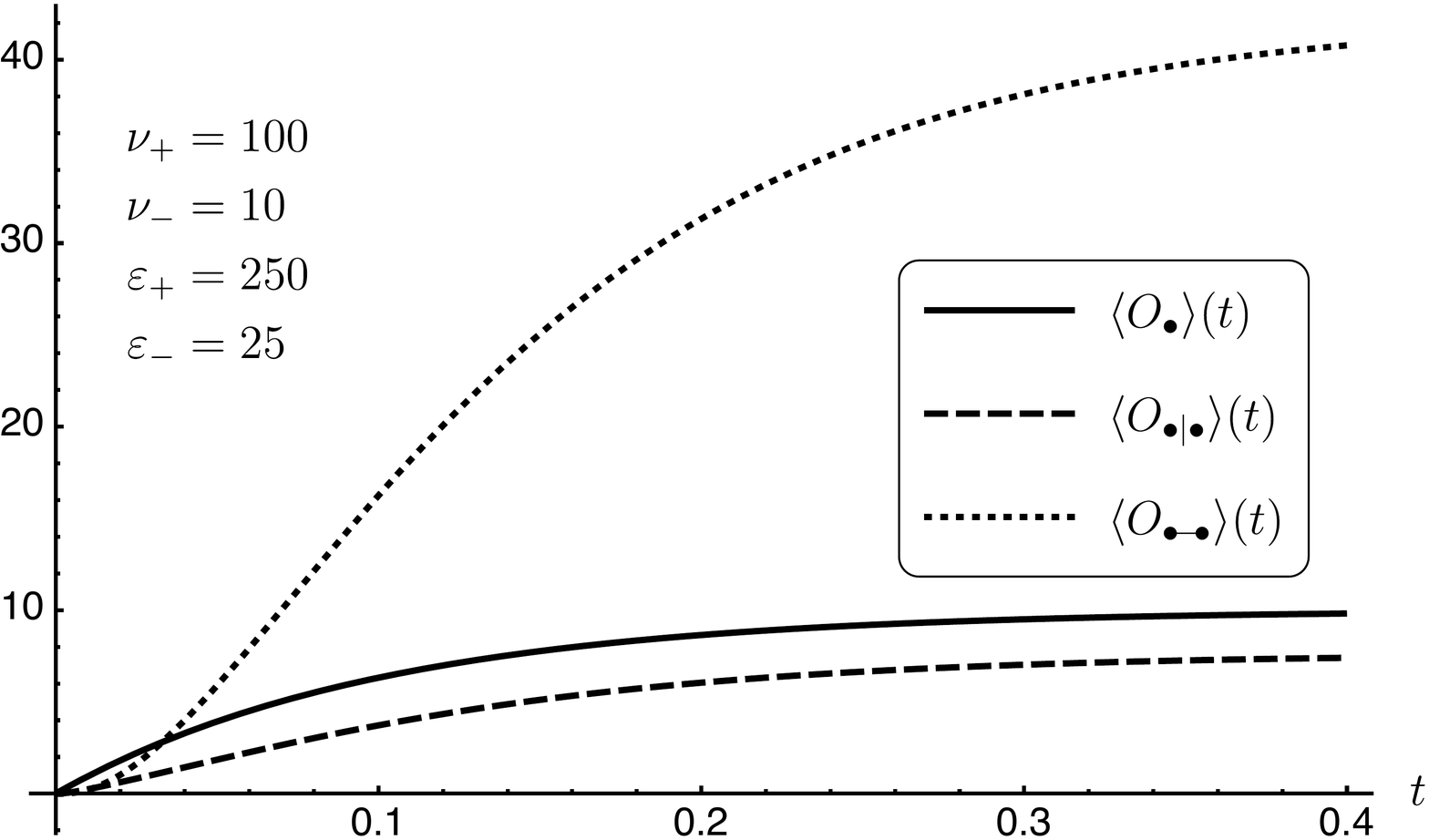}}
     \subfigure[\label{fig:b}Vertices tend to be unlinked.]{\includegraphics[width=0.45\textwidth]{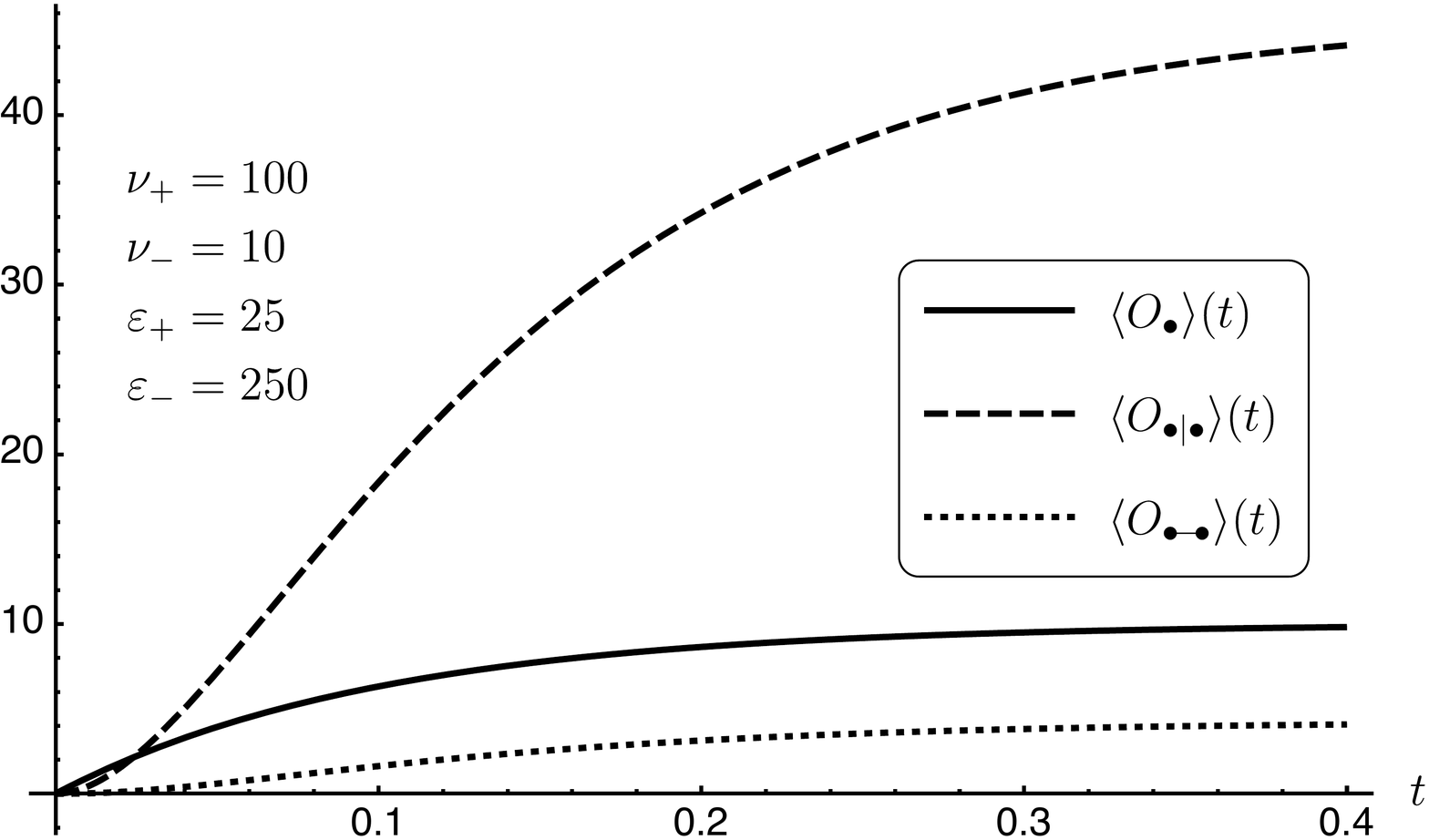}}
\caption{Time-evolutions of pattern count observables for different parameter choices.}
\end{figure}

While the example presented was chosen mainly to illustrate the computational techniques, it highlights the typical feature of the emergence of contributions in the relevant (nested) commutator calculations that may in general not necessarily have been included in the non-diagonal part $\hat{H}$ of the infinitesimal generator of the CTMC (leading to the phenomenon of so-called moment-ODE ``non-closure'', cf.\ e.g.\ Example~\ref{ex:KappaWorkedExample} in Section~\ref{sec:bcgr}). We refer the interested readers to~\cite{bdg2019} for an extended discussion of this phenomenon, and for computation strategies for higher-order moment evolution equations.

\section{Restricted rewriting theories}\label{sec:rrt}

One of the quintessential and well-known problems in working in the setting of rules with conditions is the increase in complexity each time a $\Shift$ construction is employed in rule composition operations. Concretely, even for moderately complex original application conditions, performing a $\Shift$ operation in general leads to a shifted condition that consists of a large number of atomic subformulae (cf.\ Example~\ref{ex:sgraph-cpc}). This is particularly problematic for the implementation of rule-algebraic computations, since here it is necessary to perform comparison operations on the outcomes of rule composition operations (i.e., isomorphism checks on the ``plain'' composite rules, and equivalence checks on the application conditions of the composite rules). In this section, we present a number of technical developments that aim to render rule-algebraic computations more efficient. We first demonstrate in Section~\ref{sec:crrt} that the concurrency theorems in both DPO- and SqPO-type rewriting may be refined into variants that are applicable in the setting of \emph{restricted rewriting}, wherein the rewriting of objects satisfying global structure constraints is considered. We then present in Section~\ref{sec:aps} a paradigm inspired by the \KAP{} framework (and in close analogy to~\cite{Danos2014}; see also Section~\ref{sec:bcgr}), namely a class of restricted rewriting theories with a certain characteristic structure on global constraints that permits to further simplify the implementations.

\subsection{Concurrency theorems for restricted rewriting}\label{sec:crrt}

An important concept in the work of Habel and Pennemann~\cite{habel2009correctness,Pennemann:aa} in their framework for rewriting with conditions has been a characterization of application conditions for rules by their capacity to preserve or to guarantee constraints on objects, which we slightly reformulate here as follows (compare~\cite[Defn.~9 and Cor.~5]{habel2009correctness}):

\begin{lemma}[Constraint-preserving and -guaranteeing completions]\label{lem:cgcp}
Let $\bfC\in \cM-\mathbf{CAT}_{\bT}$ be an $\cM$-adhesive category suitable for rewriting of type $\bT\in\{DPO,SqPO\}$, and let $R\equiv[(r,\ac{c}_I)]_{\sim}\in \LinAc{\bfC}$ be a linear rule with application condition. Given a \emph{global constraint} $\ac{c}_{\mIO}$ (considered to be kept fixed throughout all computations), let $\widetilde{\ac{c}_I}$ denote the \emph{constraint-guaranteeing completion} and $\overline{\ac{c}_I}$ the \emph{constraint-preserving completion} of $\ac{c}_I$ w.r.t.\ to $\ac{c}_{\mIO}$, with the defining properties
\begin{subequations}
\begin{align}
&\forall X\in \obj{\bfC}\,,\; m\in \MatchGT{\bT}{\widetilde{R}}{X}: R_m(X) \vDash \ac{c}_{\mIO} \label{def:cgua}\\
&\forall X\in \obj{\bfC}\,,\; m\in \MatchGT{\bT}{\overline{R}}{X}:\; X\vDash\ac{c}_{\mIO} \Rightarrow R_m(X) \vDash \ac{c}_{\mIO} \label{def:cpres}\,,
\end{align}
\end{subequations}
where we used the notations $\widetilde{R}\equiv[(r,\widetilde{\ac{c}_I})]_{\sim}$ and $\overline{R}\equiv[(r,\overline{\ac{c}_I})]_{\sim}$. Then the two variants of completions are computed from $\ac{c}_I$ as follows:
\begin{subequations}\label{eq:lemcgcp}
\begin{align}
\widetilde{\ac{c}_{I}}&= \ac{c}_I \land \Trans(r, \Shift(O\hookleftarrow \mIO, \ac{c}_{\mIO})) \land \Shift(I\hookleftarrow \mIO,\ac{c}_{\mIO})\\
\overline{\ac{c}_I}& = \Shift(I\hookleftarrow \mIO,\ac{c}_{\mIO})\Rightarrow \widetilde{\ac{c}_I}\,. 
\end{align}
\end{subequations}
Here, for two conditions $\ac{c},\ac{c}'\in \cond{\bfC}$ (with both conditions defined over the same object), we define the \emph{implication operation} as $\ac{c}'\Rightarrow \ac{c} := \ac{c} \lor (\neg \ac{c}')$. Finally, the equations in~\eqref{eq:lemcgcp} encode the additional useful relation
\begin{equation}\label{eq:cGuaV2}
 \widetilde{\ac{c}_I} =\overline{\ac{c}_I} \land \Shift(I\hookleftarrow \mIO,\ac{c}_{\mIO})\,. 
\end{equation}
\end{lemma}
\begin{proof}
The proof for the claim regarding $\widetilde{\ac{c}_I}$ follows form the defining properties of $\Trans$ and $\Shift$ (Theorem~\ref{thm:STdefns}) and from the compositionality property of $\Shift$ (cf.\ Theorem~\ref{thm:STcompComp}), where $(m^{*}: \tilde{R}_m(X)\hookleftarrow O)\in \cM$ denotes the comatch of the match $(m:X\hookleftarrow I)$:
\begin{align*}
m&\vDash\widetilde{\ac{c}_I}\;\Leftrightarrow \; m\vDash \Trans(r,\Shift(O\hookleftarrow \mIO,\ac{c}_{\mIO}))\;\land\;
m\vDash \Shift(I\hookleftarrow \mIO,\ac{c}_{\mIO})\\
\Leftrightarrow\quad m^{*}&\vDash \Shift(O\hookleftarrow \mIO,\ac{c}_{\mIO})\;\land\; m\vDash \Shift(I\hookleftarrow \mIO,\ac{c}_{\mIO})
\qquad\text{(by definition of $\Trans$)}\\
\Leftrightarrow\quad m^{*}\circ (O\hookleftarrow \mIO) 
&= (\tilde{R}_m(X)\hookleftarrow \mIO)\vDash \ac{c}_{\mIO}\;\land\;
m\circ(I\hookleftarrow \mIO)=(X\hookleftarrow \mIO) \vDash \ac{c}_{\mIO} \qquad\text{(by compositionality of $\Shift$)}
\end{align*}%
We have thus proved that $(m:X\hookleftarrow I)\vDash \widetilde{\ac{c}_I}$ guarantees that both $X\vDash \ac{c}_{\mIO}$ and $\tilde{R}_m(X)\vDash\ac{c}_{\mIO}$.

To prove the statement regarding the constraint-preserving completion $\overline{\ac{c}_I}$, we utilize the following logical tautology (for $Z\in \obj{\bfC}$ and $\ac{c}_Z,\ac{c}'_Z\in \cond{\bfC}$):
\begin{equation} \label{eq:usefulTautology}
(\ac{c}_Z'\Rightarrow \ac{c}_Z)\land \ac{c}_Z' = (\ac{c}_Z\lor (\neg\ac{c}_Z'))\land \ac{c}_Z' 
=(\ac{c}_Z\land \ac{c}_Z') \lor ((\neg\ac{c}_Z')\land \ac{c}_Z') =\ac{c}_Z\land \ac{c}_Z'\,.
\end{equation}
Applying this tautology to the case at hand, we may derive the statement of~\eqref{eq:cGuaV2}:
\begin{equation}
\overline{\ac{c}_I}
\land 
\Shift(I\hookleftarrow \mIO,\ac{c}_{\mIO})
=(\Shift(I\hookleftarrow \mIO,\ac{c}_{\mIO})
\Rightarrow \widetilde{\ac{c}_I})\land \Shift(I\hookleftarrow \mIO,\ac{c}_{\mIO}) =\widetilde{\ac{c}_{I}}\,.
\end{equation}
Finally, for some admissible match $(m:X\hookleftarrow I)$ of $\overline{R}$ into $X$, by definition of $\Shift$, $X\vDash \ac{c}_{\mIO}$ together with $(X\hookleftarrow \mIO)=m\circ (I\hookleftarrow \mIO)$ implies that $m\vDash \Shift(I\hookleftarrow \mIO, \ac{c}_{\mIO})$, so that
\begin{equation}
X\vDash \ac{c}_{\mIO}\land m\vDash \bar{\ac{c}}_{I}\; \Leftrightarrow \;
m\vDash \Shift(I\hookleftarrow \mIO, \ac{c}_{\mIO})\land m\vDash \bar{\ac{c}}_{I}
\land m \vDash (\Shift(I\hookleftarrow \mIO, \ac{c}_{\mIO})\land \bar{\ac{c}}_{I})
\;\Leftrightarrow\; m\vDash \tilde{\ac{c}}_I\,,
\end{equation}
so that in particular $\bar{R}_m(X)\vDash \ac{c}_{\mIO}$.
\end{proof}

\begin{remark} It is worthwhile emphasizing that completing the application condition of an arbitrary linear rule (which possibly already carries some non-trivial application condition prior to completion) into a constraint-\emph{guaranteeing} completion yields a rule that does not have any admissible matches into objects that do not themselves satisfy the global constraint; in contrast, the constraint-\emph{preserving} completion only guarantees that the global constraint is preserved upon application along admissible matches into objects that themselves satisfy the global constraint. Notably, objects that do \emph{not} satisfy the global constraint might still possess admissible matches for the latter setting, so that it is a priori not evident that one may build a consistent rule-composition operation with just the weaker information of constraint-preserving completions available. Finally, it is possible in general that the constraint-guaranteeing or -preserving completion of a rule might yield an application condition that is equivalent to $\ac{false}$, which signals that the particular rule is intrinsically incompatible with the chosen global constraint $\ac{c}_{\mIO}$. 
\end{remark}

\begin{definition}[Constraint completions for ``plain'' rules]\label{def:constrcpr}
In a slight extension of the terminology introduced in Lemma~\ref{lem:cgcp}, given a ``plain'' rule $r\in \Lin{\bfC}$ (i.e., a rule without application condition), we will define the \emph{constraint-guaranteeing completion} $\tilde{r}$ and the \emph{constraint-preserving completion} $\bar{r}$ of the ``plain'' rule $r$ as 
\begin{equation}
\tilde{r}:=[\widetilde{(r,\ac{true})}]_{\sim}\,,\quad \bar{r}:=[\overline{(r,\ac{true})}]_{\sim}\,.
\end{equation}
In other words, $\tilde{r},\bar{r}\in \LinAc{\bfC}$ are defined\footnote{This definition is consistent since it is clear from the definition of DPO- and SqPO-type semantics that the rule $R=(r,\ac{true})$ emulates precisely the semantics of the rule $r$ in $\bT$-type rewriting (for $\bT\in \{DPO,SqPO\}$) for rules without conditions: any $\cM$-morphism $(m:I\hookrightarrow X)$ satisfies the trivial condition $\ac{true}$, and thus is a $\bT$-admissible match of $r$ (Definition~\ref{def:dd}) if it satisfies the respective admissibility condition for ``plain'' rules; moreover, for two ``plain'' rules $r_1,r_2\in \Lin{\bfC}$, letting $R_j=(r_j,\ac{true})$ ($j=1,2$), an $\cM$-span $(I_2\hookleftarrow M\hookrightarrow O_1)$ qualifies as a $\bT$ admissible match of $R_2$ into $R_1$ iff it is a $\bT$-admissible match of $r_1$ into $r_1$ in the $\bT$-semantics without conditions, and composite rules of $R_2$ with $R_1$ all have trivial conditions $\ac{true}$.} as the constraint-guaranteeing and -preserving completions of the rule $R=(r,\ac{true})$, respectively. We will furthermore employ the notational conventions $\tilde{r}=[(r,\widetilde{\ac{c}_I})]_{\sim}$ and $\bar{r}=[(r,\overline{\ac{c}_I})]_{\sim}$ (typically in order to explicitly refer to the application conditions $\widetilde{\ac{c}_I}$ or $\overline{\ac{c}_I}$, respectively) .
\end{definition}

\begin{example}\label{ex:sgraph-cpc}
To illustrate the utility of the notion of constraint-\emph{preserving} completions as opposed to constraint-\emph{guaranteeing} completions of application conditions in practical computations, consider the setting of rewriting in $\mathbf{uGraph}$ and with global constraint $\ac{c}_{\mIO}$ chosen such as to prohibit multi-edges (in the notational convention of~\eqref{eq:constrNotation}):
\begin{equation}\label{eq:nmec}
	\ac{c}_{\mIO} :=  {\color{red} \neg \exists\left(\ti{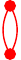} \right)}
\end{equation}
Let $r_{E_{\pm}}\in \Lin{\mathbf{uGraph}}$ denote the linear rules that link/unlink two vertices with an edge:
\begin{equation}
r_{E_{+}}:=\left(\ti{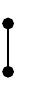}\hookleftarrow \ti{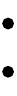}\hookrightarrow \ti{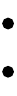}\right)\,, \quad
r_{E_{-}}:=\left(\ti{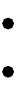}\hookleftarrow \ti{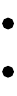}\hookrightarrow \ti{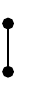}\!\!\!\right)
\end{equation}
In order to compute the constraint-guaranteeing and -preserving completions for $R_{e_{\pm}}:=[(r_{E_{\pm}},\ac{true})]_{\sim}$, we first need to compute the $\Shift$ and $\Trans$ operations to extend $\ac{c}_{\mIO}$ to the input interfaces of the two rules:
{\gdef\tpScale{0.43}
\begin{equation}
\begin{aligned}
\Shift\left(
\mIO\hookrightarrow \ti{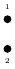},
{\color{red}\neg \exists\left(\ti{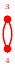}
\right)}
\right)&=
\bigwedge_{{\color{red}N'\in \cN'}}\neg \exists\left(
     \ti{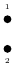}\hookrightarrow {\color{red}N'}\right)\,,	&
{\color{red}\cN'}&=\left\{
\ti{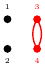},\;
\ti{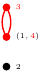},\;
\ti{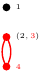},\;
\ti{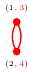}
\right\}
\\
\Trans\left(r_{E_{-}},\Shift\left(
\mIO\hookrightarrow \ti{shiftA1},
{\color{red}\neg \exists\left(\ti{shiftA2}
\right)}
\right)\right)&=
\bigwedge_{{\color{red}N''\in \cN''}}\neg \exists\left(
     \ti{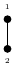}\hookrightarrow {\color{red}N''}\right)\,,		&
{\color{red}\cN''}&=\left\{
\ti{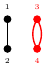},\;
\ti{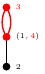},\;
\ti{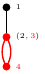},\;
\ti{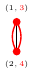}
\right\}
\\
\Shift\left(
\mIO\hookrightarrow \ti{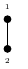},
{\color{red}\neg \exists\left(\ti{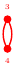}
\right)}
\right)&=
\bigwedge_{{\color{red}N'''\in \cN'''}}\neg \exists\left(
     \ti{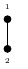}\hookrightarrow {\color{red}N'''}\right)\,,	&
{\color{red}\cN'''}&=\left\{
\ti{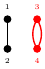},\;
\ti{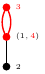},\;
\ti{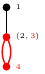},\;
\ti{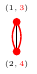},\;
\ti{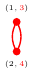}
\right\}
\\
\Trans\left(r_{E_{+}},\Shift\left(
\mIO\hookrightarrow \ti{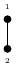},
{\color{red}\neg \exists\left(\ti{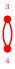}
\right)}
\right)\right)&=
\bigwedge_{{\color{red}N''''\in \cN''''}}\neg \exists\left(
     \ti{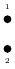}\hookrightarrow {\color{red}N''''}\right)\,,	&
{\color{red}\cN''''}&=\left\{
\ti{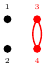},\;
\ti{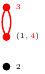},\;
\ti{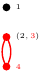},\;
\ti{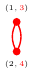},\;
\ti{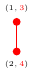}
\right\}
\end{aligned}
\end{equation}
\gdef\tpScale{0.7}}
We may then compute the two sets of completions according to Lemma~\ref{lem:cgcp} as follows:
\begin{equation}
\begin{aligned}
\widetilde{\ac{c}_{I_{E_{+}}}}&= \bigwedge_{N\in \cN'\cup\cN'''} \neg \exists\left(\ti{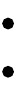}
\hookrightarrow N \right)\,,\quad &
\widetilde{\ac{c}_{I_{E_{-}}}}&= \bigwedge_{N\in \cN''\cup\cN''''} \neg \exists\left(\ti{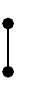}\hookrightarrow N \right)\\
\overline{\ac{c}_{I_{E_{+}}}}&= \neg \exists\left(\ti{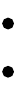}\hookrightarrow \ti{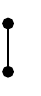} \right)\,, &
\overline{\ac{c}_{I_{E_{-}}}}&= \ac{true}\,.
\end{aligned}
\end{equation}
\end{example}
Consequently, even for this particularly simple example of rewriting rules and choice of global constraint, it is evident that being able to work with constraint-\emph{preserving} rather than constraint-\emph{guaranteeing} completions of application conditions offers a significant decrease in complexity. In light of the bio- and organo-chemical graph rewriting in the form introduced in Sections~\ref{sec:bcgr} and~\ref{sec:ocgr}, Example~\ref{ex:sgraph-cpc} moreover harbors yet another important consequence, whose statement requires the following technical lemma:
\begin{lemma}[Bridges and the $\Shift$ construction]\label{lem:Bridges}
Let $A\in \obj{\mathbf{uGraph}}$ be an undirected graph with the special property that there exists an edge $e\in E_A$ such that $A\setminus e = A_1 + A_2$ (i.e., removing $e$ renders two disconnected subgraphs $A_1$ and $A_2$; $e$ is referred to as a \emph{bridge} in graph theory). Then for any two $\cM$-morphisms $(\beta_j:A_j\hookrightarrow B_j)\in \cM$ ($j=1,2$), letting $B:=\pO{A\hookleftarrow A_1+A_2\hookrightarrow B_1+B_2}$, one finds
\begin{equation}
\Shift(A_1+A_1 \xhookrightarrow{[\beta_1,\beta_2]} B_1+B_2, \neg \exists(A_1+A_2\hookrightarrow A)) =
\neg \exists(B_1+B_2\hookrightarrow B)\,.
\end{equation}
\end{lemma}
\begin{proof}
See~\ref{app:proof-lem-bridges}.
\end{proof}
Crucially, the technical results presented in Example~\ref{ex:sgraph-cpc} and Lemma~\ref{lem:Bridges} permit to understand that in the important application scenario of rewriting \emph{simple} undirected graphs via imposing a global non-multi-edge constraint $\ac{c}_{\mIO}$ as defined in~\eqref{eq:nmec}, all rewriting rules obtained via composing (finitely many) copies of the edge-linking and -unlinking rules $r_{E_{\pm}}$ possess extremely simple constraint-preserving application conditions when computing their constraint-preserving completions w.r.t.\ $\ac{c}_{\mIO}$. Concretely, any composite rule of the type described possesses a constraint-preserving application condition that merely states that certain pairs of vertices in the input pattern of the rule must not be linked. With a similar statement true for \emph{typed} undirected graphs (as will be the starting point for both bio- and organo-chemical rewriting systems), this significantly reduces the complexity of algorithms for rule compositions, direct derivations and ultimately rule-algebraic CTMC computations, especially as it will often be the case in applications that the global constraint $\ac{c}_{\mIO}$ consists of a large number of atomic constraints.

In light of this empirical finding, it is thus highly desirable to find an algorithmic approach to work consistently with the constraint-preserving completions also in rule-composition computations, which motivates the introduction of the following ``restricted'' variants of the concurrency theorems:
\begin{theorem}[Restricted Concurrency Theorems]\label{thm:rcct}
Let $\bfC\in \cM-\mathbf{CAT}_{\bT}$ be an $\cM$-adhesive category suitable for type $\bT\in \{DPO,SqPO\}$ rewriting, and fix a \emph{global constraint} $\ac{c}_{\mIO}\in \cond{\bfC}$. Given rewriting rules $R_j\equiv(r_j, \ac{c}_{I_j})\in \LinAc{\bfC}$ ($j=1,2$), denote by $\widetilde{R}_j\equiv [(r_j,\widetilde{\ac{c}_{I_j}})]_{\sim}$ and $\overline{R}_j\equiv[(r_j,\overline{\ac{c}_{I_j}})]_{\sim}$ the constraint-guaranteeing and constraint-preserving completions of the two rules. Then the following properties hold:
\begin{enumerate}[label=(\roman*)]
\item For all $\bT$-admissible matches $\tilde{\mu}=(I_2\hookleftarrow M_{21}\hookrightarrow O_1)\in \MatchGT{\bT}{\widetilde{R}_2}{\widetilde{R}_1}$ of $\widetilde{R}_2$ into $\widetilde{R}_1$, performing the $\bT$-type rule composition according to~\eqref{eq:defRcomp}, we find that the application condition $\ac{c}_{\tilde{I}_{21}}$ of the composite rule, computed according to~\eqref{eq:acRcomp} as
\begin{equation}
\ac{c}_{\tilde{I}_{21}} = \Trans(P_{21}\leftharpoonup I_{21},\Shift(P_{21}\hookleftarrow I_2,\widetilde{\ac{c}_{I_2}})) \land \Shift(I_{21}\hookleftarrow I_1,\widetilde{\ac{c}_{I_1}})\,.
\end{equation}
may be equivalently expressed as\footnote{Note that the condition $\ac{c}_{\widetilde{I}_{21}}$ does \emph{not} constitute the constraint-guaranteeing completion of the composite of $\widetilde{R}_2$ with $\widetilde{R}_1$ along match $\tilde{\mu}$, but merely the condition computed by composing the two rules along the standard $\bT$-type composition operation (i.e., one would need an additional explicit completion operation according to Lemma~\ref{lem:cgcp} in order to obtain the constraint-guaranteeing variant of the condition). In contrast, $\overline{\ac{c}_{\overline{I}_{21}}}$ is precisely the constraint-preserving completion of the condition $\ac{c}_{\overline{I}_{21}}$ obtained by computing the $\bT$-type composition of the rules $\overline{R}_2$ and $\overline{R}_1$ along $\tilde{\mu}$, whence the subtle difference in the notations utilized.}
\begin{equation}
\begin{aligned}
\ac{c}_{\tilde{I}_{21}} & \dot{\equiv}\, \overline{\ac{c}_{\overline{I}_{21}}} \land \Shift(I_{21}\hookleftarrow \mIO,\ac{c}_{\mIO})\\
\overline{\ac{c}_{\overline{I}_{21}}}&:=
\Shift(I_{21}\hookleftarrow \mIO,\overline{\ac{c}_{\mIO}}) \Rightarrow \bigg(
\Trans(P_{21}\leftharpoonup I_{21},\Shift(P_{21}\hookleftarrow I_2,\overline{\ac{c}_{I_2}})) \land \Shift(I_{21}\hookleftarrow I_1,\overline{\ac{c}_{I_1}})\bigg)\,.
\end{aligned}
\end{equation}
\item For every object $\overline{X}\in \obj{\bfC}$ such that $\overline{X}\vDash \ac{c}_{\mIO}$, there exists an \emph{isomorphism} $\alpha:\widetilde{\cA}\rightarrow \overline{\cA}$ between sets of pairs of admissible matches of the form
\begin{equation}
\begin{aligned}
\widetilde{\cA} &:= \{ (\tilde{m}_1,\tilde{m}_2) \mid \tilde{m}_1\in \MatchGT{\bT}{\widetilde{R}_1}{\overline{X}}\,,\; 
\tilde{m}_2\in \MatchGT{\bT}{\widetilde{R}_2}{\widetilde{R}_{1_{\tilde{m}_1}}(\overline{X})}\}\\
\overline{\cA} &:= \{ (\bar{m}_1,\bar{m}_2) \mid \bar{m}_1\in \MatchGT{\bT}{\overline{R}_1}{\overline{X}}\,,\; 
\bar{m}_2\in \MatchGT{\bT}{\overline{R}_2}{\overline{R}_{1_{\bar{m}_1}}(\overline{X})}\}
\end{aligned}
\end{equation}
such that for all $(\tilde{m}_1,\tilde{m}_2)\in \widetilde{\cA}$, letting $(\bar{m}_1,\bar{m}_2)=\alpha(\tilde{m}_1,\tilde{m}_2)\in \overline{\cA}$,
\begin{equation}
\widetilde{R}_{2_{\tilde{m}_2}}\left(\widetilde{R}_{1_{\tilde{m}_1}}(\overline{X})\right)\cong 
\overline{R}_{2_{\bar{m}_2}}\left(\overline{R}_{1_{\bar{m}_1}}(\overline{X})\right)\,.
\end{equation}
\item Let the \emph{restricted set of $\bT$-admissible matches} $\overline{\RMatchGT{\bT}{\overline{R}_2}{\overline{R}_1}}$ of $\overline{R}_2$ into $\overline{R}_1$ be defined as
\begin{equation}
\overline{\RMatchGT{\bT}{\overline{R}_2}{\overline{R}_1}}:= \{ \bar{\mu}\in \RMatchGT{\bT}{\overline{R}_2}{\overline{R}_1} \mid \Shift(I_{21}\hookleftarrow \mIO,\ac{c}_{\mIO}) \not\!\! {\equiv}\; \ac{false}\}\,.
\end{equation}
Then there exists an \emph{isomorphism} $\beta:\widetilde{\cB}\rightarrow \overline{\cB}$ between pairs of admissible matches of the form\footnote{Note in particular that the set $\overline{\cB}$ contains pairs of the form $(\overline{\mu},\bar{m}_{21})$ where $\bar{m}_{21}$ is an admissible match into $\overline{X}$ of the rule $\overline{\compGT{\bT}{\overline{R}_2}{\bar{\mu}}{\overline{R}_1}}$, which is the constraint-preserving completion (computed according to Lemma~\ref{lem:cgcp}) of the composite rule ${\compGT{\bT}{\overline{R}_2}{\bar{\mu}}{\overline{R}_1}}$.}
\begin{equation}
\begin{aligned}
\widetilde{\cB} &:= \{ (\tilde{\mu},\tilde{m}_{21}) \mid \tilde{\mu}\in \RMatchGT{\bT}{\widetilde{R}_2}{\widetilde{R}_1}\,,\; 
\tilde{m}_{21}\in \MatchGT{\bT}{\compGT{\bT}{\widetilde{R}_2}{\tilde{\mu}}{\widetilde{R}_1}}{\overline{X}}\}\\
\overline{\cB} &:= 
\{ (\bar{\mu},\bar{m}_{21}) \mid 
\bar{\mu}\in \overline{\RMatchGT{\bT}{\overline{R}_2}{\overline{R}_1}}\,,\; 
\bar{m}_{21}\in \MatchGT{\bT}{\overline{\compGT{\bT}{\overline{R}_2}{\bar{\mu}}{\overline{R}_1}}}{\overline{X}}\}
\end{aligned}
\end{equation}
such that for all $(\tilde{\mu},\tilde{m}_{21})\in \widetilde{\cB}$, letting $(\bar{\mu},\bar{m}_{21})=\beta(\tilde{\mu},\tilde{m}_{21})\in \overline{\cB}$,
\begin{equation}
(\compGT{\bT}{\widetilde{R}_2}{\tilde{\mu}}{\widetilde{R}_1})_{\tilde{m}_{21}}(\overline{X})
\cong (\overline{\compGT{\bT}{\overline{R}_2}{\bar{\mu}}{\overline{R}_1}})_{\bar{m}_{21}}(\overline{X})\,.
\end{equation}
\item There exists an \emph{isomorphism} $\gamma: \overline{\cA} \rightarrow \overline{\cB}$ such that for all $(\tilde{m}_1,\tilde{m}_2)\in \widetilde{\cA}$, letting $(\bar{\mu},\bar{m}_{21}) = \gamma(\tilde{m}_1,\tilde{m}_2)$,
\begin{equation}
\overline{R}_{2_{\bar{m}_2}}\left(\overline{R}_{1_{\bar{m}_1}}(\overline{X})\right) \cong(\overline{\compGT{\bT}{\overline{R}_2}{\bar{\mu}}{\overline{R}_1}})_{\bar{m}_{21}}(\overline{X})\,.
\end{equation}
\end{enumerate}
\end{theorem}
\begin{proof}
The central step of the proof consists in proving statement $(i)$, which in combination with the $\bT$-type concurrency theorems (Theorem~\ref{thm:concur}) then permits to derive $(ii-iv)$. Let thus $\ac{c}_{\widetilde{I_{21}}}$ and $\ac{c}_{\overline{I_{21}}}$ denote the application conditions for the compositions of the rules $\widetilde{R}_2$ with $\widetilde{R}_1$ along $\tilde{\mu}$, and of the rules $\overline{R}_2$ with $\overline{R}_1$ along $\tilde{\mu}$, respectively. According to the definition of $\bT$-type rule compositions as provided in~\eqref{eq:defRcomp} and~\eqref{eq:acRcomp},
\begin{align*}
\ac{c}_{\widetilde{I_{21}}}
&= \tilde{\ac{c}}^{(1)}\land \tilde{\ac{c}}^{(2)}
\,,\; & 
\tilde{\ac{c}}^{(1)}
&= \Shift(I_{21}\hookleftarrow I_1,\widetilde{\ac{c}_{I_1}})\,,\;
& \tilde{\ac{c}}^{(2)}
&=  \Trans(P_{21}\leftharpoonup I_{21},\Shift(P_{21}\hookleftarrow I_2,\widetilde{\ac{c}_{I_2}}))
\\
\ac{c}_{\overline{I_{21}}}
&= \bar{\ac{c}}^{(1)}\land \bar{\ac{c}}^{(2)}\,,\;
& 
\bar{\ac{c}}^{(1)}
&= \Shift(I_{21}\hookleftarrow I_1,\overline{\ac{c}_{I_1}})
\,\; 
& \bar{\ac{c}}^{(2)}
&=  \Trans(P_{21}\leftharpoonup I_{21},\Shift(P_{21}\hookleftarrow I_2,\overline{\ac{c}_{I_2}}))\,.
\end{align*}
The first step in our proof amounts to applying~\eqref{eq:cGuaV2} of Lemma~\ref{lem:cgcp}, whereby
\begin{equation*}
\widetilde{\ac{c}_{I_1}}  = \overline{\ac{c}_{I_1}} \land \Shift(I_1\hookleftarrow \mIO,\ac{c}_{\mIO})\,,\quad
\widetilde{\ac{c}_{I_2}}  = \overline{\ac{c}_{I_2}} \land \Shift(I_2\hookleftarrow \mIO,\ac{c}_{\mIO})\,.
\end{equation*}
Combining this result with the compositionality property of $\Shift$ (Theorem~\ref{thm:STcompComp}), indicated below as {\color{blue}$\mathsf{CS}$}, we may convert the contribution $\tilde{\ac{c}}^{(1)}$ into the form
\begin{equation}\label{eq:c1tildeA}
\begin{aligned}
\tilde{\ac{c}}^{(1)}&=
\Shift(I_{21}\hookleftarrow I_1,\overline{\ac{c}_{I_1}} \land \Shift(I_1\hookleftarrow \mIO,\ac{c}_{\mIO}))
=\bar{\ac{c}}^{(1)} \land \Shift(I_{21}\hookleftarrow I_1,\Shift(I_1\hookleftarrow \mIO,\ac{c}_{\mIO}))\\
&\equiv \bar{\ac{c}}^{(1)} \land \Shift(I_{21}\hookleftarrow \mIO,\ac{c}_{\mIO}) & \text{(by {\color{blue}$\mathsf{CS}$})}\,.
\end{aligned}
\end{equation}
Proceeding analogously for the contribution $\tilde{\ac{c}}^{(2)}$, we find the intermediate result
\begin{align*}
\tilde{\ac{c}}^{(2)}&=
 \Trans(P_{21}\leftharpoonup I_{21},\Shift(P_{21}\hookleftarrow I_2,\overline{\ac{c}_{I_2}} \land \Shift(I_2\hookleftarrow \mIO,\ac{c}_{\mIO})))\\
 &\equiv \bar{\ac{c}}^{(2)} \land 
 \Trans(P_{21}\leftharpoonup I_{21},\Shift(P_{21}\hookleftarrow \mIO,\ac{c}_{\mIO}))\,.
\end{align*}
Upon closer inspection of the commutative diagram~\eqref{eq:defRcomp} that is part of the definition of the $\bT$-type rule composition operation, we find that $(P_{21}\hookleftarrow M_{21}) = (P_{21}\hookleftarrow O_1)\circ (O_1\hookleftarrow M_{21})$, whence pre-composing with the (unique) initial morphism $(M_{21}\hookleftarrow \mIO)\in \cM$, this yields $(P_{21}\hookleftarrow \mIO)= (P_{21}\hookleftarrow O_1)\circ (O_1\hookleftarrow \mIO)$. Moreover, the diagram in~\eqref{eq:defRcomp} encodes that $(I_{21}\hookleftarrow I_1)$ is a DPO-admissible (and thus also SqPO-admissible) match for $r_1$ into the composite input interface $I_{21}$, and with $(P_{21}\leftharpoonup I_{21})$ the direct derivation of $I_{21}$ with rule $r_1$ along match $(I_{21}\hookleftarrow I_1)$. Therefore, the compatibility property of $\Shift$ and $\Trans$ according to Theorem~\ref{thm:STcompComp} (marked below as ${\color{green}\mathsf{CST}}$) is applicable in the following form:
\begin{align*}
\Trans(P_{21}\leftharpoonup I_{21},\Shift(P_{21}\hookleftarrow \mIO,\ac{c}_{\mIO}))
&\equiv \Trans(P_{21}\leftharpoonup I_{21},\Shift(P_{21}\hookleftarrow O_1,\Shift(O_1\hookleftarrow \mIO,\ac{c}_{\mIO})))\quad &\text{(by {\color{blue}$\mathsf{CS}$})}\\
&\,\dot{\equiv}\,
\Shift(I_{21}\hookleftarrow I_1,\Trans(r_1, \Shift(O_1\hookleftarrow \mIO,\ac{c}_{\mIO})))
&\text{(by {\color{green}$\mathsf{CST}$})}
\end{align*}
To proceed, it is necessary to identify the above contribution as a subformula in $\bar{\ac{c}}^{(1)}\land \Shift(I_{21}\hookleftarrow \mIO,\ac{c}_{\mIO})$, which may be achieved via inserting the explicit formula for the constraint-preserving completion $\overline{\ac{c}_{I_1}}$ as provided in\eqref{eq:lemcgcp} of Lemma~\ref{lem:cgcp}:
\begin{align*}
\bar{\ac{c}}^{(1)} &= \Shift(I_{21}\hookleftarrow I_1,\overline{\ac{c}_{I_1}})
= \Shift(I_{21}\hookleftarrow I_1, (\neg \Shift(I_1\hookleftarrow \mIO,\ac{c}_{\mIO}))\lor \widetilde{\ac{c}_{I_1}})\\
&= (\neg \Shift(I_{21}\hookleftarrow \mIO,\ac{c}_{\mIO}))
 \lor \bigg(
\Shift(I_{21}\hookleftarrow I_1, \ac{c}_{I_1})\\
&\qquad \land 
\Shift(I_{21}\hookleftarrow I_1, \Trans(r_1,\Shift(O_1\hookleftarrow \mIO,\ac{c}_{\mIO})))
\land
\Shift(I_{21}\hookleftarrow I_1, \Shift(I_1\hookleftarrow \mIO,\ac{c}_{\mIO})
\bigg)
\end{align*}
Inserting this intermediate result into the formula for $\tilde{\ac{c}}^{(1)}$ according to~\eqref{eq:c1tildeA}, and utilizing the logical tautology~\eqref{eq:usefulTautology}, we find:
\begin{align*}
\tilde{\ac{c}}^{(1)}
& = \bar{\ac{c}}^{(1)} \land \Shift(I_{21}\hookleftarrow \mIO,\ac{c}_{\mIO})\\
&\equiv
\Shift(I_{21}\hookleftarrow I_1, \ac{c}_{I_1})\land 
\Shift(I_{21}\hookleftarrow I_1, \Trans(r_1,\Shift(O_1\hookleftarrow \mIO,\ac{c}_{\mIO})))
\land
\Shift(I_{21}\hookleftarrow \mIO,
\ac{c}_{\mIO}) &\text{(by {\color{blue}$\mathsf{CS}$})}
\end{align*}
Note in particular that the middle term in the last line is precisely the second term in $\tilde{\ac{c}}^{(2)}$, which is why we obtain as the intermediate result
\begin{align*}
\ac{c}_{\widetilde{I_{21}}}
&= \tilde{\ac{c}}^{(1)} \land \tilde{\ac{c}}^{(2)}
\,\dot{\equiv}\, \bar{\ac{c}}^{(1)}\land \bar{\ac{c}}^{(2)} \land \Shift(I_{21}\hookleftarrow \mIO,\ac{c}_{\mIO})
= \ac{c}_{\overline{I_{21}}} \land \Shift(I_{21}\hookleftarrow \mIO,\ac{c}_{\mIO})\,.
\end{align*}
Finally, via the tautology~\eqref{eq:usefulTautology} and by definition of the constraint-preserving completion according to Lemma~\ref{lem:cgcp}, we may prove statement $(i)$:
\begin{align*}
\ac{c}_{\widetilde{I_{21}}}
&= (\Shift(I_{21}\hookleftarrow \mIO,\ac{c}_{\mIO})\Rightarrow\ac{c}_{\overline{I_{21}}} )\land \Shift(I_{21}\hookleftarrow \mIO,\ac{c}_{\mIO})
=\overline{\ac{c}_{\overline{I_{21}}}} \land \Shift(I_{21}\hookleftarrow \mIO,\ac{c}_{\mIO})\,.
\end{align*} 
For the proof of the statements $(ii-iv)$, it is sufficient to specialize of the general $\bT$-type concurrency theorems as in Theorem~\ref{thm:concur} to the setting of rewriting objects $\overline{X}\in \obj{\bfC}$ with $\overline{X}\vDash \ac{c}_{\mIO}$, and for rules with conditions $\widetilde{R}_2$ and $\widetilde{R}_1$ (i.e., for the constraint-guaranteeing completions of some generic rules $R_2,R_1\in \LinAc{\bfC}$) and their constraint-preserving variants $\overline{R}_2$ and $\overline{R}_1$, respectively. For statement $(ii)$, it is evident that $\widetilde{R}_1$ and $\overline{R}_1$ have precisely the same number of admissible matches into the constraint-satisfying object $\overline{X}$. By definition of constraint-preserving completions, the application of $\overline{R}_1$ to $\overline{X}$ along any admissible match $\bar{m}_1$ results in an object that again satisfies the global constraint, i.e., $\overline{R}_{1_{\bar{m}_1}}(\overline{X})\vDash\ac{c}_{\mIO}$, thus one may repeat the preceding argument to show that there exists an isomorphism between the sets of two-step direct derivations along $\overline{R}_1$ followed by $\overline{R}_2$, and along $\widetilde{R}_1$ followed by $\widetilde{R}_2$, respectively (and thus also an isomorphism $\alpha$ between the sets of pairs of admissible matches as claimed).  The proof of statement $(iii)$ is completely analogous, with the only additionally necessary technical detail concerning the definition of the restricted set of admissible matches $\overline{\RMatchGT{\bT}{\overline{R}_2}{\overline{R}_1}}$: this restriction is necessary in order to obtain the isomorphism $\beta$, since the constraint-preserving rules $\overline{R}_2$ and $\overline{R}_1$ may in general have admissible matches $\bar{\mu}$ which result in $\Shift(I_{21}\hookleftarrow \mIO,\ac{c}_{\mIO})=\ac{false}$, in contrast to compositions of the constraint-guaranteeing rules $\widetilde{R}_2$ with $\widetilde{R}_1$. However, rules with $\Shift(I_{21}\hookleftarrow \mIO,\ac{c}_{\mIO})=\ac{false}$ do not possess $\bT$-admissible matches into a constraint-satisfying object $\overline{X}$, which is why the restriction to matches in $\overline{\RMatchGT{\bT}{\overline{R}_2}{\overline{R}_1}}$ precisely reproduces those rule compositions that are computable equivalently as compositions of $\widetilde{R}_2$ with $\widetilde{R}_1$. Finally, to prove statement $(iv)$, it suffices to combine $(ii)$ and $(iii)$ with the generic $\bT$-type concurrency theorems.
\end{proof}

In summary, we have obtained modified versions of the $\bT$-type concurrency theorems that permit us to work throughout all of our computations in the setting of rewriting of objects under global constraints with the most restricted version of application conditions for rules, i.e., in a sense the \emph{most minimal possible} such constraints. Importantly, by virtue of Theorem~\ref{thm:rcct}, this restriction is indeed ``compositional'', in that composing two rules with constraint-guaranteeing application conditions (and thus typically enormously complicated such conditions) may be avoided, namely by instead computing the constraint-preserving variants of the rules, reducing each contribution to the set of sequential composition to just constraint-preserving completions as well. An immediate consequence of this line of reasoning are the following definition and theorem, quintessential for the development of efficient rule-algebraic algorithms in particular in the settings of  bio- and organo-chemical rewriting of Sections\ref{sec:bcgr} and~\ref{sec:ocgr}:

\begin{definition}[Restricted rule algebras and representations]\label{def:RRA}
For an $\cM$-adhesive category $\bfC\in \cM-\mathbf{CAT}_{\bT}$ suitable for type $\bT\in \{DPO,SqPO\}$ rewriting, and for a \emph{global constraint} $\ac{c}_{\mIO}\in \cond{\bfC}$, define the \emph{equivalence relation} $\overline{\sim}$ on $\LinAc{\bfC}$ via
\begin{equation}
\forall (r,\ac{c}_I), (r',\ac{c}_{I'})\in \LinAc{\bfC}:\quad
(r,\ac{c}_I) \overline{\sim} (r',\ac{c}_{I'})\quad
:\Leftrightarrow\quad  \left(r\cong r' \land \overline{\ac{c}_I}\,\dot{\equiv}\, \overline{\ac{c}_{I'}}\right)\,.
\end{equation}
Letting $\LinAc{\bfC}_{\overline{\sim}}$ denote the set\footnote{As with the standard definition of rule algebras, should for the chosen category $\bfC$  and global constraint $\ac{c}_{\mIO}$ the quotient of the class $\LinAc{\bfC}$ by the equivalence relation $\overline{\sim}$ not yield a set, but again a proper class, one must restrict from all constraint-preserving rule equivalence classes to a countable set (which in practice is typically generated by a countable number of pattern observables and a finite number of rewriting rules in the transition operators).} of equivalence classes of $\LinAc{\bfC}$ under $\overline{\sim}$, and for a field $\bK$ of characteristic $0$ (such as $\bK=\bR$ or $\bK=\bC$), define the \emph{$\bK$-vector space} $\overline{\cR}_{\bfC}$ via a bijection $\bar{\delta}:\LinAc{\bfC}_{\overline{\sim}}\xrightarrow{\cong} \mathsf{basis}(\overline{\cR}_{\bfC})$ from the set of $\overline{\sim}$-equivalence classes of linear rules to the set of basis elements of $\overline{\cR}_{\bfC}$. %
Introduce the \emph{$\bT$-type restricted rule algebra products} $\rrap{\bT}{}{}: \overline{\cR}_{\bT}\times \overline{\cR}_{\bT}\rightarrow \overline{\cR}_{\bT}$ as the binary operations defined via their action on basis elements,
\begin{equation}\label{eq:rrapdef}
 \rrap{\bT}{\bar{\delta}(R_2)}{\bar{\delta}(R_1)} := \sum_{\bar{\mu}\in\overline{\RMatchGT{\bT}{R_2}{R_1}}} \bar{\delta}\left(
 \compGT{\bT}{R_2}{\bar{\mu}}{R_1}
 \right)\,.
\end{equation}
Then the \emph{$\bT$-type restricted rule algebras} are defined as $\overline{\cR}^{\bT}_{\bfC}\equiv (\overline{\cR}_{\bfC}, \rrap{\bT}{}{})$.

Let the $\bK$-vector space $\hat{\bar{\bfC}}$ be defined via a bijection $\ket{.}:\obj{\bar{\bfC}}_{\cong}\xrightarrow{\cong} \mathsf{basis}(\hat{\bar{\bfC}})$ from the set $\obj{\bar{\bfC}}_{\cong}$ of isomorphism classes of objects of $\bfC$ satisfying the constraint $\ac{c}_{\mIO}$ into the set of basis elements of $\hat{\bar{\bfC}}$, with
\begin{equation}
X\in \obj{\bar{\bfC}} :\Leftrightarrow X\in \obj{\bfC}\land X\vDash \ac{c}_{\mIO}\,.
\end{equation}
Then we define the \emph{$\bT$-type restricted representations} as the homomorphisms $\bar{\rho}^{\bT}_{\bfC}: \overline{\cR}^{\bT}_{\bfC}\rightarrow End_{\bK}(\hat{\bar{\bfC}})$, fully specified via their action on basis vectors according to
\begin{equation}\label{eq:defRRAR}
\bar{\rho}^{\bT}_{\bfC}(\bar{\delta}(R))\ket{\overline{X}} := \sum_{\bar{m}\in \MatchGT{\bT}{\overline{R}}{\overline{X}}} \ket{\overline{R}_{\bar{m}}(\overline{X})}\,.
\end{equation}
\end{definition}

\begin{remark}
It is worthwhile to note that the sum in~\eqref{eq:defRRAR} ranges over admissible matches of the constraint-preserving completion $\overline{R}$ of $R$ (which due to $\overline{R}\overline{\sim} R$ is evidently in the same $\overline{\sim}$-equivalence class), which ensures that the right-hand side is indeed a vector in $\hat{\bar{\bfC}}_{\bT}$. Moreover, at various points in our definitions, we have chosen notations that keep the dependence on the choice of the constraint $\ac{c}_{\mIO}$ implicit for succinctness, since the choice of $\ac{c}_{\mIO}$ is in practice taken ``globally'', i.e., as part of the input data of a given set of computations and kept fixed throughout all computations. Another notational simplification taken in the above definition concerns the symbol $\bar{\delta}(R)$ for ``restricted'' basis states, which through the dependency on the equivalence relation $\overline{\sim}$ strictly speaking of course also depends on the style $\bT$ of rewriting; however, in all computations, this type will be clear from the context, so we chose to omit the annotation of restricted states by $\bT$.
\end{remark}

\begin{theorem}\label{thm:rrap}
With notations as in Definition~\ref{def:RRA}, and letting $\cR^{\bT}_{\bfC} \equiv (\cR_{\bfC}, \rap{\bT}{}{})$ and $\rho^{\bT}_{\bfC}$ denote the $\bT$-type rule algebras and their representations in the \emph{unrestricted} setting (compare Definition~\ref{def:RA} and~\ref{def:RAR}), we find:
\begin{enumerate}[label=(\roman*)]
\item The algebras $\bar{\cR}^{\bT}_{\bfC}$ are \emph{associative unital algebras}, with unit elements $\bar{\delta}(R_{\mIO})=\bar{\delta}(\mIO\hookleftarrow \mIO,\ac{true}))$. 
\item The homomorphisms $\bar{\rho}^{\bT}_{\bfC}$ are \emph{algebra homomorphisms} (and as such qualify as representations of the algebras $\bar{\cR}^{\bT}_{\bfC}$), or, equivalently, for all $R_1,R_2\in \LinAc{\bfC}_{\overline{\sim}}$, 
\begin{equation}\label{eq:rrapRhoH}
(a)\quad \bar{\rho}^{\bT}_{\bfC}(\bar{\delta}(R_2))\bar{\rho}^{\bT}_{\bfC}(\bar{\delta}(R_1)) = \bar{\rho}^{\bT}_{\bfC}\left(
\rrap{\bT}{\bar{\delta}(R_2)}{\bar{\delta}(R_1)}
\right)\,,\quad (b)\quad \bar{\rho}^{\bT}_{\bfC}(\bar{\delta}(R_{\mIO})) =Id_{End_{\bK}(\hat{\bar{\bfC}})}
\end{equation}
\item For arbitrary rules with conditions $R_1,R_2\in \LinAc{\bfC}$, letting $\widetilde{R}_j$ and $\overline{R}_j$ (for $j=1,2$) denote the constraint-guaranteeing and constraint-preserving completions of the two rules (for the global constraint $\ac{c}_{\mIO}$), and for an arbitrary restricted state $\ket{\overline{X}}\in \hat{\bar{\bfC}}$, the following equalities hold:
\begin{equation}
\begin{aligned}
\rho^{\bT}_{\bfC}(\delta(\widetilde{R}_2))\rho^{\bT}_{\bfC}(\delta(\widetilde{R}_1))\ket{\overline{X}}
&=\rho^{\bT}_{\bfC}\left(\delta\left(\rap{\bT}{\widetilde{R}_2}{\widetilde{R}_1}\right)\right)\ket{\overline{X}}\\
&=\bar{\rho}^{\bT}_{\bfC}\left(\bar{\delta}\left(\rrap{\bT}{\overline{R}_2}{\overline{R}_1}\right)\right)\ket{\overline{X}}
=\bar{\rho}^{\bT}_{\bfC}(\bar{\delta}(\overline{R}_2))\bar{\rho}^{\bT}_{\bfC}(\bar{\delta}(\overline{R}_1))\ket{\overline{X}}\,.
\end{aligned}
\end{equation}
\end{enumerate}
\end{theorem}
\begin{proof}
The proof follows from combining Theorem~\ref{thm:rcct} with the results presented in Sections~\ref{sec:ra} and~\ref{sec:stochMech} and in~\ref{app:ACthms} for the general setting of rewriting with conditions (i.e., Theorems~\ref{thm:RAmain}, \ref{thm:canrep}, \ref{thm:CTMCs}, \ref{thm:concur} and \ref{thm:assocR}).
\end{proof}

One of the main applications of the notion of restricted rule algebras in view of the main theme of the present paper is the formulation of CTMCs for restricted stochastic rewriting systems:
\begin{corollary}\label{cor:CTMCrrt}
Let $\bfC\in \cM-\mathbf{CAT}_{\bT}$ be an $\cM$-adhesive category suitable for type $\bT\in \{DPO,SqPO\}$ rewriting, and let $\ac{c}_{\mIO}\in \cond{\bfC}$ be a \emph{global constraint}. For $R\in \Lin{\bfC}$, let $\bar{\delta}(R):=\bar{\delta}([R]_{\overline{\sim}})$. Denote by $\hat{\bar{\bfC}}$ the sub-$\bR$-vector space of the $\bR$-vector space $\hat{\bfC}$ spanned by states $\ket{\overline{X}}$ with  $\overline{X}\vDash \ac{c}_{\mIO}$, and let $\bra{}:\hat{\bar{\bfC}}\rightarrow \bR$ be the functional on $\hat{\bar{\bfC}}$ defined via $\braket{}{\overline{X}}:= 1_{\bR}$.
\begin{enumerate}[label=(\roman*)]
\item \textbf{Restricted DPO-type jump-closure:} for any  $\ket{\overline{X}}\in \hat{\bar{\bfC}}$ and $R\equiv(O\hookleftarrow K\hookrightarrow I,\ac{c}_{I})\in \LinAc{\bfC}$,
\begin{equation}
\bra{}\bar{\rho}_{DPO}(\bar{\delta}(R))\ket{\overline{X}}
=\bra{}\overline{\bO}_{DPO}(\bar{\delta}(R))\ket{\overline{X}}\,,\;
\overline{\bO}_{DPO}(\bar{\delta}(O\hookleftarrow K\hookrightarrow I,\ac{c}_{I})):=
\bar{\delta}(I\hookleftarrow K\hookrightarrow I,\ac{c}_{I})\,.
\end{equation}
\item \textbf{Restricted SqPO-type jump-closure:} for any  $\ket{\overline{X}}\in \hat{\bar{\bfC}}$ and $R\equiv(O\hookleftarrow K\hookrightarrow I,\ac{c}_{I})\in \LinAc{\bfC}$,
\begin{equation}
\bra{}\bar{\rho}_{SqPO}(\bar{\delta}(R))\ket{\overline{X}}
=\bra{}\overline{\bO}_{SqPO}(\bar{\delta}(R))\ket{\overline{X}}\,,\;
\overline{\bO}_{SqPO}(\bar{\delta}(O\hookleftarrow K\hookrightarrow I,\ac{c}_{I})):=
\bar{\delta}(I\hookleftarrow I\hookrightarrow I,\ac{c}_{I})\,.
\end{equation}
\item \textbf{CTMCs based upon restricted rewriting:} given a finite set of pairs of base rates and rules with conditions $\cT:=\{(\kappa_j,R_j)\}_{j=1}^n$ and an input state $\ket{\overline{\Psi}_0}\in \Prob{\bar{\bfC}}$, this data defines a CTMC as follows:
\begin{equation}
\tfrac{d}{dt}\ket{\overline{\Psi}(t)}=\overline{H}\ket{\overline{\Psi}(t)}\,,\;
\ket{\overline{\Psi}(0)}=\ket{\overline{\Psi}_0}\,,\; \overline{H}:= \hat{\overline{H}}- \bar{\bO}_{\bT}\left(\hat{\overline{H}}\right)\,,\;
\hat{\overline{H}}:= \sum_{j=1}^n \kappa_j \bar{\rho}_{\bT}\left(\bar{\delta}(R_j)\right)\,.
\end{equation}
\end{enumerate}
\end{corollary}

In summary, for stochastic rewriting systems restricted via the choice of a global constraint on objects, the notion of constraint-preserving completions lends itself naturally to (often significantly) reduce the complexity in formulating possible pattern-observables arising in the computations of moment-evolution equations. In the next section, we will provide a class of such restricted rewriting theories to exemplify the utility of this notion.

\subsection{Ambient, pattern and state categories}\label{sec:aps}

Inspired by concepts from the \KAP{} framework, and indeed closely following~\cite{Danos2014}, a versatile class of restricted rewriting theories may be characterized by a special structure of the global constraints chosen to define the theories:

\begin{definition}
Let $\bA \in \cM-\mathbf{CAT}_{\bT}$ denote an $\cM$-adhesive category suitable for type $\bT\in \{DPO,SqPO\}$ rewriting, henceforth referred to as the \emph{ambient category}. Define the \emph{negative constraint} $\ac{c}_{-}$ and the \emph{positive constraint} $\ac{c}_{+}$ as
\begin{equation}
\ac{c}_{-} :=\bigwedge_{N\in \cN} \neg \exists (N)\,,\quad
\ac{c}_{+} := \bigwedge_{P\in \cP} \forall\left(P, \bigvee_{(p:P\hookrightarrow Q)\in \cQ_P} \exists(p)\right)\,,
\end{equation}
where $\cN\subset \obj{\bA}_{\cong}$ is a finite set of \emph{forbidden patterns}, where $\cP$ is defined as a finite set of patterns $P\in \obj{\bA}_{\cong}$ with $P\vDash\ac{c}_{-}$, and where for each $P\in \cP$, $\cQ_P$ is a finite set of $\cM$-morphisms $(p:P\hookrightarrow Q)$ with $Q\vDash\ac{c}_{-}$. We then define the \emph{pattern category} $\bP$ and the \emph{state category} $\bS$ via restrictions of $\bA$ as follows:
\begin{equation}
X \in \obj{\bP} :\Leftrightarrow X\in \obj{\bA}\land X\vDash \ac{c}_{-}\,,\quad
Y \in \obj{\bS} :\Leftrightarrow Y\in \obj{\bA}\land Y\vDash \ac{c}_{-}\land Y\vDash \ac{c}_{+}\,.
\end{equation}
For some applications (and in particular for bio- and organo-chemical rewriting), we assume in addition that all patterns $N,P,Q$ occurring in the definitions of constraints are \emph{connected}.
\end{definition}

A first consequence of these definitions~\cite{Danos2014} is that by virtue of the $\Shift$ construction,
\begin{equation}
 \forall (f:X\hookrightarrow Y)\in \cM: \quad Y\vDash \ac{c}_{-} \; \Rightarrow \; X\vDash \ac{c}_{-}\,,
\end{equation}
so that in particular the pattern category $\bP$ is closed under subobjects. Combined with the definition of constraint-guaranteeing and -preserving completions of application conditions (Lemma~\ref{lem:cgcp}), we thus find that rewriting rules that have completions of their application conditions that do not evaluate to $\ac{false}$ (i.e., rules that can act non-trivially on constrained objects) are required to be defined in terms of \emph{patterns} rather than arbitrary elements of the ambient category $\bA$. Regarding the role of the positive constraints $\ac{c}_{+}$, it will prove useful to note the following auxiliary result:

\begin{corollary}
Let $\overline{R}_j \equiv [(r_j, \overline{\ac{c}_{I_j}})]_{\overline{\sim}} \in \LinAc{\bA}_{\overline{\sim}}$ (for $j=1,2$) denote constraint-preserving completions of linear rules with respect to the global constraint $\ac{c}_{-}\land \ac{c}_{+}$ for type $\bT$ rewriting, and suppose that $\overline{\ac{c}_{I_1}}^{+}\,\dot{\equiv}\,\overline{\ac{c}_{I_2}}^{+}\,\dot{\equiv}\, \ac{true}$ (where $\overline{\ac{c}_{I_j}}^{+}$ denotes the contribution to the constraint-preserving application condition arising from $\ac{c}_{+}$). Then for arbitrary $\bT$-admissible matches $\mu\in \MatchGT{\bT}{\overline{R}_2}{\overline{R_1}}$, the constraint-preserving completion of the composite rule satisfies $\overline{\ac{c}_{\bar{I}_{21}}}^{+}\, \dot{\equiv}\,\ac{true}$. 
\end{corollary}
\begin{proof}
The statement follows via combining Lemma~\ref{lem:cgcp} with Theorem~\ref{thm:rcct}.
\end{proof}

\section{Application scenario 1: biochemistry with \KAP{}}\label{sec:bcgr}

The \href{https://kappalanguage.org}{\KAP{} platform}~\cite{danos2004computational,danos2004formal} for rule-based modeling of biochemical reaction systems is based upon the notion of \emph{site-graphs} that abstract proteins and other complex macro-molecules into \emph{agents} (with \emph{sites} representing interaction capacities of the molecules). This open source platform offers a variety of high-performance \emph{simulation algorithms} (for CTMCs based upon \KAP{} rewriting rules) as well as several variants of static analysis tools to analyze and verify biochemical models~\cite{Boutillier:2018aa}. %
Since the start of the \KAP{} development, the simulation-based algorithms have been augmented by \emph{differential semantics} modules aimed at deriving ODE systems for the evolution of pattern-count observable average values, based upon ideas from abstract interpretation~\cite{Feret_2008,Danos_2010,Harmer_2010}. 
Differential semantics for a given set of \KAP{} rules relies on the computation of a set of \KAP{} graphs, called fragments, that is closed under the action of the rules and therefore amenable to ODE representation. %

Interestingly, the computation of the fragments is purely syntactical, and based on a static analysis of rule interference. In particular, albeit a rewriting-based graphical formalism, \KAP{} as originally introduced in~\cite{danos2004formal} is \emph{not} a categorical rewriting formalism, although categorical approaches have been employed to model certain aspects of its semantics~\cite{danos_et_al:LIPIcs:2012:3866}. %
This renders comparing this approach with the computation of commutators presented in this paper highly intricate.%
On the other hand, an interesting line of work by Danos et al.~\cite{danos2015moment,danos2020} demonstrated that at least in the less technically involved setting of rewriting over adhesive categories without conditions or constraints, some of the syntactic methods utilized in the \KAP{} framework could be reinterpreted in the adhesive rewriting setting in order to obtain algorithms for computing first-order moment ODEs for pattern counting observables\footnote{%
In retrospective, one may understand the results of loc. cit.\ as a special case of our universal rule-algebraic CTMC theory, providing in essence a syntactic variant (inspired by the \KAP{} calculus) of the definition of the linear operators $\rho^{\bT}_{\bfC}(\delta(r))$ (for $\bT\in \{DPO,SqPO\}$, $\bfC$ an adhesive category and $r\in \Lin{\bfC}$) via a prescription for the matrix elements of $\rho^{\bT}_{\bfC}(\delta(r))$. Utilizing our framework, one can indeed rigorously validate the correctness of the interpretation of these operators postulated in loc cit.\ (i.e., the representation property of Theorem~\ref{thm:canrep}, which hinges upon Theorems~\ref{thm:RAmain} and~\ref{thm:concur}), and verify in particular that the approach provides an algorithm for computing first-order moment ODEs that is equivalent to the relevant special case of Theorem~\ref{thm:CTMCmev}.}, providing a first hint at the possible existence of a general rewriting-based CTMC formalism for \KAP{}.

In this section, we will show that it is possible to retrieve the expressiveness of \KAP{} site-graphs while remaining in a suitable $\cM$-adhesive category of typed undirected multigraphs. Patterns and states may be formulated consistently via certain negative and positive structural constraints (following the general construction presented in Section~\ref{sec:aps}), so that ultimately we are able to derive a full-fledged CTMC theory for \KAP{} directly from our universal framework based upon rule algebras for restricted rewriting theories (Section~\ref{sec:rrt}). Notably, this formulation is not only fully equivalent to the aforementioned \KAP{} notion of differential semantics (for the averages of pattern-counts), but in addition provides a fully general computational theory for deriving higher-order moment ODEs for \KAP{} pattern-counting observables. 

\subsection{The \KAP{} framework for biochemical reaction systems analysis}

We will begin our presentation with a brief introduction to the \KAP{} formalism, referring the interested readers to~\cite{Kappa_manual} for further details. One of the key practical features of \KAP{} is its foundation upon the notion of \emph{rigidity}~\cite{Danos2014}. In a nutshell, the property that partial embeddings of connected \KAP{} graphs extend to at most one complete embedding. This ensures that subgraph isomorphism checks can be computed efficiently, a key property since \KAP{} graphs are constantly matched against a potentially very large set of rewriting rules during stochastic simulations \cite{Boutillier17}. \KAP{} graphs belong to the family of port graphs \cite{ANDREI200867}, a particular graph formalism in which connections between nodes are made through ports, called \emph{sites} in the context of \KAP{}. We recall here the definition of \KAP{} graphs, largely following a simplification of the presentation in~\cite[Sec. 3]{danos_et_al:LIPIcs:2012:3866}.

\begin{definition}[Signature]\label{def:KS}
 A \KAP{} signature is a tuple $\Sigma=(\Ag,\St,\AgSt,\Prp)$ where $\Ag=\set{A,B,\dots}$ is a countable set of \emph{agent} types, $\St=\set{i,j,k,\dots}$ a countable set of \emph{site} types, $\AgSt:\Ag\to\cP_{\sf fin}(\St)$ maps the agent types to the sites they possess, and $\Prp=\set{\sf p,q,r,\dots}$ is a countable set of \emph{properties}. We define the shorthand notation $(s:A) :\Leftrightarrow (s\in \AgSt(A))$ (for $A\in \Ag$ and $s\in \St$) to indicate that a site type $s$ is present on agents of type $A$. Without loss of expressivity, we will assume purely for technical convenience that a given site type is specific to precisely one agent type, i.e.,
\begin{equation}
\forall s\in \St, A,B\in \Ag:\; (s:A \land s:B) \Rightarrow A=B\,.
\end{equation}
\end{definition}

In a biological modeling context, \KAP{} signatures are used to map protein interactions to a port graph encoding: agent types usually represent protein names, site types represent interaction capabilities, and properties refer to post-translational modifications\footnote{Post-translational modifications are chemical tags that can be attached to protein residues and that can influence the spatial configuration of the protein.} such as phosphorylation or methylation.

\begin{definition}[\KAP{} graphs]
    A \KAP{} graph over a signature $\Sigma$ is a tuple $G=(\mc{A},\mc{S},\mc{E},\tAg,p)$ where 
    \begin{enumerate}[label=(\roman*)]
    \item $\mc A=\set{a,b,c,\dots}$ is a countable set of \emph{agents}, 
    \item $\mc S$ a countable set of \emph{sites} satisfying 
    $\mc S\subseteq \set{(a,i)\mid a\in\mc A, i\in\AgSt(\tAg(a))}$,
    \item $\mc E\subseteq \mc S\times \mc S$ is an irreflexive \emph{link relation on sites}, 
    \item $\tAg:\mc A\to\Ag$ assigns types to agents, and
    \item $p:\mc S\pto\Prp$ is a partial map of sites to properties.
\end{enumerate}
\end{definition}

The above definition implies several invariants over the structure of \KAP{} graphs (over a signature $\Sigma$). First, it is not possible to define an agent that would have two sites of type $i$ (because one uses sets and not multisets for $\mc S$). Second, edges of \KAP{} graphs connect sites but cannot connect agents directly and irreflexivity prevents sites to be connected to themselves, i.e. $(a,i),(a,j)\in\mc E$ implies $i\neq j$. Lastly, only sites may have a property, and at most one in a given state.

\begin{figure}
    \centering
    $\vcenter{\hbox{\ti{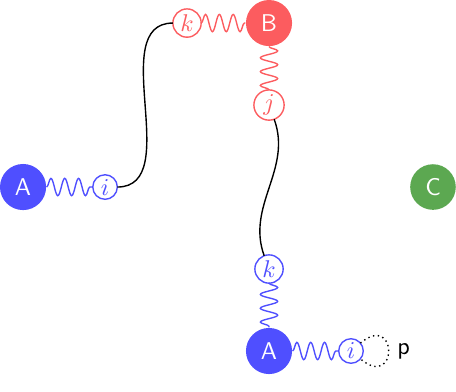}}}\qquad\widehat{=}\qquad$
    $\vcenter{\hbox{\ti{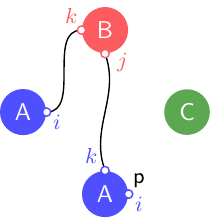}}}$
    \caption{Two variations of graphical notations for \KAP{} graphs: %
    as a typed undirected graph in the sense of Definition~\ref{def:KTGAPS} (left), %
    and in a compressed graphical notation (right) where sites are drawn on the boundaries of the agent vertices to which they are incident, and where properties are drawn next to the site carrying the properties. The latter graphical notation is the one used in the most recent version of the \KAP{} documentation~\cite{Kappa_manual}. Note that the coloring of the site type names indicates that site types are exclusive to one agent type (i.e., two site types with the same name, but different colors are technically different types in our encoding, indicated via different site name colors). In the example presented, two agents of type $A$ are bound to an agent of type $B$ on different sites, and there is moreover an additional disconnected agent of type $C$ present. Note that the $C$ agent has no apparent site, and that site $i$ of the second agent $A$ has property $\sf p$. The other sites of the graph have no property (the $p$ map is undefined for them). One may thus in particular identify this \KAP{} as an example of a \KAP{} \emph{pattern} in the sense of Definition~\ref{def:KTGAPS} .}
    \label{fig:KAPPA-example}
\end{figure}

\begin{definition}
\def\h{h_{\mc A}}
    A \KAP{} graph homomorphism $h:G\to H$ is a total map on agents $\h:\mc A_G\to \mc A_H$ that is edge, type, site and property preserving:
    \begin{itemize}
        \item $(a,i),(b,j)\in \mc E_G$ implies $(\h(a),i),(\h(b),j)\in \mc E_H$
        \item $\tAg(a)=\tAg(\h(a))$
        \item $(a,i)\in \mc S_G$ implies $(\h(a),i)\in \mc S_H$
        \item $p_G$ defined on $(a,i)$ implies $p_H$ defined on $(\h(a),i)$ and $p_G(a,i)=p_H(\h(a),i)$.
    \end{itemize}
\end{definition}

Finally, in order to formulate rewriting systems within \KAP{}, we require the following additional notion:
\begin{definition}[\KAP{} models]\label{def:KM}
A \emph{\KAP{} model} $K:=(\Sigma,\StSt,\StP)$ is as an \emph{extended signature}, where $\Sigma=(\Ag,\St,\AgSt,\Prp)$ is a \KAP{} signature, $\StSt\subset \St\times\St$ is a (symmetric) site-site type incidence relation, and where $\StP:\St\rightarrow \cP_{\mathsf{fin}}(\Prp)$ is a function that assigns to a site type the property types that can be carried by sites of the given type. We introduce the shorthand notations
\begin{equation}
\forall s,s'\in \St:\; (s:s') :\Leftrightarrow (s,s')\in \StSt\,,\quad \forall s\in \St, \mathsf{p}\in \Prp:\; (\mathsf{p}:s) :\Leftrightarrow \mathsf{p}\in \StP(s)\,.
\end{equation}
A further notation concerns a certain form of \emph{compatibility} of an assignment of concrete property types to the site types of some agent type within a given \KAP{} model. Given an agent type $A\in \Ag$ and a partial function $\sigma_{P}:\AgSt(A)\rightarrow \Prp$ from the site types of the agent type $A$ to properties, we define
\begin{equation}
(A,\sigma_P):K\; :\Leftrightarrow (dom(\sigma_P)\setminus dom(\StP) = \varnothing) \land (\forall s \in dom(\sigma_P): \sigma_P(s)\in \StP(s))\,.
\end{equation}
In other words, $(A,\sigma_P):K$ iff the assignment of property types to each of the site types available for a given agent type is consistent with the extended signature $K$.
\end{definition}

\subsection{An equivalent encoding of \KAP{} models as restricted rewriting theories}\label{sec:KMrrt}

We will now proceed with the encapsulation of \KAP{} models into the framework presented in this paper. To this end, we follow the general strategy introduced in Section~\ref{sec:aps}, whereby by starting from a suitable \emph{ambient category} (which is in particular $\cM$-adhesive) as a ``host'' category for the restricted rewriting theory, from which then \emph{pattern} and \emph{state categories} are obtained via restriction with certain types of structural constraints.

\begin{definition}\label{def:KTGAPS}
For a \KAP{} model $K=(\Sigma,\StSt,\StP)$ (with $\Sigma=(\Ag,\St,\AgSt,\Prp)$ a \KAP{} signature), let $\bA_K=\mathbf{uGraph}/T_K$ be the category of finite undirected multigraphs typed over a type graph $T_K$, where $T_K$ is defined as follows:
\begin{enumerate}[label=(\roman*)]
\item For each \emph{agent type} $A\in \Ag$, $T_K$ contains an \emph{agent type vertex} \ti{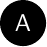}.
\item For each \emph{site type} $s\in \St$, $T_K$ contains a \emph{site type vertex} \ti{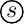}.
\item For each agent type $A\in \Ag$ and site type $s\in \St$ with $s:A$, the corresponding vertices \ti{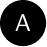} and \ti{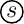} in $T_K$ are linked by an \emph{agent-site incidence type edge} depicted as \ti{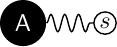}.
\item For each site type $s\in \St$ and property type $\mathsf{p}\in \Prp$ with $\mathsf{p}:s$, $T_K$ contains a \emph{property type loop} depicted as \ti{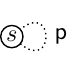}.
\item For each pair of site types $i,j\in \St$ such that $(i,j)\in \StSt$, $T_K$ contains a \emph{site link type edge} depicted as \ti{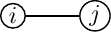}.
\end{enumerate}
In order to formulate the additional structural constraints implied by the \KAP{} formalism, %
let us utilize a graphical notation wherein $\bullet$ is a placeholder for a vertex and a dashed line $\ti{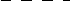}$ for an edge of any type. We may then introduce the \emph{negative constraints} defining the \emph{pattern category} $\bP_K$ as $\ac{c}_{\cN_K}:=\land_{N\in \cN_K}\neg\exists(\mIO\hookrightarrow N)$, with the set $\cN_K$ of \emph{``forbidden patterns''} defined as\footnote{There is a certain degree of freedom in choosing the precise definition of the negative constraints (i.e., depending on the precise variant of \KAP{} considered), yet the version presented here closely reflects the standard \KAP{} implementation at the time of writing (modulo the convenience choice made purely for aesthetic reasons to have each site type be exclusive to precisely one agent type). In particular, sites are restricted to bind to at most one other site.}
\begin{equation}
	\cN_K:=\left\{
	\ti{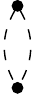}\right\} 
	\cup
	\bigcup\limits_{\substack{s\in \St\\\mathsf{p},\mathsf{q}\in \StP(s)}}\left\{
	\ti{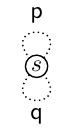}\,,\;\ti{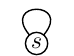}\,,\;\ti{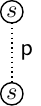}
	\right\}
	\cup
	\bigcup\limits_{\substack{i,j,k\in \St\\ (i,j), (j,k)\in \StSt}}\left\{
	\ti{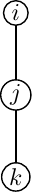}
	\right\}\cup 
	\bigcup\limits_{\substack{A\in \Ag\\ s\in \AgSt(A)}}\left\{
	\ti{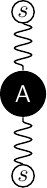}\,,\;
	\ti{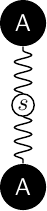}
	\right\}\,.
\end{equation}
Finally, the \emph{state category} $\bS_K$ is obtained from $\bP_K$ via imposing a positive constraint $\ac{c}_{\cP_K}$ that ensures that each agent of type $A$ is linked to exactly one site of type $s:A$ for each of the site types $s\in\AgSt(A)$, and if a site of type $s$ can carry a property or alternative variants thereof (i.e., if $s\in dom(\StP)$), it carries a loop of property type $\mathsf{p}$ for one of the property types $\mathsf{p}\in \StP(s)$. Moreover, for each agent type $A\in \Ag$ and site type\footnote{Recall that by virtue of the technical assumption taken in Definition~\ref{def:KS}, for each site type $s\in \St$ there is precisely one agent type $A\in \Ag$ such that $s:A$.} $s\in \AgSt(A)$, every site of type $s$ must be linked to an agent of type $A$ (i.e., sites cannot occur in isolation). 
\end{definition}

\begin{remark}\label{rem:KGNC}
In order to compress the graphical notations of \KAP{} graphs in the sense of the above encoding back into a more compact graphical notation, we adopt the convention that is standard in the \KAP{} literature whereby site vertices are simply drawn adjacent to the agent vertices they are incident to, and whereby property loops are compressed into just a label with the property type drawn adjacent to the site that carries the given property. We illustrate this convention in Figure~\ref{fig:KAPPA-example}.
\end{remark}

\subsection{\KAP{} rules in the restricted rewriting semantics}\label{sec:KRrrt}

Based upon the definition of a restricted rewriting theory for a given \KAP{} model $K$ according to Definition~\ref{def:KTGAPS}, an interesting question arises in view of formulating \KAP{} rewriting rules in this setting: which types of ``plain'' rules $r\in \Lin{\bA_K}$ in the ambient category $\bA_K$ of the \KAP{} model $K$ can be lifted to restricted rewriting rules $\overline{R}\in \LinAc{\bA_K}$ via the strategy advocated in Section~\ref{sec:rrt}, namely by equipping ``plain'' rules with application conditions via the operation of constraint-preserving completions in SqPO-semantics (starting from $R:=(r,\ac{c}_I)$ with a trivial condition $\ac{c}_I:=\ac{true}$)? Given the complexity of the negative and positive structural constraints $\ac{c}_{\cN_K}$ and $\ac{c}_{\cP_K}$, it seems a priori unclear whether or not such restricted rules would even be feasible to formulate. Nevertheless, since in \KAP{} rewriting systems only rules that preserve the structural constraints are of any practical relevance (i.e., those that transform \KAP{} \emph{states} into \KAP{} \emph{states}), it would be desirable to understand precisely which classes of ``plain'' rules $r\in \Lin{\bA_K}$ fail to posses a non-trivial lift to restricted rules via the aforementioned operation, since lifted rules evidently act trivially on \KAP{} states whenever $\overline{\ac{c}_I} \,\dot{\equiv}\, \ac{false}$.

Postponing a full classification and analysis of the general nature of \KAP{} restricted rewriting rules to future work, we will present here a first important step towards such a theory, namely by formulating a particular subclass of ``plain'' \KAP{} that is guaranteed to posses suitable liftings, and which in a certain sense are sufficient for most practical applications of \KAP{}:

\begin{definition}[Safe \KAP{} rules]\label{def:SKR}
For a \KAP{} model $K$ and with structural constraint $\ac{c}_{K}:=\ac{c}_{\cN_K}\land \ac{c}_{\cP_K}$, let the set of \emph{constraint-preserving rules} $\LinAc{\bA_K}_{\ac{c}_K}$ (with the completion computed w.r.t.\ to $\ac{c}_{K}$ and in SqPO-semantics) be defined as
\begin{equation}
\LinAc{\bA_K}_{\ac{c}_K} := \{ R=[(r,\ac{c}_{I})]_{\sim} \mid \overline{\ac{c}_I} \not{\!\!\dot{\equiv}}\, \ac{false}\}\,.
\end{equation}
For every agent type $A\in \Ag$, let $\{a_1,\dotsc,a_{|A|}\}:=\AgSt(A)$ (where $|A|$ demotes the number of sites on agents of type $A$), and recall from Definition~\ref{def:KM} the notations for various forms of type incidences (including in particular the notation $\sigma_P:K$ for partial maps $\sigma_P:\AgSt(A)\rightarrow \Prp$). Then we define the \emph{elementary \KAP{} rules} of the \KAP{} model $K$ as follows:
\begin{enumerate}[label=(\roman*)]
\item \emph{Agent creation and deletion rules:} $\forall A\in \Ag, \sigma_P:K$,
\begin{equation}
\overline{R}^{\pm}_{(A,\sigma_P)}:= [\overline{(r^{\pm}_{(A,\sigma_P)},\ac{true})}]_{\overline{\sim}}\,,\quad
\ti{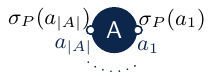}\xrightleftharpoons[\;r^{+}_{(A,\sigma_P)}\;]{r^{-}_{(A,\sigma_P)}} \mIO
\label{eq:EKRasp}
\end{equation}
Here, we take the convention that if the partial function $\sigma_P$ is not defined for a given site $a_i$, there is no property present in the \KAP{} graph in~\eqref{eq:EKRasp}, and the dotted line indicates that the rules create and delete a fully specified instance of an $A$-type agent with all of its sites (and a consistent assignment of site properties) instantiated (thus constituting an instance of. \KAP{} \emph{state} within the \KAP{} model $K$).
\item \emph{Site (un-)linking rules:} $\forall A,B\in \Ag: \forall a\in \AgSt(A), b\in \AgSt(B): (a,b)\in \StSt$,
\begin{equation}
\overline{R}^{\mathsf{(un)link}}_{(A,a),(B,b)} 
:= [\overline{(r^{\mathsf{(un)link}}_{(A,a),(B,b)},\ac{true})}]_{\overline{\sim}}\,,\quad 
\ti{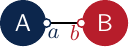}
\xrightleftharpoons[\;r^{\mathsf{link}}_{(A,a),(B,b)}\;]{r^{\mathsf{unlink}}_{(A,a),(B,b)}} 
\ti{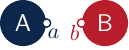}
\label{eq:EKRul}
\end{equation}
\item \emph{Site-property changing rules:} $\forall A\in \Ag: \forall a\in \AgSt(A): \forall \mathsf{p},\mathsf{p'}\in \StP(a)$,
\begin{equation}
\overline{R}_{(A,a;\mathsf{p}\to\mathsf{p}')} 
:= [\overline{(r_{(A,a;\mathsf{p}\to\mathsf{p}')} ,\ac{true})}]_{\overline{\sim}}\,,\quad 
\ti{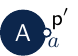}
\xleftharpoonup{r_{(A,a;\mathsf{p}\to\mathsf{p}')}} 
\ti{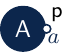}
\label{eq:EKspc}
\end{equation}
\end{enumerate}
We denote the set of \emph{elementary \KAP{} rules} for a given \KAP{} model $K$ by $\cE_K$. 
Finally, the set of \emph{safe \KAP{} rules} $\cR_K^{\mathsf{safe}}$ is defined as the subset of $\LinAc{\bA_K}_{\ac{c}_K}$ obtained via (finite) iterations of SqPO-type rule compositions of (finitely many) elementary \KAP{} rules.
\end{definition}

\begin{lemma}\label{lem:KSR}
For a \KAP{} model $K$ and the set $\cR_K^{\mathsf{safe}}$ of safe \KAP{} rules as in Definition~\ref{def:SKR}, the following properties hold:
\begin{enumerate}[label=(\roman*)]
\item The application conditions $\overline{\ac{c}_{I_X}}$ of the \emph{elementary \KAP{} rules} $\overline{R}_X=[(r_X,\overline{\ac{c}_{I_X}})]_{\overline{\sim}} \in \cE_K$ read explicitly
\begin{equation}
\overline{\ac{c}_{I_X}} \dot{\equiv} \begin{cases}
\bigwedge\limits_{N\in \cN_{(A,a),(B,b)}}
\neg \exists \left(
\ti{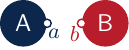}
\hookrightarrow N\right)
\quad &\text{if } \exists a:A, b:B \text{ such that } \overline{R}_X = \overline{R}^{\mathsf{link}}_{(A,a), (B,b)}\\
\ac{true} &\text{otherwise,}
\end{cases}\label{eq:acEKR}
\end{equation}
with 
\begin{equation}
\cN_{(A,a),(B,b)}:=\left\{\ti{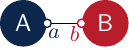}
\right\}\cup
\bigcup\limits_{\substack{%
C\in \Ag, c\in \AgSt(C)\\
(c,a)\in \StSt
}}\left\{
\ti{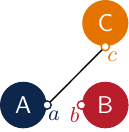}\right\}
\cup
\bigcup\limits_{\substack{%
C\in \Ag, c\in \AgSt(C)\\
(c,b)\in \StSt
}}\left\{
\ti{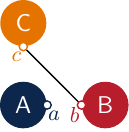}\right\}
\end{equation}
\item The only non-trivial contributions to the application condition $\overline{\ac{c}_I}$ of a generic safe \KAP{} rule $\overline{R}=[(r,\overline{\ac{c}_I})]_{\overline{\sim}}\in \cR_K^{\mathsf{safe}}$ are site-non-linkage constraints (i.e., of the form as in~\eqref{eq:acEKR}).
\end{enumerate}
\end{lemma}
\begin{proof}
Part $(i)$ of the statement follows from a direct computation of SqPO-type constraint-preserving completions according to Lemma~\ref{lem:cgcp}, while part $(ii)$ follows from combining statement $(i)$ with Lemma~\ref{lem:Bridges} and Theorem~\ref{thm:rcct}(i) (where the latter clarifies the structure of SqPO-type rule compositions in restricted rewriting). 
\end{proof}

\begin{remark}
It is instructive to consider some examples of \KAP{} rewriting rules $r\in \Lin{\bA_K}$ for some \KAP{} model $K$ that do \emph{not} lift to safe \KAP{} rules since they all possess constraint-preserving completions $\overline{R}=[\overline{(r,\ac{true})}]_{\overline{\sim}}$ with $\overline{\ac{c}_I}\,\dot{\equiv}\,\ac{false}$ (for completions computed in SqPo-semantics and w.r.t.\ the structural constraint $\ac{c}_K$ of the model $K$):
\begin{enumerate}[label=(\roman*)]
\item $\forall A\in \Ag: |\AgSt(A)|\geq 2: \forall a\in \AgSt(A)$, the ``plain'' rules $r^{\pm}_{(A,a)}\in \Lin{\bA_K}$ defined as
\begin{equation}
\ti{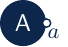}
\xrightleftharpoons[\;r^{+}_{(A,a)}\;]{r^{-}_{(A,a)}} 
\mIO\label{eq:KRAa}
\end{equation}
produce isolated site vertices when applied to states $X\in \obj{\bS_K}$.
\item $\forall A\in \Ag: \forall a\in \AgSt(A): \forall \mathsf{p}\in \StP(a)$, the ``plain'' rules $r^{\pm}_{(A,a,\mathsf{p})}$ defined as
\begin{equation}
\ti{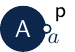}
\xrightleftharpoons[\;r^{+}_{(A,a,\mathsf{p})}\;]{r^{-}_{(A,a,\mathsf{p})}} 
\ti{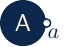}
\label{eq:KRAap}
\end{equation}
yield duplicate and missing site properties, respectively, when applied to states $X\in \obj{\bS_K}$.
\end{enumerate}
\end{remark}

In summary, it is tempting to wonder whether safe \KAP{} rules are in fact the only relevant non-trivial restricted rewriting rules for a given \KAP{} model, and we refer a dedicated analysis of this interesting theoretical question to future work.

\subsection{CTMC semantics for \KAP{} in restricted rewriting theory}\label{sec:KAPctmc}

Having identified a suitable set of constraint-preserving \KAP{} rewriting rules in the form of sage \KAP{} rules (for a given \KAP{} model $K$), we are finally in a position to specialize our universal framework for rewriting-based CTMC semantics (in the variant for SqPO-type restricted rewriting according to Corollary~\ref{cor:CTMCrrt}). The only particularly noteworthy special feature of \KAP{} CTMC semantics as compared to generic SqPO-type restricted rewriting CTMC semantics is the observation that (somewhat trivially in light of Lemma~\ref{lem:KSR}) the application of the restricted SqPO-type jump-closure operator $\overline{\bO}$ to a safe \KAP{} rule yields again a safe \KAP{} rule, so that in particular all observables of relevance in a \KAP{} stochastic rewriting system based upon safe \KAP{} rules are of this specific form. The latter feature drastically reduces the complexity of the analysis of \KAP{} ODE semantics for moments of pattern-counting observables, and we illustrate this crucial feature via the following worked example:

\begin{example}\label{ex:KappaWorkedExample}
	Consider a simple \KAP{} model with a type graph as below left that introduces two agent types $\mathsf{K}$ (for ``kinase'') and $\mathsf{P}$ (for ``protein''), where $\mathsf{K}$ has a site $k:\mathsf{K}$, and where $\mathsf{P}$ has sites $p_t,p_l,p_b:\mathsf{P}$. Moreover, the sites $p_t$ and $p_b$ can carry properties $\mathsf{u}$ (``unphosphorylated'') and $\mathsf{p}$ (``phosphorylated''), depicted as dotted loops in the type graph. Sites $k:\mathsf{K}$ and $p_l:\mathsf{P}$ can bind (as indicated by the solid line in the type graph).
\begin{equation}\label{eq:kappaExample}
\begin{array}{cc|ccc|c}
\vcenter{\hbox{\ti{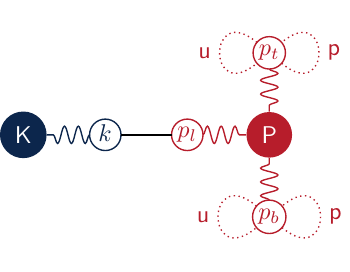}}}&\hphantom{x} &\hphantom{x} &
\begin{array}{rcl}
\mIO
	& \xrightleftharpoons[\;\overline{K}_{-}\;]{\overline{K}_{+}} &
\ti{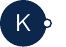}\\
\ti{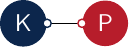}
	& \xleftharpoonup{\overline{L}_{+}} &
\ti{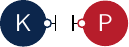}\\
\ti{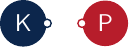}	& \xleftharpoonup{\overline{L}_{-}} &
\ti{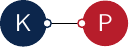}\\
\ti{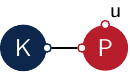}&\xrightleftharpoons[\;\overline{T}_{-}\;]{\overline{T}_{+}} &
\ti{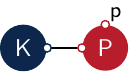}\\
\ti{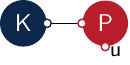}&\xrightleftharpoons[\;\overline{B}_{-}\;]{\overline{B}_{+}} &
\ti{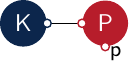}
\end{array} &\hphantom{x}&\hphantom{x}
\begin{array}{rcl}
\ti{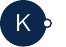}\text{$\!\!$} &\xleftharpoonup{\,r_{obs_K}\,}&
\ti{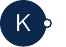}
\\
\\
\ti{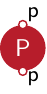}\text{$\,$} &\xleftharpoonup{\,r_{obs_P}\,}&
 \ti{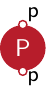}
\end{array}
\end{array} 
\end{equation}	
In order to simplify the graphical presentation of \KAP{} graphs and rewriting rules thereof in this model, we will from hereon utilize the graphical shorthand notation convention according to Remark~\ref{rem:KGNC}, and moreover omit the site-type labels (as in the present example site-types may be inferred via the positions of sites on agents in the shorthand notation). 

As a prototypical example of a \KAP{} stochastic rewriting system, consider a system based upon the rewriting rules $k_{\pm}$, $l_{\pm}$, $t_{\pm}$ and $b_{\pm}$, whose constraint-preserving completions $\bar{K}_{\pm}$, $\bar{L}_{\pm}$, $\bar{T}_{\pm}$ and $\bar{B}_{\pm}$ are depicted in the middle column of~\eqref{eq:kappaExample}. Here, we have employed a graphical notation\footnote{Coincidentally, owing to the fact that all of the rules presented are instances of \emph{safe} \KAP{} rules in the sense of Definition~\ref{def:SKR}, the only non-trivial application conditions encountered are those expressing non-boundedness of sites, which is why we are able to utilize a succinct graphical notation precisely as in the \KAP{} literature without loss of information in terms of the restricted rewriting semantics.} whereby rules with conditions are depicted as their ``plain'' rules, and with input interfaces annotated such as to indicate the structure of the constraint-preserving application conditions (with no annotation by convention in case of a trivial condition $\ac{true}$). We find that only the site-linking rule $\bar{L}_{+}$ requires a non-trivial constraint-preserving condition, namely one that ensures that the site of the $\mathsf{K}$-type agent and the left site of the $\mathsf{P}$-type agent must be \emph{free} (i.e., not linked to any other site) before binding.

Consider then for a concrete computational example the time-evolution of the average count of the pattern described in the identity rule $r_{obs_P}$. As typical in \KAP{} rule specifications, $r_{obs_P}$ as well as several of the other rules depicted only explicitly involve \emph{patterns}, but not necessarily \emph{states}, since e.g.\ in $r_{obs_P}$ the left site of the $\mathsf{P}$-type agent is not mentioned. In complete analogy to the computation presented in Example~\ref{ex:ugModel}, let us first compute the commutators of the observable  $O_{\mathsf{K}}=\bar{\rho}(\bar{\delta}(r_{obs_K};\ac{true}))$ with the operators $\hat{X}:=\bar{\rho}(\bar{\delta}(\bar{X}))$ (for $X\in \{K_{\pm}, L_{\pm}, T_{\pm},B_{\pm}\}$, and with $\bar{\rho}:=\bar{\rho}^{SqPO}_{\hat{\bS}_{\KAP{}}}$):
\begin{equation}
\begin{aligned}
	[O_{\mathsf{K}},\hat{K}_{\pm}]&=\pm \hat{K}_{\pm}\,,\; [O_{\mathsf{K}},\hat{L}_{\pm}]=[O_{\mathsf{K}},\hat{T}_{\pm}]=[O_{\mathsf{K}},\hat{B}_{\pm}]=0
\end{aligned}
\end{equation}
However, letting $O^{(\mathsf{x},\mathsf{y})}_{P}$, $O^{(\mathsf{x},\mathsf{y})}_{link}$ and $O^{(\mathsf{x},\mathsf{y})}_{free}$ denote the observables for the patterns
\[
\omega_P^{(\mathsf{x},\mathsf{y})}:=
\ti{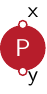}\,,\quad \omega_{link}^{(\mathsf{x},\mathsf{y})}:=
\ti{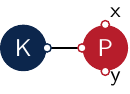}\,,\quad \omega_{free}^{(\mathsf{x},\mathsf{y})}:=
\ti{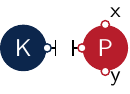}
\]
one may easily demonstrate that even a comparatively simple observable such as $O^{(\mathsf{p},\mathsf{p})}_{P}$ already leads to an infinite cascade of  contributions to the ODEs for the averages of pattern counts. As typical in these sorts of computations, the discovery of a new pattern observable via applying restricted SqPO-type jump-closure (Corollary~\ref{cor:CTMCrrt}) to the commutator contributions to $\tfrac{d}{dt}\langle O^{(\mathsf{p},\mathsf{p})}_{P}\rangle(t)$ leads to the discovery of new pattern observables yet again, such as in
\[
	[O_{\mathsf{P}},\hat{T}_{+}]=\hat{T}_{+}^{(\mathsf{p})}\,,\;
	\jcOp{\hat{T}_{+}^{(\mathsf{p})}}=O^{(\mathsf{u},\mathsf{p})}_{link}\,,\;
	[O^{(\mathsf{u},\mathsf{p})}_{link},\hat{L}_{+}]=\hat{L}^{(\mathsf{u},\mathsf{p})}\,,\;
	\jcOp{\hat{L}^{(\mathsf{u},\mathsf{p})}}=O^{(\mathsf{u},\mathsf{p})}_{free}\,.
\]
In particular the last observable $O^{(\mathsf{u},\mathsf{p})}_{free}$ is found to lead to an infinite tower of other observables (i.e., ``ODE system non-closure''), starting from
\begin{equation*}
	[O^{(\mathsf{u},\mathsf{p})}_{free},\hat{L}_{+}]=-\hat{L}^{(\mathsf{u},\mathsf{p})}
	-\left(
	\ti{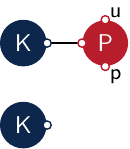}\leftharpoonup \ti{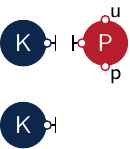}
	\right)-\left(\ti{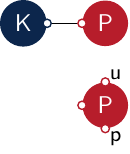}\leftharpoonup \ti{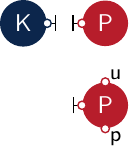}\right)\,.
\end{equation*}
\end{example}

This exemplary and preliminary analysis reveals that while the rule-algebraic CTMC implementation is indeed fully faithfully applicable to the formulation and analysis  of biochemical reaction systems, further algorithmic and theoretical developments will be necessary (including possibly ideas of \emph{fragments} and \emph{refinements} as in~\cite{danos2008rule,Danos_2010,Harmer_2010}) in order to better understand the precise nature of differential semantics and in particular ODE systems for the moments of pattern-counting observables in our new rewriting-theoretic implementation of \KAP{}.

\section{Application scenario 2: organic chemistry with \MOD{}}\label{sec:ocgr}

From a purely rewriting-theoretical standpoint, it is fascinating to observe that reactions in organic chemistry are in fact nothing but certain types of rewriting rules.
Important contributions towards rendering this heuristic observation into a tractable algorithmic theory were one of the key early achievements of the \MOD{} framework~\cite{Andersen_2016} (which will be briefly reviewed in Section~\ref{sec:modBackground}). 
In order to establish a CTMC theory for organic chemistry as a suitable
specialization of our universal rule-algebraic CTMC theory, it is necessary to faithfully encode the defining structural properties of chemistry (including laws for the possible atoms, molecule structures and reactions thereof) within the semantics of restricted rewriting. 
Despite a rich history spanning almost 20 years of developments of the \MOD{} framework and its predecessors to date, the precise encoding of chemistry models as instances
of a rewriting theory had not been formulated at the time of writing of the 
present paper. While the \MOD{} framework provides a number of algorithmic capabilities that could aid the \emph{implementation} of
such a model, it was not primarily designed for the purposes of organo-chemical rewriting, so that in particular its underlying rewriting subsystem is one of certain forms of undirected typed simple graphs, and notably without the capability to explicitly enforce the structural constraints relevant to organic chemistry (cf.\ Section~\ref{sec:modBackground}). Indeed, in typical application scenarios of \MOD{} in the chemistry setting, one does not require a restricted rewriting semantics for its rewriting subsystem, as rewriting rules are typically employed exclusively for calculating direct derivations 
(i.e., in order to compute chemical spaces as hypergraphs).

Consequently, the material presented in Section~\ref{sec:modEncoding} in fact 
constitutes the first-of-its-kind full-fledged theoretical underpinning of organic chemistry via a rewriting-theoretical approach, in the form of  
 a restricted DPO-type rewriting theory over $\cM$-adhesive categories of undirected typed multigraphs,
and with chemical models defined via choices of structural constraints encoding the laws of organic chemistry.
This original result not only permits to provide a formalization of the key operation of 
rule composition for organic chemistry (simply as the one in compositional DPO-type restricted rewriting),
but in particular also to establish a CTMC theory for organic chemistry as a suitable specialization of our universal rule-algebraic CTMC theory,
constituting yet another key result of the present paper.
As a first step towards an algorithmic implementation of ODEs for moments of pattern-counting observables,
we assess in Section~\ref{sec:rcMOD} the current state of the \MOD{}-platform in view of the requisite operations of rule compositions
and the analysis of the resulting composite rules in terms of equivalence relations, identifying in the process some interesting directions for future work.

\subsection{Graph-based algorithmic cheminformatics}\label{sec:modBackground}

\href{https://cheminf.imada.sdu.dk/mod/}{\MOD{}}~\cite{Andersen_2016} is a generic framework for graph-based algorithmic cheminformatics.
It includes a system to generate chemical reaction networks (i.e., derivation graphs as directed multi-hypergraphs)
based on chemical reactions modeled as graph grammars in DPO-style rewriting.
A domain-specific programming language for specifying the strategy of how to apply transformation rules allows for a controlled expansion of the derivation graph \cite{strat:14}. %
In a prototypical application scenario of \MOD{}, derivation graphs are generated and subsequently analyzed using integer hyperflows %
in order to automatically infer and enumerate chemical transformation motifs (such as autocatalysis) or to compute optimal or near-optimal pathways,
e.g., for enzymatic design questions (see \cite{Andersen_2019} for details).
In general, \MOD{} follows \cite{ehrig:2006fund} in syntax, whence rules are written from left to right as $p = (L\leftarrow K\rightarrow R)$,
with $L$ denoting the input pattern (i.e., $L = I$ and $R = O$ in comparison to Def.~\ref{def:rwc}). The structural properties present in organic chemistry at a high level motivated to base the formulation of \MOD{} on DPO-type semantics.

\begin{example}
The rule depicted in Figure~\ref{fig:aldol} (denoted as $r_+$) illustrates the so-called \emph{aldol addition} reaction as a rewriting rule.
Aldol addition is an industrially important chemical reaction which merges two compounds to form so-called aldols
(the naming stems from the aldehyde that merges with an alcohol, which is a structural pattern seen in many of the products of an aldol addition).
\begin{figure}
\centering
\subfigure[\label{fig:aldol}$r_+$]{%
	\mi{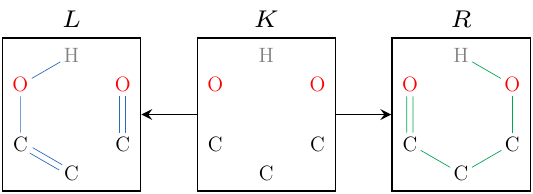}
}%
$\hspace{2cm}$
\subfigure[\label{fig:aldolid}$r_{id}$]{%
	\mi{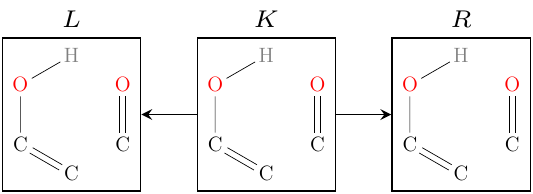}
}%
\caption{Depiction of two chemical DPO rules: \subref{fig:aldol} aldol addition $r_{+}$ and \subref{fig:aldolid} an ``identity rule'' based upon the input motif $I_{+}$ of the aldol addition rule $r_{+}$ (here depicted as $L$).}
\label{fig:formoseRules}
\end{figure}
In order to illustrate to which extent \MOD{}~can be employed for the automatic computation of commutators,
we will employ a second rule (denoted as $r_{id}$, Figure~\ref{fig:aldolid}), a simple identity rule based on the left side of the aldol addition rule.
\end{example}

While the focus of the \MOD{} framework is on the generation and the analysis of chemical reaction networks, it was designed and implemented in a much more generic way,
using rule composition as the underlying algorithmic primitive for enumerating direct derivations and without a particular restriction to chemistry. %
More specifically, the interface of \MOD{} allows limited access to certain forms of rule composition algorithms through several different partial overlap enumeration operators
(see \cite{andersen2018rule} for details). %
In typical chemical application scenarios, the overlaps are generated to represent a subset of connected components of a rule side,
which under the assumption of the input rules being chemically consistent in turn provides certain limited guarantees that composed rules preserve chemical validity,
e.g., making it possible to use composition to trace atoms through sequences of (bio-)chemical reactions \cite{trace:14}.
However, the interface of \MOD{} also exposes a more general overlap enumeration operator,
where all partial overlaps corresponding to \emph{common vertex-induced subgraphs} of the rule sides are enumerated for composition.
While the composed rules are not necessarily chemical,
we will in an algorithmic case study presented in Section~\ref{sec:rcMOD} simply filter the results for chemical validity as an initial prototype implementation,
highlighting in particular the future developments that will be necessary to implement a correct encoding of the notion of chemical rule composition
that is required for CTMC semantics of organic chemistry (i.e., the one of restricted rewriting theory).

\subsection{Formalization of \MOD{} as a restricted rewriting theory}\label{sec:modEncoding}

From a theoretical perspective, the \MOD{} framework includes a rewriting subsystem over what one could call \emph{pre-chemical graphs},
namely undirected simple graphs that are typed in a form that closely mimics the structure of chemical formulas (i.e., vertices are labeled by atom types, edges come in varieties of chemical bond types etc.)~\cite{andersen2018rule}. %
The restriction to simple graphs is not directly implemented as a theoretical specification with conditions, but is enforced on an algorithmic level.
When modeling chemical systems, each graph is in practice a \emph{chemical graph},
which heuristically may be described as a pre-chemical graph with a suitable typing that in addition satisfies certain structural constraints dictated via the laws of chemistry (such as admissible bonding patterns, valency constraints etc.).
These constraints are however not implemented in the \MOD{} algorithms, %
but are to be enforced manually via ensuring that chemically correct input rules are provided to the algorithms, and by induction through direct derivations. %
This lack of a formal specification of the chemical constraints and related algorithmic details poses a considerable obstacle in view of faithfully encoding chemistry in terms of rewriting theory.

In view of the present paper, in order to firmly root the construction in standard DPO-type categorical rewriting theory, %
a necessary prerequisite consists in identifying an \emph{ambient category} that satisfies suitable ($\cM$-) adhesivity properties. %
A first attempt in this direction was made in~\cite{andersen2018rule}, where it was postulated that (pre-)chemical graphs could possibly be interpreted as objects of a certain typed and undirected variant of the category $\mathbf{PLG}$ of partially labeled directed graphs. 
While the latter category had been introduced in~\cite{Habel_2012} as a key example of an $\cM$-$\cN$-adhesive category,
with the motivation of permitting label-changes in rewriting rules, it was also demonstrated in loc cit.\ that $\mathbf{PLG}$ is \emph{not} $\cM$-adhesive.
Since moreover no concrete construction of a tentative variant $\mathbf{uPLG}$ of $\mathbf{PLG}$ for undirected graphs,
let alone results on the possible adhesivity properties of such a category are known in the literature, this attempt proved unsuccessful. %
While it might seem feasible to forego the semantic capability of partial relabeling in favor of utilizing some typed variant of undirected simple graphs as the basic data structure, it is in fact well-known~\cite{quasi-topos-2007} that the category of directed as well as the category of undirected \emph{simple} graphs (which may be formally encoded as the category $\mathbf{BRel}$ of binary relations for the directed and a certain symmetric restriction thereof for the undirected case) are both \emph{not} adhesive, but only form quasi-topoi, which thus in particular prohibits the use of undirected \emph{simple} graphs as a base category in chemical rewriting. In summary, due to the lack of a well-defined base category, it is strictly speaking not even clear whether or not the semantics of rule applications and compositions as implemented in \MOD{} are even instances of DPO-rewriting constructions.\\

In this section, we resolve this conundrum via introducing a fully consistent and faithful encoding of chemical rewriting within the formalism of restricted rewriting (Section~\ref{sec:rrt}). Upon a careful analysis of the algorithmic constructions implemented in the \MOD{} framework~\cite{Andersen_2016,andersen2018rule} with regards to organic chemistry, we base our new construction on an \emph{ambient category} of undirected \emph{multi-}graphs\footnote{Inspired by the \KAP{} constructions in the previous section, we opt to represent \emph{properties} (which may in a more general setting also include e.g.\ \emph{charges} on atoms)
as \emph{typed loop edges} on vertices representing atoms,
whence the change of a property (which was the main motivation in~\cite{andersen2018rule} for postulating the need for employing a variant of $\mathbf{PLG}$)
may be encoded in a rewriting rule simply via deletion/creation of property-encoding loops. 
} typed over a type graph whose \emph{vertex types} represent \emph{atom types}, and whose \emph{edge types} represent \emph{bond types}. It should be noted that this category in particular qualifies as an $\cM$-adhesive category suitable for DPO-type rewriting in the sense of Assumption~\ref{as:main}. 
In a second step, via the definition of certain negative and positive \emph{global structural constraints} (encoding the chemical constraints such as bond configurations etc.), we define a \emph{pattern category} (whose objects and monomorphisms are utilized to define chemical rewriting rules) as well as a \emph{state category} (whose objects are precisely the \emph{chemical graphs}). A specific model of organic chemistry then is defined in terms of certain pieces of data, such as atom types relevant to the model, bond configuration patterns specific to these atom types and possibly additional global constraints from the given practical application in chemistry. 

\begin{definition}\label{def:modRRT}
Let $\cA := \{A_1,\dotsc,A_N\}$ (for some $N\in \bZ_{>0}$) denote a set of \emph{atom types}, i.e., a subset of the atom types present in the periodic table of elements.
Let $\cB:=\{\mathsf{-},\mathsf{=},\mathsf{\equiv}\}$ denote the single-, double- and triple-bond types, respectively. 
Then \emph{a model of organic chemistry via restricted rewriting theory}, or a \emph{\CHEM{} model} for short, is specified via a tuple of data $\mathsf{M}=(\cA,\cB,\cN_{\cA},\cP_{\cA})$ as follows (where $\cN_{\cA}$ and $\cP_{\cA}$ are sets of data used to define the negative and positive sub-constraints, respectively, of a \emph{structural constraint} $\ac{c}_{\CHEM{}}:= \ac{c}_{-}\land \ac{c}_{+}$, see below):
\begin{enumerate}[label=(\roman*)]
\item We define the \emph{ambient category} $\bA_{\CHEM{}}$ of the \CHEM{} model%
\footnote{To simplify notations, we will mostly omit the explicit mentions of the data of a given \CHEM{} model,
	since this data is assumed to be kept fixed throughout a given set of computations.}
\begin{equation}
\bA_{\CHEM{}} := \mathbf{uGraph}\diagup T_{\cA}\,,\quad T_{\cA}:=
\ti{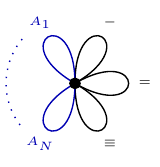}
\end{equation}
The type graph $T_{\cA}$ thus contains a universal vertex type, one edge type for each atom type (colored in {\color{h1color}dark blue}), and one edge type per bond type (colored in black). 
\item The \emph{negative constraint} $\ac{c}_{-}$,
\begin{equation}
\ac{c}_{-} := \bigwedge_{N\in \cN_0\cup \cN_{\cA}} \neg \exists (\mIO\hookrightarrow N)\,,
\end{equation}
is defined in terms of certain \emph{elementary structural constraints} (specified via a set $\cN_0$ of negative patterns expressing simplicity of the graphs and certain other consistency properties),
\begin{equation}
\cN_0:= \bigcup_{\varepsilon, \varepsilon' \in E_{T_{\cA}}}\left\{
\ti{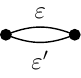}
\right\} \cup
\bigcup_{j=1}^N\left\{
\ti{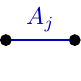}
\right\}\cup
\bigcup_{j,k=1}^N\left\{
\ti{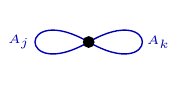}
\right\}
\cup \bigcup_{\beta\in \{\mathsf{-},\mathsf{=},\mathsf{\equiv}\}}\left\{
\ti{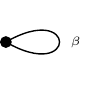}
\right\}
\end{equation}
and a set $\cN_{\cA}=\cup_{A\in \cA} \cN_A$ of \emph{forbidden bond patterns}, with the latter constituting an additional piece of data necessary for specifying a \CHEM{} model. 
\item The \emph{positive constraint} $\ac{c}_{+}$ expresses that each vertex must carry an atom type loop, and in addition that each vertex carrying a given loop of  atom type $A\in \cA$ must extend into one of a set $\cP_A$ of \emph{permitted bond patterns} for this atom type (with $\cP_{\cA}:=\cup_{A\in \cA}\cP_A$ constituting part of the data specifying the \CHEM{} model):
\begin{equation}
\ac{c}_{+}:= \forall\left(\mIO \hookrightarrow \ti{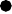}, \bigvee_{A\in \cA} \exists\left(
\ti{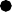}\hookrightarrow 
\ti{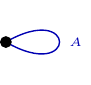}\right)\right)\land
\bigwedge_{A\in \cA}
\forall\left(
\mIO\hookrightarrow
\ti{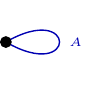}\,, \bigvee_{P\in \cP_A}\exists\left(
\ti{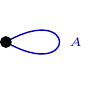}\hookrightarrow P
\right)
\right)
\end{equation}
\item The \emph{pattern category} $\bP_{\CHEM{}}$ and the \emph{state category} $\bS_{\CHEM{}}$ of the \CHEM{} model are defined (compare Section~\ref{sec:aps}) as the full subcategories of $\bA_{\CHEM{}}$ obtained via restriction of objects to those satisfying $\ac{c}_{-}$ and to those satisfying $\ac{c}_{\CHEM{}}:=\ac{c}_{-}\land \ac{c}_{+}$, respectively.
\end{enumerate}
\end{definition}

\begin{remark} It is worthwhile clarifying that while our present approach is easily extensible to capture more generic features of organic chemistry such as
charges, lone electron pairs, or aromaticity, we focus here on a more basic variant of the theory for concreteness (i.e., purely on the aforementioned subset of bond-types),
since this version of the theory is fully sufficient to introduce the key concepts and theoretical structures.
\end{remark}

\begin{example}\label{ex:modMod}
In a very rudimentary model of organic chemistry involving only three atom types $\cA =\{H,C,O\}$ (hydrogen, carbon and oxygen), one could impose the forbidden bond configurations $\cN_{\cA}$ and the permitted bond configurations $\cP_{\cA}$ in the following form:
\begin{equation}
\begin{aligned}
\cN_H&= \left\{\!\!\!
\ti{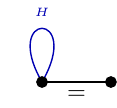}\,,\;
\ti{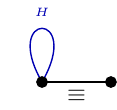}\,,\;
\ti{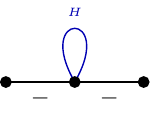}
\right\} & 
\cN_O&= \left\{\!\!\!
\ti{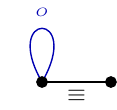}\,,\;
\ti{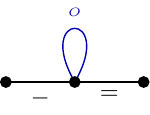}\,,\;
\ti{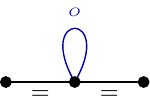}
\right\} \\
\cP_H&= \left\{\!\!\!
\ti{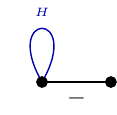}
\right\} & 
\cP_O&= \left\{
\ti{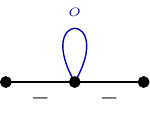}\,,\;
\ti{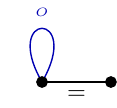}
\right\} 
\end{aligned}
\end{equation}
An interesting and well-known (in a sense even quintessential) feature of organic chemistry is the complexity of the specifications $\cN_C$ and $\cP_C$ for the carbon atom type,
which is illustrated in Figure~\ref{fig:typeC} (where also all bond configurations are depicted which may occur in patterns involving a carbon atom,
i.e., configurations which do not violate the negative constraints). In order to permit suitably compact presentation, Figure~\ref{fig:typeC} utilizes a graphical shorthand notation akin to the one of the \MOD{} platform, wherein a vertex with an atom-type loop is simply presented by the name of the atom type, and where single-, double- and triple-bond types are typeset in chemical notation.
\begin{figure}
\centering
\ti{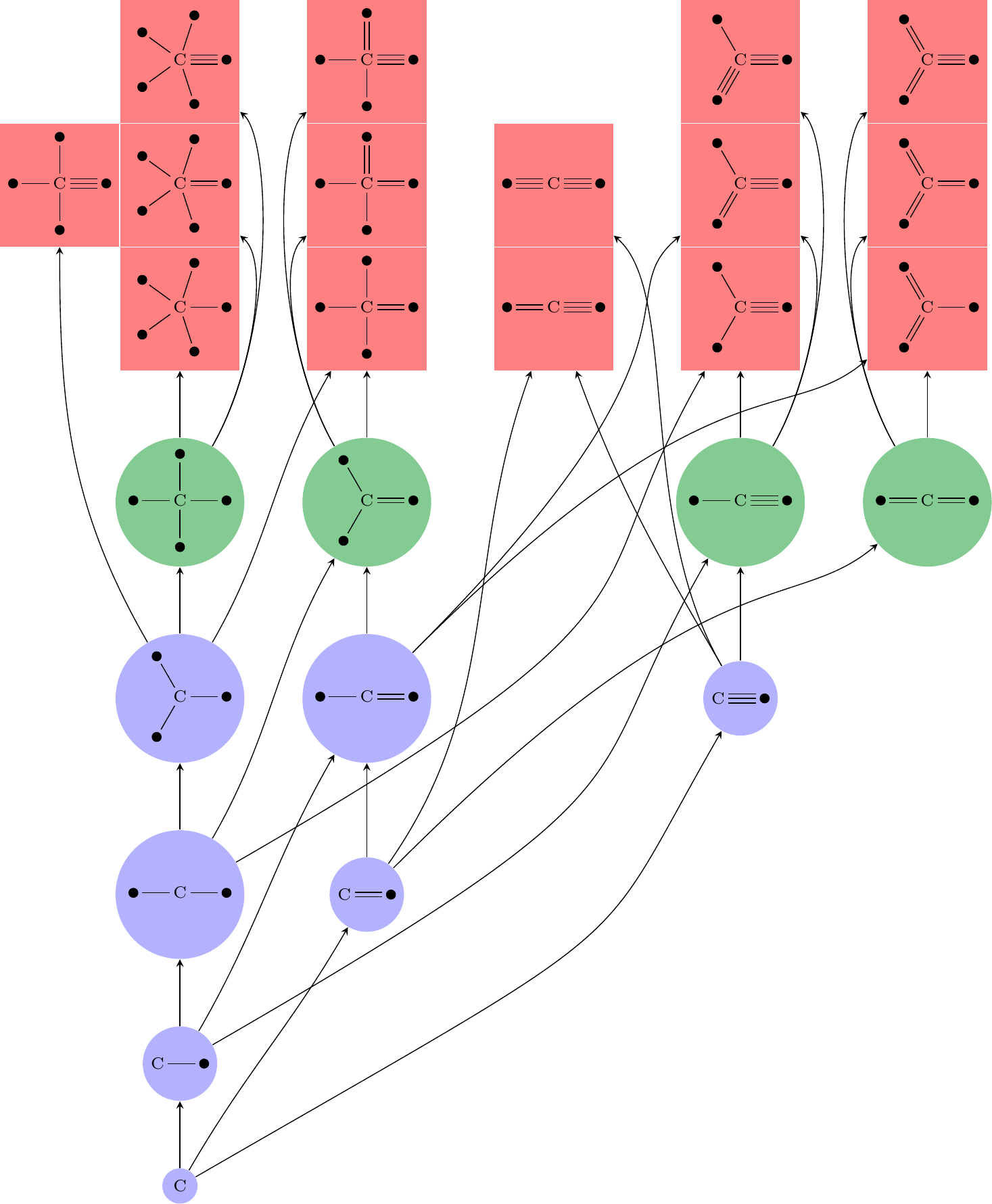}
\caption{Specification of the hierarchy of bond patterns for carbon (C) with the forbidden bond patterns in {\color{red!50!white}red squares}.
The patterns in {\color{ForestGreen!50!white}green circles} represent those found in molecules, while the remaining patterns in {\color{blue!30!white}light blue circles} are molecule patterns.
To simplify the figure the atom and bond types are depicted as in chemistry.}
\label{fig:typeC}
\end{figure}%
\end{example}

With the encoding of organic chemistry as presented in Definition~\ref{def:modRRT} as a particular instance of a DPO-type restricted rewriting theory over an $\cM$-adhesive category that satisfies Assumption~\ref{as:main}, we may leverage our universal CTMC framework in order to arrive at a fully consistent and first-of-its-kind formulation of the CTMC semantics of organic chemistry as a direct application of Theorem~\ref{thm:CTMCs} (universal rewriting-based CTMC theory) and Corollary~\ref{cor:CTMCrrt} (specialization to restricted rewriting theories). The construction of a CTMC for a given reaction system in organic chemistry (specified typically in the form of ``plain'' chemical rules in the literature) in terms of a CTMC in restricted rewriting theory specified via a \CHEM{} model equivalent to the data of the chemical input data may be obtained as follows:
\begin{enumerate}[label=(\roman*)]
\item For each ``plain'' chemical rule $r$ of the reaction system, compute its constraint-preserving completion $\overline{R}:=[(r,\overline{\ac{c}_I})]_{\overline{\sim}}$ (i.e., starting from $R:=[(r,\ac{c}_I)]_{\sim}$ with trivial condition $\ac{c}_{I}:=\ac{true}$).
\item If $\kappa\in \bR_{>0}$ denotes the base rate of the ``plain'' chemical rule $r$ in the reaction system, the contribution to the infinitesimal CTMC generator $\cH$ resulting from the rule is (with $\overline{\rho}:=\overline{\rho}^{DPO}_{\bA_{\CHEM{}}}$)
\begin{equation*}
	\cH_r:= \kappa\left(\overline{\rho}\left(\overline{\delta}(\overline{R})\right) - \overline{\bO}\left(\overline{\delta}(\overline{R})\right)\right)\,.
\end{equation*}
\end{enumerate}
From hereon, the general CTMC theory for DPO-restricted rewriting as introduced in Section~\ref{sec:rrt} is applicable.

While a general discussion or even a classification of the resulting concrete instances of CTMCs for the enormously rich variety of organo-chemical reaction systems
encountered in practice is outside the scope of the present paper and thus left for future work,
suffice it here to mention a few interesting preliminary observations on the structure of \CHEM{}-based CTMCs.
The first observation concerns the typical structure of the ``plain'' rules specified in the language of standard organic chemistry in the form of reactions.
For most of the computations typically considered for analysis with \MOD{}, and motivated by practical applications,
these ``plain'' rules are either of the form $r_{M}^{+}:=(M\leftharpoonup \mIO)$ or $r_{M}^{-}:=(\mIO\leftharpoonup M)$
(the creation or deletion of a fully specified molecule pattern $M\in \obj{\bS_{\CHEM{}}}$),
or more generally ``plain'' rules of the form $r=(O\hookleftarrow K\hookrightarrow I)$ where $O,K,I\in \obj{\cP_{\CHEM{}}}$ are patterns,
and such that $r\vert_V=(V_O\xleftarrow{\cong}V_K\xrightarrow{\cong}V_I)$ (i.e., the restriction of $r$ to the vertex-part of the span of $\CHEM{}$-pattern) is a span of isomorphisms.
A preliminary analysis reveals (see Example~\ref{ex:modModB} below for the concrete case of $r_{+}$) that %
the constraint-preserving completions of such ``plain'' rules %
yield application conditions $\overline{\ac{c}_I}$ that satisfy $\overline{\ac{c}_I}\,\not{\!\!\dot{\equiv}} \,\ac{false}$, %
and whose only non-trivial contributions are of the form of atom-vertex non-linking negative constraints,
akin to the structure discovered in simple graph rewriting in Example~\ref{ex:sgraph-cpc} and Lemma~\ref{lem:Bridges}.

A second and more technical observation concerns a peculiarity related to the computation of the action of the restricted DPO-type jump-closure operator $\overline{\bO}$ on the aforementioned types of constraint-completed \CHEM{} rules. As one may easily verify from the concrete definition of pushouts and pushout complements in the category $\mathbf{uGraph}$, and via the definition of the action of a DPO-rule $r=(O\hookleftarrow K\hookrightarrow I)\in \LinAc{\bA_{\CHEM{}}}$ on some object $X\in \bS_{\CHEM{}}$, one finds the following result:
\begin{lemma}
Given a \CHEM{}-rule $R:=(r,\ac{c}_I)$ with the special property $r\vert_V=(V_O\xleftarrow{\cong}V_K\xrightarrow{\cong}V_I)$, then for $\overline{R}:=[(r,\overline{\ac{c}_I})]_{\overline{\sim}}$ and $\overline{R}_{id}:=[(r_{id},\overline{\ac{c}_I})]_{\overline{\sim}}$ (with $r_{id}:=(I\xleftarrow{id}I\xrightarrow{id}I)$) we find that
\begin{equation}
\bra{} \overline{\bO}\left(\overline{\delta}(\overline{R})\right)\ket{X} = \bra{} \overline{\bO}\left(\overline{\delta}(\overline{R}_{id})\right)\ket{X}\,.
\end{equation}
This equation follows in turn from the stronger property $\overline{\bO}\left(\overline{\delta}(\overline{R})\right)=\overline{\bO}\left(\overline{\delta}(\overline{R}_{id})\right)$ (in the sense of equality of linear operators acting on the vector space $\hat{\overline{\bS}}_{\CHEM{}}$ of \CHEM{}-states.
\end{lemma}
Via this lemma, when working with \CHEM{}-rules of the aforementioned structure, it is sufficient to consider \emph{pattern-counting observables} that are of the form $\overline{\bO}\left(\overline{\delta}(\overline{R}_{id})\right)$ (i.e., based upon ``identity rules'' $r_{id}$) throughout the CTMC calculus, which thus poses a convenient algorithmic simplification.

We illustrate some of the features alluded to above with the following example:

\begin{example}[Ex.~\ref{ex:modMod} continued]\label{ex:modModB}
Inspecting the complexity of the structural constraints even for the comparatively rudimentary \CHEM{} model introduced in Example~\ref{ex:modMod}, it might at first appear entirely infeasible to present any chemical rewriting rules with constraint-preserving completions of their application conditions explicitly. Quite remarkably, this is however not the case upon closer inspection, as most of the contributions to the structural constraints $\ac{c}_{\CHEM{}}$ typically do not contribute non-trivially to the constraint-preserving completions. Postponing a more comprehensive analysis of this phenomenon to future work, suffice it here to illustrate this effect via the example of the aldol addition rule $r_{+}$. Starting from $R_{+}:=(r_{+},\ac{true})$ (i.e., encoding the ``plain'' rule $r_{+}$ as a rule $R_{+}$ with the same underlying rule and a trivial application condition $\ac{c}_{I_{+}}:=\ac{true}$), we find the constraint-preserving completion $\overline{\ac{c}_{I_{+}}}$ of $\ac{c}_{I_{+}}:=\ac{true}$ to evaluate to
\begin{equation}\label{eq:acRplus}
\overline{\ac{c}_{I_{+}}}:=\bigwedge_{{\color{red}N \in \{ I_{+}^{\mathsf{-}}, I_{+}^{\mathsf{=}}\}}} \neg \exists (I_{+}\hookrightarrow {\color{red}N})\,.
\end{equation}
Here, $I_{+}^{\mathsf{-}}$ and $I_{+}^{\mathsf{=}}$ are variants of $I_{+}$ where the two carbons that are linked with a single bond via the rule $r_{+}$ (i.e., the ones marked $\langle0\rangle$ and $\langle5\rangle$ in the depiction of $r_{+}$ in the top right part of Figure~\ref{fig:modcomp}) are linked with a single and a double bond, respectively. The interested readers are invited to compare this result to the ones of Example~\ref{ex:sgraph-cpc} and Lemma~\ref{lem:Bridges}, which in particular facilitates to interpret the condition $\overline{\ac{c}_{I_{+}}}$ as the one that prevents a double-edge to be produced when applying $r_{+}$ to a chemically consistent graph. Finally, the pattern-counting observable associated to $\overline{R}_{+}:=[(r_{+},\overline{\ac{c}_{I_{+}}})]_{\overline{\sim}}$ (i.e., the diagonal operator whose Eigenvalues count the numbers of admissible matches of $\overline{R}_{+}$ into \CHEM{}-states) to evaluate to
\begin{equation}\label{eq:oplus}
\hat{\bar{O}}_{+}:=\overline{\rho}\left(\overline{\delta}\left(r_{id},\overline{\ac{c}_{I_{+}}}\right)\right)\,,
\end{equation}
with $r_{id}:=(I_{+}\xleftarrow{id}I_{+}\xrightarrow{id}I_{+})$ the ``identity (plain) rule'' on $I_{+}$.
\end{example}

\subsection{Case study: rule compositions in \MOD{}}\label{sec:rcMOD}

As a first exploration of the algorithmic capabilities of the current implementation of the \MOD{} platform in the context of CTMC semantics, let us contemplate a simple example for the computation of a pattern-count observable moment-evolution ODE in a chemical reaction system consisting purely of the aldol addition reaction. 

\subsubsection{Computational problem statement for the aldol addition example}

Unfortunately, at present \MOD{} does not permit to explicitly implement application conditions, which prohibits a priori to study the observable $\hat{\bar{O}}_{+}$ (which has a non-trivial application condition), so we will in effect have to consider alternatively the case of an observable $\hat{O}_{+}:=\overline{\rho}(\overline{\delta}(r_{id},\ac{true}))$ that is a variant of $\hat{\bar{O}}_{+}$ with a trivial application condition (and which thus counts patterns of the form $I_{+}$ regardless of whether or not the two connected components of $I_{+}$ are matched into different components in a given state). Specializing Theorem~\ref{thm:CTMCmev} to the case of computing the evolution equation for the averages of either $\hat{\bar{O}}_{+}$ or of $\hat{O}_{+}$ for a reaction system just consisting of the aldol addition rule (for simplicity at unit rate), i.e., for an infinitesimal CTMC generator
\begin{equation}
\overline{\cH} = \hat{\bar{H}}_{+} - \hat{\bar{O}}_{+}\,,\quad  \hat{\bar{H}}_{+}:= \bar{\rho}(\bar{\delta}(\overline{R}_{+}))\,,
\end{equation}
and choosing some input probability distribution $\ket{\Psi(0)}=\ket{\Psi_0}\in \Prob{\bS_{\CHEM{}}}$ at time $t=0$ (supported over \emph{states}, i.e., over basis vectors of $\hat{\bar{\bS}}_{\CHEM}$), we obtain the evolution equations:
\begin{subequations}
\begin{align}
\tfrac{d}{dt}\bra{}\hat{\bar{O}}_{+}\ket{\Psi(t)}
&= \bra{}[\hat{\bar{O}}_{+},\overline{\cH}] \ket{\Psi(t)}\,, 
& & \bra{}\hat{\bar{O}}_{+}\ket{\Psi(0)}=\bra{}\hat{\bar{O}}_{+}\ket{\Psi_0} \label{eq:MevoA}\\
\tfrac{d}{dt}\bra{}\hat{O}_{+}\ket{\Psi(t)}
&= \bra{}[\hat{O}_{+},\overline{\cH}] \ket{\Psi(t)}\,, 
& & \bra{}\hat{O}_{+}\ket{\Psi(0)}=\bra{}\hat{O}_{+}\ket{\Psi_0}\,. \label{eq:MevoB}
\end{align}
\end{subequations}
Here, we have employed the standard notation $[A,B] := AB-BA$ for the \emph{commutator} (for linear operators $A,B\in End_{\bR}(\hat{\bar{\bS}}_{\CHEM{}}$). 
Since diagonal operators such as $\hat{\bar{O}}_{+}$ and $\hat{O}_{+}$ \emph{commute} (i.e., $[\hat{\bar{O}}_{+},\hat{O}_{+}]=0$), we may simplify the evolution equations to
\begin{subequations}
\begin{align}
\tfrac{d}{dt}\bra{}\hat{\bar{O}}_{+}\ket{\Psi(t)}
&= \bra{}[\hat{\bar{O}}_{+},{\color{blue}\hat{\bar{H}}_{+}} ] \ket{\Psi(t)}\,, 
& & \bra{}\hat{\bar{O}}_{+}\ket{\Psi(0)}=\bra{}\hat{\bar{O}}_{+}\ket{\Psi_0} \label{eq:MevoA2}\\
\tfrac{d}{dt}\bra{}\hat{O}_{+}\ket{\Psi(t)}
&= \bra{}[\hat{O}_{+},{\color{blue}\hat{\bar{H}}_{+}} ] \ket{\Psi(t)}\,, 
& & \bra{}\hat{O}_{+}\ket{\Psi(0)}=\bra{}\hat{O}_{+}\ket{\Psi_0}\,. \label{eq:MevoB2}
\end{align}
\end{subequations}
To proceed, we must therefore compute the contributions to the two commutators $[\hat{\bar{O}}_{+},{\color{blue}\hat{\bar{H}}_{+}} ]$ and $[\hat{O}_{+},{\color{blue}\hat{\bar{H}}_{+}} ]$ via the restricted rewriting rule-algebraic calculus according to Theorem~\ref{thm:rrap} (utilizing the explicit definition of the rule-algebraic composition $\rrap{DPO}{}{}$ as introduced in~\eqref{eq:rrapdef}). With shorthand notations $\bar{\rho}=\bar{\rho}^{DPO}_{\bA_{\CHEM{}}}$ for the restricted DPO-type rule algebra representation on $\hat{\bar{\bS}}_{\CHEM{}}$ and $\rrap{}{}{}\equiv\rrap{DPO}{}{}$, and recalling that we defined $\bar{R}_{+}:=[(r_{+},\overline{\ac{c}_{I_{+}}})]_{\overline{\sim}}\,$, $\hat{\bar{O}}_{+}:=\overline{\rho}(\overline{\delta}(r_{id},\overline{\ac{c}_{I_{+}}}))$ and $\hat{O}_{+}:=\overline{\rho}(\overline{\delta}(r_{id},\ac{true}))$, we thus have to perform the following computations (utilizing in particular~\eqref{eq:rrapRhoH} in the steps marked $(*)$ and the definition of $\rrap{}{}{}$ according to~\eqref{eq:rrapdef} in the steps marked $(**)$):
\begin{subequations}
\begin{align}
\begin{aligned}
[\hat{\bar{O}}_{+},{\color{blue}\hat{\bar{H}}_{+}} ]&= \hat{\bar{O}}_{+}{\color{blue}\hat{\bar{H}}_{+}} -  {\color{blue}\hat{\bar{H}}_{+}} \hat{\bar{O}}_{+}\overset{(*)}{=}\bar{\rho}\left(\rrap{}{\overline{\delta}\left(r_{id},\overline{\ac{c}_{I_{+}}}\right)}{{\color{blue}\overline{\delta}\left(r_{+},\overline{\ac{c}_{I_{+}}}\right)}}\right) 
-  \bar{\rho}\left(\rrap{}{{\color{blue}\bar{\delta}\left(r_{+},\overline{\ac{c}_{I_{+}}}\right)}}{\overline{\delta}\left(r_{id},\overline{\ac{c}_{I_{+}}}\right)}\right)
 \\
 &\overset{(**)}{=} 
 \sum_{%
 	\bar{\mu}
	\in\overline{%
		\RMatchGT{DPO}{(r_{id},\overline{\ac{c}_{I_{+}}})}{{\color{blue}r_{+},\overline{\ac{c}_{I_{+}}}}}}} \bar{\delta}\left(
 \compGT{DPO}{(r_{id},\overline{\ac{c}_{I_{+}}})}{\bar{\mu}}{{\color{blue}\,(r_{+},\overline{\ac{c}_{I_{+}}})}}
 \right)\\
 &\qquad -
  \sum_{%
 	\bar{\mu}
	\in\overline{%
		\RMatchGT{DPO}{
			{\color{blue}(r_{+},\overline{\ac{c}_{I_{+}}})}
			}{%
			r_{id},\overline{\ac{c}_{I_{+}}}
			}
			}} 
	\bar{\delta}\left(
 \compGT{DPO}{{\color{blue}(r_{id},\overline{\ac{c}_{I_{+}}})}}{\bar{\mu}}{\,(r_{+},\overline{\ac{c}_{I_{+}}})}
 \right)
 \end{aligned}\label{eq:modCA}\\
 \begin{aligned}
[\bar{O}_{+},{\color{blue}\hat{\bar{H}}_{+}} ]
&= \bar{O}_{+}{\color{blue}\hat{\bar{H}}_{+}} -  {\color{blue}\hat{\bar{H}}_{+}} \bar{O}_{+}
\overset{(*)}{=}\bar{\rho}\left(\rrap{}{\overline{\delta}\left(r_{id},\ac{true}\right)}{{\color{blue}\overline{\delta}\left(r_{+},\overline{\ac{c}_{I_{+}}}\right)}}\right) 
-  \bar{\rho}\left(\rrap{}{{\color{blue}\bar{\delta}\left(r_{+},\overline{\ac{c}_{I_{+}}}\right)}}{\overline{\delta}\left(r_{id},\ac{true}\right)}\right)
 \\
 &\overset{(**)}{=} 
 \sum_{%
 	\bar{\mu}
	\in\overline{%
		\RMatchGT{DPO}{(r_{id},\ac{true})}{{\color{blue}r_{+},\overline{\ac{c}_{I_{+}}}}}}} \bar{\delta}\left(
 \compGT{DPO}{(r_{id},\ac{true})}{\bar{\mu}}{{\color{blue}\,(r_{+},\overline{\ac{c}_{I_{+}}})}}
 \right)\\
 &\qquad -
  \sum_{%
 	\bar{\mu}
	\in\overline{%
		\RMatchGT{DPO}{
			{\color{blue}(r_{+},\overline{\ac{c}_{I_{+}}})}
			}{%
			r_{id},\ac{true}
			}
			}} 
	\bar{\delta}\left(
 \compGT{DPO}{{\color{blue}(r_{id},\ac{true})}}{\bar{\mu}}{\,(r_{+},\overline{\ac{c}_{I_{+}}})}
 \right)
 \end{aligned}\label{eq:modCB}
\end{align}
\end{subequations}
As an intermediate summary, our task thus amounts to computing a number of restricted rule compositions for rules with constraint-preserving conditions. 

\subsubsection{Assessment of adequacy of \MOD{} for computing restricted rule compositions}

A careful analysis of the current version of the \MOD{} platform and of the available rule composition algorithms available therein~\cite{andersen2018rule} allowed us to compare these algorithms to the DPO-type rule composition operation (Definition~\ref{def:Rcomp}) in the ambient category $\bA_{\CHEM{}}$ of a given \CHEM{} model. The most generic of these operations, denoted $\bullet$ in \MOD{}, has an analogous interpretation within DPO-rewriting theory as follows:
\begin{definition}\label{def:obMOD}
Let $\bA_{\CHEM{}}$ denote the ambient category of a \CHEM{} model, and let $r_j=(O_j\hookleftarrow K_j\hookrightarrow I_j)\in \Lin{\bA_{\CHEM{}}}$ ($j=1,2$) be two linear rules (without conditions). Then we define\footnote{For notational succinctness, we employ here a slight abuse of the notations introduced in Definition~\ref{def:Rcomp}: since every linear rule $r\in \Lin{\bA_{\CHEM}}$ without a condition is semantically equivalent in restricted rewriting to a rule $R=(r,\ac{true})$ with the same ``plain'' rule $r$ and a trivial application condition $\ac{true}$, and since composing rules that have trivial application conditions results only in rules that also have trivial conditions, it is consistent to let $\RMatchGT{DPO}{r_2}{r_1}:=\RMatchGT{DPO}{R_2}{R_1}$ and $\compGT{DPO}{r_2}{\mu}{r_1}:=\compGT{DPO}{R_2}{\mu}{R_1}$ (for $R_j:=(r_j,\ac{true})$, $j=1,2$), for the latter in addition interpreting the composite rule again as a ``plain'' rule (i.e., dropping the application condition $\ac{true}$).}
\begin{equation}
r_1\bullet r_2 := \biguplus_{\mu \in \RMatchGT{\bullet\:DPO}{r_2}{r_1}} \{\compGT{DPO}{r_2}{\mu}{r_1}\}\,.
\end{equation}
Here, the symbol $\biguplus$ is used to indicate that the operation renders a \emph{multiset} of composite rules $\compGT{DPO}{r_2}{\mu}{r_1}$, where DPO-type admissible matches $\mu$ are taken from the restricted set of matches $\RMatchGT{\bullet\:DPO}{r_2}{r_1}$ defined as
\begin{equation*}
\RMatchGT{\bullet\:DPO}{r_2}{r_1}:= \left\{
\mu=(I_2\hookleftarrow M\hookrightarrow O_1)\in \RMatchGT{DPO}{r_2}{r_1}\mid
(M\hookrightarrow O_1), (M\hookrightarrow I_2)\in \cM_{\mathsf{er}}
\right\}\,.
\end{equation*} 
The restriction of matches is thus taken to only consider spans of \emph{edge-reflecting $\cM$-morphisms} (denoted above $\cM_{\mathsf{er}}$).
An $\cM$-morphism $(m:A\hookrightarrow B)\in \cM$ in the category $\bA_{\CHEM}$ is defined to be \emph{edge-reflecting} iff for every edge $e'$ in $B$ such that the endpoint vertices of $e'$ are all in the codomain of $m$, there exists an edge $e$ in $A$ with $m(e)=e'$ (which since $m$ is a monomorphism entails in particular that $e$ must have the same number of endpoints as $e'$). Using notations as in Definition~\ref{def:ugraph} for untyped undirected graphs, and with $m=(m_V:V_A\hookrightarrow V_B,m_E:E_A\hookrightarrow E_B)$, this statement may be expressed more formally as  $m\in \cM_{\mathsf{er}}$ iff $m\in \cM$ and
\begin{equation}
\begin{aligned}
&\forall e'\in E_B: \big(\forall v'\in i_B(e'): \exists v\in V_A: m_V(v)=v'\big)\\
&\qquad\qquad  \Rightarrow \big(\exists e\in E_A: m_E(e)=e' \land i_B(e')=\cP^{(1,2)}(m_V)\circ i_A(e)\big)\,.
\end{aligned}
\end{equation}
\end{definition}

\begin{remark}
The seemingly ad hoc nature of the above definition may be understood more transparently in light of its algorithmic counterpart in \MOD{},
where the rule composition algorithm for the operation $\bullet$ is implemented based upon \emph{McGregor's vertex-induced common subgraph algorithm}\footnote{\url{https://www.boost.org/doc/libs/1_74_0/libs/graph/doc/mcgregor_common_subgraphs.html}}
(in the form available via the \textsc{Boost Graph} library~\cite{10.5555/504206}). %
Incidentally, using this rule composition operator, the simplicity constraint of chemical graphs are in many practically relevant examples typically not violated.
To illustrate why spans of edge-reflecting $\cM$-morphisms provide such a feature, consider for illustration the following two pushout diagrams in $\mathbf{uGraph}$ (where colors are used to indicate the mapping of vertices):
\begin{equation}
\ti{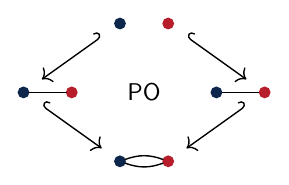}\qquad\qquad
\ti{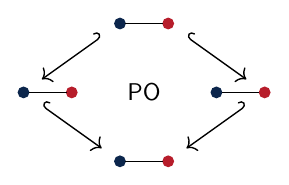}
\end{equation}
The span in the diagram on the left is found to \emph{not} consist of edge-reflecting monomorphisms, in contrast to the one in the diagram on the right.
More generally, as a feature akin to the well-known analogous feature in directed simple graphs,
the pushout of a span of edge-reflecting monomorphisms in $\mathbf{uGraph}$ with each of the three objects of the span being a simple graph is guaranteed to yield again a simple graph
(and analogously in the setting for typed undirected graphs such as $\bA_{\CHEM{}}$).
Consequently, the choice for the construction of $\bullet$ according to Definition~\ref{def:obMOD} implies that for any two rules $r_1,r_2\in \Lin{\bP_{\CHEM{}}}$
in the \emph{pattern category} $\bP_{\CHEM{}}$ (for which thus $O_j$, $K_j$ and $I_j$ ($j=1,2$) are in particular simple undirected typed graphs),
at least the pushout object of an admissible match $\mu\in \RMatchGT{\bullet\:DPO}{r_2}{r_1}$ is guaranteed to be a simple undirected typed graph as well.
However, this is not a guarantee that the composite rule $\compGT{DPO}{r_2}{\mu}{r_1}$ will be an element of $\Lin{\bP_{\CHEM{}}}$, %
which is why in \MOD{} any composite rule construction is aborted ``on-the-fly'' if it gives rise to non-simple graphs.
\end{remark}

By definition, the operation $\bullet$ evidently does not cover all contributions possible in the full DPO-type rule composition semantics, and it is moreover possible to verify via counter-examples that $\bullet$ is neither a unital nor an associative operation (in the sense of the associativity and concurrency theorems). Another evident problem in view of implementing faithfully the restricted rewriting operation of rule composition is the evident lack of taking into account the structural constraints and constraint-preserving completions of application conditions. Therefore, it is clear a priori that $\bullet$ will not be algorithmically sufficient in order to compute the commutators of the form relevant for differential pattern-counting moment-semantics according to Theorem~\ref{thm:CTMCmev}. In order to properly assess the necessary extensions to \MOD{} in order to resolve these issues, it is nevertheless of interest to work through a non-trivial example of a rule composition analysis via the operator $\bullet$, which we will present in the remainder of this section.

Upon closer inspection of the definition of the DPO-type rule composition (Definition~\ref{def:Rcomp}),
we find that a span $\mu=(I_2\hookleftarrow M\hookrightarrow O_1)$ of $\cM$-morphisms is a DPO-admissible match of rules $R_j = (r_j,\ac{c}_{I_j})$ ($j=1,2$)
if it is an admissible match of the ``plain'' rules (i.e., more explicitly, of $R^{\mathsf{plain}}_j := (r_j,\ac{true})$, $j=1,2$),
and if in addition the condition $\ac{c}_{I_{21}}$ of the composite rule satisfies $\ac{c}_{I_{21}}\, \not\!\!\dot{\equiv} \,\ac{false}$.
We may thus split the restricted DPO-type rule composition computations into a part where we pre-compute all ``plain'' rule compositions,
and a second part in which we compute for each composite rule its condition $\ac{c}_{I_{21}}$,
potentially in this step having to discard some of these compositions due to $\ac{c}_{I_{21}}\, \dot{\equiv}\, \ac{false}$.\\

\noindent\textbf{``Plain'' rule compositions via the $\bullet$ operator of \MOD{}.} According to Definition~\ref{def:obMOD}, we may utilize the \MOD{} algorithm for the rule composition operation $\bullet$ to compute at least the subset of all possible DPO-type ``plain'' rule compositions, namely those along vertex-induced common subgraph of the relevant rule interfaces.
We thus find for the compositions $r_{id} \bullet r_+$ (resp.\ $r_+ \bullet r_{id}$) an overall number of 18 (resp.~9) possible composed rules. In Figure~\ref{fig:modcomp} (which was auto-generated with \MOD{}), an explicit example of a rule composition of $r_{id} \bullet r_+$ is depicted for illustration. The dashed red lines indicate the match used for the construction of the composed overall rule (bottom of the figure).

Utilizing another core algorithmic functionality of \MOD{}, we further analyzed the sets of composite rules in terms of isomorphism classes of rules, and thus in particular convert the multisets rendered by the operation $\bullet$ into a set of non-isomorphic rules with a \emph{count} for the number of occurrences of each isomorphism class in $r_{id} \bullet r_+$ (resp.\ $r_+ \bullet r_{id}$). The result of this analysis is presented in Table~\ref{tab:res}. We find that in this particular example, each isomorphism class appears at most once in either set of composite rules, and that 6 of the isomorphism classes occur once in \emph{both} sets.

In \ref{app:modcode} we list a short Python code fragment in order to illustrate how rules are composed with the \MOD{} framework. %
The example code used for this section can be accessed from \url{https://cheminf.imada.sdu.dk/papers/tcs-2021/}. %
The source code repository for \MOD{} can be found at \url{https://github.com/jakobandersen/mod}.\\

We directly compared the results of the \MOD{}-based computation with a manual computation of the DPO-type compositions of $r_{id}$ with $r_+$ and of $r_+$ with $r_{id}$, finding that apart from the compositions along trivial overlaps, the \MOD{} computation as explained via the definition of $\bullet$ indeed missed an additional three possible ``plain'' rule compositions; these are labeled ${\color{red}r_A}$, ${\color{red}r_B}$ and ${\color{red}r_C}$ in Table~\ref{tab:res}. With $O_{+}$ denoting the output pattern of $r_{+}$ and $I_{+}$ the one of $r_{id}$, and taking the labeling scheme for atoms in Figure~\ref{fig:modcomp} as a convention for labeling individual atoms in these patterns, we may compactly encode the three relevant admissible matches ${\color{red}\mu_A},{\color{red}\mu_B},{\color{red}\mu_C}\in \MatchGT{DPO}{r_{id}}{r_{+}}$ of $r_{id}$ into $r_{+}$ as\footnote{For ${\color{red}\mu_C}$, the partial overlap also includes the double bond present between two of the pairs of overlapped vertices, which is not explicitly mentioned in the equation for brevity.}
\begin{equation}
{\color{red}\mu_A}=I_{+}\xleftharpoonup{\substack{%
	\langle0\rangle \leftarrow  \langle 0\rangle\\
	\langle5\rangle \leftarrow  \langle 5\rangle}}O_{+}\,,\quad
{\color{red}\mu_B}=I_{+}\xleftharpoonup{\substack{%
	\langle0\rangle \leftarrow  \langle 5\rangle\\
	\langle5\rangle \leftarrow  \langle 0\rangle}}O_{+}\,,\quad
{\color{red}\mu_C}=I_{+}\xleftharpoonup{\substack{%
	\langle0\rangle \leftarrow  \langle 0\rangle\\
	\langle5\rangle \leftarrow  \langle 1\rangle\\
	\langle4\rangle \leftarrow  \langle 2\rangle}}O_{+}\,. \label{eq:modMO}
\end{equation}

\begin{figure}
\centering
\mi{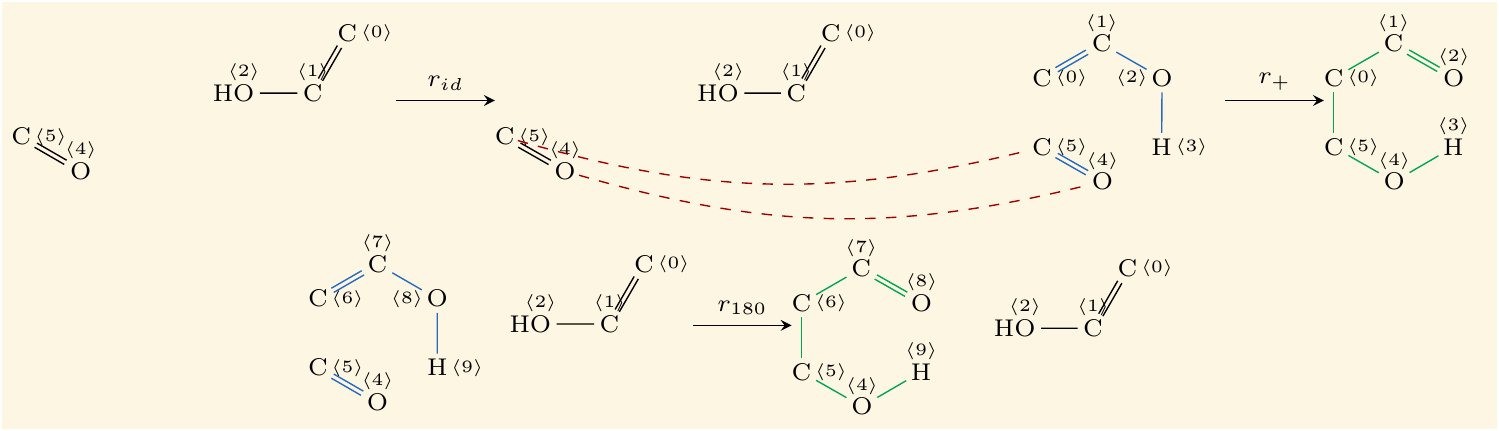}
\caption{One of 18 possible compositions resulting from $r_{id} \bullet r_+$;  top: DPO rules $r_{id}$ and $r_+$, the dashed red lines illustrate the match used for the inference of the composed rule $r_{180}$ (bottom); all non-isomorphic compositions of $r_{id} \bullet r_+$ and $r_+ \bullet r_{id}$ are listed in \ref{app:comp}}
\label{fig:modcomp}
\end{figure}%

\noindent\textbf{Application condition computations for composite rules. } Before diving into the details of this computation, it is worthwhile noting that due to the particularly simple structure of the condition $\overline{\ac{c}_{I_{+}}}$,
one may in fact determine the constraint-preserving application conditions $\overline{\ac{c}_{\bar{I}}}$ for each composite rule $(\bar{r}, \overline{\ac{c}_{\bar{I}}})$ that contributes to \eqref{eq:modCA} and~\eqref{eq:modCB}.
To this end, let us denote by $\langle 0\rangle$ and $\langle5\rangle$ the labels that mark carbon atoms on the input (i.e., left-hand side) of $r_{+}$ in the depiction of Figure~\ref{fig:modcomp} (top right part) over which the condition $\overline{\ac{c}_{I_{+}}}$ formulates the non-existence constraint of bonds (compare~\eqref{eq:acRplus}); again referring to Figure~\ref{fig:modcomp}, let the same labels mark the carbon atoms on the input (i.e., left-hand side) of $r_{id}$ (top left part of the figure) which carry the non-bonding condition for the rule $\bar{O}_{+}:=(r_{id}, \overline{\ac{c}_{I_{+}}})$ relevant to construct the observable $\hat{\bar{O}}_{+}$. As both $r_{+}$ and $r_{id}$ \emph{preserve vertices}, performing a DPO-type composition along an admissible match of these ``plain'' rules (such as the one depicted in Figure~\ref{fig:modcomp} as dashed lines from the output of $r_{id}$ to the input of $r_{+}$), it is possible to ``trace'' each vertex of the input of interfaces of the two ``plain'' rules to their images in the input interface $\bar{I}$ of the composite rule.

With these preparations, and applying (a typed version of) Lemma~\ref{lem:Bridges}, we find that the contribution of the condition $\overline{\ac{c}_{I_{+}}}$ of the rule $\bar{R}_{+}$ to the conditions $\overline{\ac{c}_{\bar{I}}}$ of the composite rules in either order of composition of $\bar{R}_{+}$ with $O_{+}:=(r_{id},\ac{true})$ or $\bar{O}_{+}:=(r_{id},\overline{\ac{c}_{I_{+}}})$ is simply a negative constraint expressing that the images in $\bar{I}$ (the input interface of the composite rule) of the carbon-type vertices of $I_{+}$ marked $\langle0\rangle$ and $\langle5\rangle$ must not be linked by any bond. Analogously, when considering compositions of $\bar{R}_{+}$ with $\bar{O}_{+}$ in either order, the contribution of the application condition of the rule $\bar{O}_{+}$ to the condition $\overline{\ac{c}_{\bar{I}}}$ of a composite rule is that the images in $\bar{I}$ of the carbon-type vertices marked $\langle0\rangle$ and $\langle5\rangle$ in the input interface of $r_{id}$ must not be linked. For instance, considering the rule composition depicted in Figure~\ref{fig:modcomp}, in which the ``plain'' rule labeled $r_{180}$ is obtained as a particular composition of the rule $r_{+}$ with the rule $r_{id}$, the contribution of $\overline{\ac{c}_{I_{+}}}$ as a condition on $r_{+}$ results in a constraint on the input interface $I_{180}$ of $r_{180}$ that the carbon vertices marked $\langle5\rangle$ and $\langle6\rangle$ in $I_{180}$ must not be linked by any bond, while for the case where we also have a condition $\overline{\ac{c}_{I_{+}}}$ on the input interface of $r_{id}$ (i.e., when computing the possible rule compositions of $\bar{R}_{+}$ with $\bar{O}_{+}$), we obtain an additional constraint that the carbon vertices marked $\langle0\rangle$ and $\langle5\rangle$ in $I_{180}$ must not be linked by any bond.

A tedious manual analysis of all possible rule compositions of $\bar{R}_{+}$ with $O_{+}$ either $\bar{O}_{+}$ in both possible orders reveals that
\begin{enumerate}[label=(\roman*)]
\item The matches resulting in the composite ``plain'' rules labeled ${\color{red}r_A}$, ${\color{red}r_B}$ and ${\color{red}r_C}$ in Table~\ref{tab:res} are all disqualified as admissible matches of $\bar{O}_{+}$ into $\bar{R}_{+}$ (i.e., due to the application conditions of the composite rules evaluating to $\ac{false}$), while they are in contrast admissible matches of $O_{+}$ into $\bar{R}_{+}$.
\item The case of composite rule $r_{180}$ in Table~\ref{tab:res} (depicted in Figure~\ref{fig:modcomp} to arise as a particular composite of ``plain'' rules $r_{+}$ with $r_{id}$) is particularly interesting, since while there exists one admissible match in each order of composition of $r_{+}$ and $r_{id}$ to yield a composite rule in the isomorphism class of $r_{180}$, the underlying admissible matches yield \emph{non-equivalent} composite rules when composing $\bar{R}_{+}$ and $\bar{O}_{+}$. In contrast, the respective composites of $\bar{R}_{+}$ and $O_{+}$ are in fact equivalent as rules with conditions.
\end{enumerate}

\subsubsection{Discussion}

In summary, as the analysis of the (constraint-preserving completions of) application conditions of the composite rules reveals, it is indeed not possible to utilize the \MOD{} algorithms in their present development state in order to compute commutators for the purpose of deriving pattern-counting observable moment-evolution equations, even though the ``plain'' rule compositions obtained via the \MOD{} operation $\bullet$ are in fact for the particular example considered almost exhaustive. The restriction of partial overlaps contributing when computing with $\bullet$ to spans of edge-reflecting $\cM$-morphism is for the example at hand found to partially emulate the semantics of the constraint-preserving application condition $\overline{\ac{c}_{I_{+}}}$ in the compositions of $\bar{R}_{+}$ with either $O_{+}$ or $\bar{O}_{+}$, yet it is in fact impossible to fully reproduce the correct rule composition semantics necessary. In particular, as the example of the rule $r_{180}$ in Table~\ref{tab:res} discussed in the previous section highlights, one indeed requires an implementation the full calculus of the restricted DPO-type rule composition in order to correctly classify the equivalence classes of composite rules that arise in either order of composition or two given rules with conditions (yielding ultimately the desired implementation of the computation of \emph{commutators} for differential semantics).

\subsubsection{Perspective: ``convenience constraints''}

Beyond the elementary necessity of a full-fledged implementation of restricted rewriting rule composition algorithms, our case study led to the discovery of a few additional avenues for future work worth exploring, all of which related to the possibilities offered by the ability to endow chemical rules with ``convenience constraints'' (i.e., constraints $\ac{c}_I$ on a chemical rule that are chosen on in addition to the constraint-preserving conditions necessary to ensure the chemical validity of the rule).  Returning once more to our case study, the aldol addition reaction is an interesting case as the specifics of how ``the'' chemistry will limit the number of possible compositions based on additional \emph{convenience constraints} (i.e., constraints in addition to the type graph that are based on the specifics of the chemical system to be analyzed). 
The rule as depicted in Figure~\ref{fig:aldol} is arguably too generic due to two reasons:
\begin{enumerate}[label=(\roman*)]
\item In the aldol addition reaction, any carbon atoms adjacent to an oxygen atom are usually constrained to have only a single oxygen atom neighbor.
Otherwise the rule would, e.g., also allow to match on carboxyl groups of molecules (i.e., carbon atoms with two oxygen atom neighbors).
\item Under realistic chemical conditions where an aldol addition takes place, it is very unlikely that a carbon atom is found to have two incident double bonds.
\end{enumerate}
Endowing the ``plain'' aldol addition rule $r_{+}$ in addition to the constraint-preserving application condition with a condition based upon these two convenience constraints, one finds (cf.\ Table~\ref{tab:res}) that this modified aldol addition rule possesses far fewer admissible matches into the rule $O_{+}$, i.e., only 6 (resp.\ 2) compositions for $r_{id} \bullet r_+$  (resp.\ $r_+ \bullet r_{id}$) remain valid as compositions of the aforementioned rules with conditions. Upon closer inspection of the data of Table~\ref{tab:res}, we find that each of the 2 aforementioned rules in $r_+ \bullet r_{id}$ are isomorphic to one of the 6 rules in $r_{id} \bullet r_+$, whence Table~\ref{tab:res} a tentative commutator computation would finally result in just 4 non-isomorphic rules (all with occurrence count 1), a remarkable reduction in complexity in view of the 30 non-trivial contributions in the unrestricted setting.

We are thus led to suspect that the ability offered by our novel framework of chemical rewriting as a restricted rewriting theory (and thus in particular the ability to freely endow rules with ``convenience constraints'') will not only pose an intriguing option to formally encode practical knowledge in organic chemistry, but might eventually prove quintessential in deriving meaningful differential semantics for organo-chemical reaction systems in the first place. 

\section{Conclusion and outlook}

Rewriting theories of DPO- and SqPO-type for rules with conditions over $\cM$-adhesive categories
are poised to provide a rich theoretical and algorithmic framework for modeling stochastic dynamical systems in the life sciences.
The main result of the present paper consists in the introduction of a \emph{rule algebra framework} that extends the pre-existing constructions~\cite{bp2018,nbSqPO2019,bdg2016}
precisely via incorporating the notion of conditions. %
The sophisticated \KAP{}~\cite{Boutillier:2018aa} and \MOD{}~\cite{Andersen_2016} bio-/organo-chemistry frameworks and related developments
have posed one of the main motivations for this work. %
We introduce in this paper the first-of-its-kind fully faithful encoding of bio- and organo-chemical rewriting systems in terms of our novel original universal theory of rewriting-based CTMC semantics. %
More specifically, we provide a formulation of chemistry as a restricted rewriting theory over certain $\cM$-adhesive categories of typed undirected multigraphs, thus in particular permitting to establish a rigorous and original CTMC theory for organic chemistry. The encoding as restricted rewriting theory will be beneficial also in the development of tracelet-based techniques~\cite{behr2019tracelets}, and is current work in progress.
In order to achieve a complete algorithmic implementation of ODEs for moments of pattern-counting observables, our fully-worked and non-trivial example based on the \MOD{} platform illuminates the missing technical ingredients, namely the implementation of compositions stemming from common subgraphs that are not induced and the implementation of no-edge constraints.

An intriguing perspective for future developments in categorical rewriting theory consists in developing a robust and versatile methodology
for the analysis of ODE systems of pattern-counting observables in stochastic rewriting systems.
While the results of this paper permit to formulate dynamical evolution equations for arbitrary higher moments of such observables,
in general cases (as illustrated in Section~\ref{sec:bcgr}) the non-closure of the resulting ODE systems remains a fundamental technical challenge.
In the \KAP{} literature, sophisticated conceptual and algorithmic approaches to tackle this problem have been developed such as refinements~\cite{danos2008rule,Danos2014},
model reduction techniques~\cite{Danos_2010} and stochastic fragments~\cite{ferethal00975861} (see also~\cite{bdg2019} for an extended discussion).
We envision that a detailed understanding of these approaches from within the setting of categorical rewriting and of rule algebra theory
could provide a very fruitful enrichment of the methodology of rewriting theory and of algorithmic cheminformatics alike.

\section*{Funding}
This work is supported by the Novo Nordisk Foundation grant NNF19OC0057834 and by the Independent Research Fund Denmark, Natural Sciences, grants DFF-0135-00420B and DFF-7014-00041.

\clearpage
\appendix

\section{Comparison with the ICGT 2020 conference paper version}\label{sec:confToExt} 

The present paper is an extended journal version of our ICGT 2020 conference paper~\cite{BK2020}, with additional materials implemented as follows:
\begin{itemize}
\item While~\cite{BK2020} already contained the general theory of rule algebras, their representations and the stochastic mechanics frameworks  for rewriting over $\cM$-adhesive categories with conditions in both DPO- and SqPO-semantics, the present paper provides in addition an important specialization of this theory to the setting of \emph{restricted rewriting} (Section~\ref{sec:rrt}). The specialization assumes a \emph{global structural constraint} on all objects over which the rewriting is to be performed, which is precisely the case in many important application examples (including in particular the bio- and organo-chemical rewriting theories featured in this paper). Building upon the notion of \emph{constraint-preserving conditions} as introduced by Habel and Pennemann~\cite{habel2009correctness}, we introduce a novel formalism of rule-algebraic calculus under the assumption of global constraints, which permit a critical improvement over the general variant in terms of the complexity of application constraints and rule compositions.
\item Another original result of~\cite{BK2020} was the encoding of the biochemistry platform language \KAP{}~\cite{danos2004computational,danos2004formal}  in terms of a rewriting theory with a certain type of structural constraints in SqPO-semantics. The present paper provides an extended review and comparison with the original \KAP{} encoding, highlighting several important technical points. Referring to Section~\ref{sec:bcgr} for further details, the key achievement of our rewriting-theoretic encoding is identified as a streamlining of the calculus of ODEs for the moments of \KAP{} pattern-count observables, which in the original formulation of \KAP{} required (despite the origins of \KAP{} as a stochastic rewriting formalism) a highly domain-specific set of constructions of a deeply intricate algorithmic nature. Via our novel theory, we demonstrate that \KAP{} is an instance of SqPO-type restricted rewriting theory, and thus the aforementioned ODEs may be derived in a transparent fashion from our uniform theory of CTMC semantics for rule-based systems (with restricted rewriting theory providing a succinct formulation of \KAP{} rules). %
\item In our conference paper~\cite{BK2020}, we had proposed that the algorithmic approach to organic chemistry as implemented in the \MOD{} cheminformatics platform~\cite{Andersen_2016} might give rise to yet another practically highly relevant instance of a restricted rewriting theory,thus promising to open novel possibilities in terms of CTMC calculus in this setting. %
The present paper achieves this vision (as a result of joint work with J.L.\ Andersen and D.\ Merkle, co-authors of this extended journal version),
providing the first-of-its-kind implementation of a theoretical framework for organo-chemical rewriting systems. %
Referring to Section~\ref{sec:ocgr} for further details, our novel encoding permits to specialize models for chemistry to be faithfully encoded as DPO-type restricted rewriting theories, %
and that the somewhat ad hoc definition of various rule composition operations implemented algorithmically in \MOD{} fall in fact
under the general umbrella of the standard DPO-type rule composition operation for rules with conditions. %
Crucially, our novel formulation of a specialized version of organic chemistry as a restricted rewriting theory over an $\cM$-adhesive category of typed undirected multigraphs permits to leverage our universal rule-algebraic CTMC theory in order to obtain a faithful encoding of \emph{stochastic dynamics} of organo-chemical reaction systems, which is an original result of this paper. %
As illustrated in Section~\ref{sec:ocgr} via a fully worked example, we are able to identify clearly which particular aspects of the \MOD{} framework will require some extensions (mostly in terms of rule composition and application condition algorithms) in future work
in order to realize an algorithmic implementation of \emph{differential semantics} for organo-chemical reaction systems. %
From a purely theoretical standpoint, the results of Section~\ref{sec:ocgr} may moreover be interpreted as a poster-example of a restricted rewriting theory with structural constraints of considerable complexity, yet for which (certain practically relevant classes of) rewriting rules with conditions may nevertheless be remarkably compactly presented in the constraint-preserving completion form that is used in restricted rule algebra calculus.
\end{itemize}

\section{Background material on adhesive categories and rewriting with conditions}\label{sec:appendixA}

As a reference for notational conventions and in order to recall some of the standard definitions necessary in the main text, we collect here some of the materials contained in our recent paper~\cite{behrRaSiR} for the readers' convenience.

\subsection{$\cM$-adhesive categories}\label{sec:MACapp}

\begin{definition}
An $\cM$-adhesive category~\cite{ehrig2010categorical} $(\bfC,\cM)$ is a category $\bfC$ together with a class of monomorphisms $\cM$ that satisfies the following properties:
\begin{enumerate}
\item $\bfC$ has pushouts and pullbacks along\footnotemark $\cM$-morphisms.
\item The class $\cM$ contains all isomorphisms and is stable under pushout, pullback and composition.
\item Pushouts along $\cM$-morphisms are $\cM$-van Kampen squares.
\end{enumerate}
\begin{minipage}[t]{0.6\linewidth}
The latter property entails that in a commutative diagram such as the one on the right where the bottom square is a pushout along an $\cM$-morphism, where the back and right faces pullbacks and where all vertical morphisms are in $\cM$, the bottom square is $\cM$-van Kampen if the following property holds: the top square is a pushout if and only if the front and left squares are pullbacks.
\end{minipage}%
\begin{minipage}[t]{0.4\linewidth}
\vspace{-1.5em}
\null\hfill\\[-\dimexpr\baselineskip+1.2em\relax]
\centering
\text{$\;$}\includegraphics{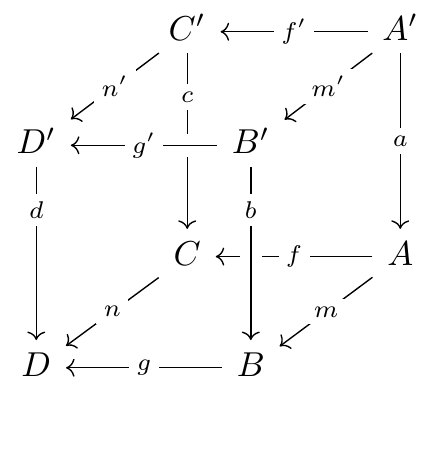}
\end{minipage}
\end{definition}
\footnotetext{Here, ``along'' entails that at least one of the two morphisms involved in the relevant (co-)span is in $\cM$.}%
Throughout the following definitions, let $(\bfC,\cM)$ be an $\cM$-adhesive category.
\begin{definition}
    $(\bfC,\cM)$ is said to be \emph{finitary}~\cite{GABRIEL_2014} if every object has only finitely many $\cM$-subobjects up to isomorphism.
\end{definition}
\begin{definition}
    $(\bfC,\cM)$ possesses an $\cM$-initial object $\mIO$~\cite{GABRIEL_2014} if for all objects $X\in \obj{\bfC}$ there exists a unique $\cM$-morphism $\iota_X:\mIO\hookrightarrow X$.
\end{definition}
\begin{definition}
    $(\bfC,\cM)$ possesses an \emph{epi-$\cM$-factorization}~\cite{habel2009correctness} if every morphism $f\in \mor{\bfC}$ factorizes as $f=m\circ e$ with $m\in \cM$ and with $e\in \epi{\bfC}$ an epimorphism, and such that this factorization is unique up to isomorphism.
\end{definition}
\begin{definition}
    $(\bfC,\cM)$ has \emph{$\cM$-effective unions} if for every cospan $(B\hookrightarrow D\hookleftarrow C)$ of $\cM$-morphisms that is the pushout of a span $(B\hookleftarrow A\hookrightarrow C)$, the following property holds: for every cospan $(B\hookrightarrow E\hookleftarrow C)$ whose pullback is given by $(B\hookleftarrow A\hookrightarrow C)$, the morphism $D\rightarrow E$ that exists by universal property of the pushout is in $\cM$.
\end{definition}

We next recall the notion of final pullback complements that is an important technical ingredient of the theory of SqPO-rewriting.

\begin{definition}\label{def:FPC}
    Let $(b,a)$ be a composable pair of morphisms in a category $\bfC$. Then a pair of morphisms $(c,d)$ is called a \emph{final pullback complement (FPC)}~\cite{Corradini_2006} if $(a,d)$ is the pullback of $(b,c)$, and if for every $(a\circ p,q)$ that is the pullback of $(b,r)$, there exists a morphism $s$ such that $r=c\circ s$ that is unique up to isomorphism.
    \[
        \includegraphics{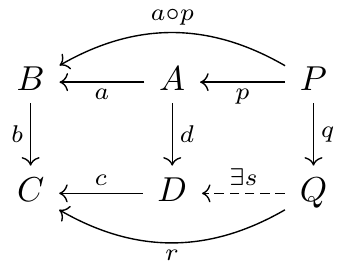}
    \]
\end{definition}

\begin{definition}
    The class of monomorphisms $\cM$ of $(\bfC,\cM)$ is said to be \emph{stable under FPCs}~\cite{behrRaSiR} if for every pair $(b,a)$ of composable $\cM$-morphisms the FPC $(c,d)$ (if it exists) is a pair of $\cM$-morphisms.
\end{definition}

\subsection{Concurrency and associativity theorems}\label{app:ACthms}

In the statements of the following two theorems, we always imply choosing concrete representatives of the relevant equivalence classes of rules with conditions in order to list the sets of admissible matches.

\begin{theorem}[Concurrency theorem~\cite{bp2018,nbSqPO2019,behrRaSiR}]\label{thm:concur}
    Let $\bfC$ be a category satisfying Assumption~\ref{as:main}, and let $\bT\in\{DPO,SqPO\}$. Then there exists a bijection $\varphi:A\xrightarrow{\cong}B$ on pairs of $\bT$-admissible matches between the sets $A$ and $B$,
    \begin{equation}
        \begin{aligned}
            A&=\{(m_2,m_1)\mid m_1\in \MatchGT{\bT}{R_1}{X_0}\,,; 
            m_2\in \MatchGT{\bT}{R_2}{X_1}\}\\
            \cong\quad 
            B&=\{(\mu_{21},m_{21})\mid \mu_{21}\in \MatchGT{\bT}{R_2}{R_1}\,,\; m_{21}\in \MatchGT{\bT}{R_{21}}{X_0}\}\,,
        \end{aligned}
        \end{equation}
    where $X_1=R_{1_{m_1}}(X_0)$ and $R_{21}=\compGT{\bT}{R_2}{\mu_{21}}{R_1}$ such that for each corresponding pair $(m_2,m_1)\in A$ and $(\mu_{21},m_{21})\in B$, it holds that
        \begin{equation}
            R_{21_{m_{21}}}(X_0) \cong
            R_{2_{m_2}}(R_{1_{m_1}}(X_0))\,.
        \end{equation}
\end{theorem}

\begin{theorem}[Associativity of rule compositions~\cite{bp2018,nbSqPO2019,behrRaSiR}]\label{thm:assocR}
    Let $\bfC$ be a category satisfying Assumption~\ref{as:main}. let $R_1,R_2,R_3\in \LinAc{\bfC}$ be linear rules with conditions, and let $\bT\in\{DPO,SqPO\}$. Then there exists a bijection $\varphi:A\xrightarrow{\cong} B$ of sets of pairs of $\bT$-admissible matches $A$ and $B$, defined as
    \begin{equation}
        \begin{aligned}
            A&:=\{(\mu_{21},\mu_{3(21)})\mid \mu_{21}\in 
            \MatchGT{\bT}{R_2}{R_1}\,,\; \mu_{3(21)}\in \MatchGT{\bT}{R_3}{R_{21}}\}\\
            B&:=\{(\mu_{32},\mu_{(32)1})\mid 
            \mu_{32}\in\MatchGT{\bT}{R_3}{R_2}\,,\;
             \mu_{(32)1}\in \MatchGT{\bT}{R_{32}}{R_1}\}\,,
        \end{aligned}
        \end{equation}
        where $R_{21}=\compGT{\bT}{R_2}{\mu_{21}}{R_1}$ and $R_{32}=\compGT{\bT}{R_3}{\mu_{32}}{R_2}$, such that for each corresponding pair $(\mu_{21},\mu_{3(21)})\in A$ and %
        $\varphi(\mu_{21},\mu_{3(21)})=(\mu_{32}',\mu_{(32)1}')\in B$, 
        \begin{equation}
            \compGT{\bT}{R_3}{\mu_{3(21)}}{\left(\compGT{\bT}{R_2}{\mu_{21}}{R_1}\right)}\cong
            \compGT{\bT}{\left(\compGT{\bT}{R_3}{\mu_{32}'}{R_2}\right)}{\mu_{(32)1}'}{R_1}\,.
        \end{equation}
    In this particular sense, the composition operations $\compGT{\bT}{.}{.}{.}$ are \textbf{associative}.
\end{theorem}

\section{Proofs}
\subsection{Proof of Theorem~\ref{thm:canrep}}\label{app:proofCanrep}

The statement of the theorem is equivalent to the following two properties:
    \begin{equation*}
        \begin{aligned}
        (i)\;&  &      \canRep{\bT}{\delta(R_{\mIO})}&=Id_{End_{\bR}(\hat{\bfC})}\\
        (ii)\;& &\forall R_1,R_2\in\LinEq{\bfC}:\quad
        \canRep{\bT}{\delta(R_{2})}\canRep{\bT}{\delta(R_{1})}&=
        \canRep{\bT}{\rap{\bT}{\delta(R_{2})}{\delta(R_{1})}}\,.
        \end{aligned}
    \end{equation*}
    By linearity, it suffices to verify these properties on an arbitrary basis vector $\ket{X}\in\hat{\bfC}$. For $(i)$, it suffices to verify that
    \[
        \canRep{\bT}{\delta(R_{\mIO})}\ket{X}=\sum_{m\in\MatchGT{\bT}{R_{\mIO}}{X}}\ket{R_{\mIO_m}(X)}=\ket{X}\,.
    \]
    Property $(ii)$ is a consequence of Theorem~\ref{thm:concur} (the Concurrency Theorem):
    \begin{align*}
        \canRep{\bT}{\delta(R_{2})}\canRep{\bT}{\delta(R_{1})}\ket{X}&=
        \sum_{m_1\in\MatchGT{\bT}{R_{1}}{X}}
        \sum_{m_2\in\MatchGT{\bT}{R_{2}}{R_{1_{m_1}}(X)}}\ket{R_{2_{m_2}}(R_{1_{m_1}}(X))}\\
        &=
        \sum_{\mu\in\MatchGT{\bT}{R_2}{R_1}}
        \sum_{m_{21}\in\MatchGT{\bT}{R_{2_{\mu}1}}{X}}
        \ket{R_{2_{\mu}1_{m_{21}}}(X)}\,.
    \end{align*}

\subsection{Proof of Theorem~\ref{thm:CTMCs}}
\label{sec:CTMCproofsApp}

\paragraph{Ad 1.:} It suffices to verify that direct derivations along a rule $R$ of the relevant form occurring in the two types of observables from any object $X$ satisfy $R_m(X)\cong X$. But this follows directly from the respective definitions of direct derivations.

\paragraph{Ad 2. \& 3.:} It again suffices to verify these properties on basis elements $\ket{X}$ of $\hat{\bfC}$, and for generic $R\in\LinEq{\bfC}$. By definition,
\begin{equation}
\bra{}\canRep{\bT}{\delta(R)}\ket{X}=\sum_{m\in \MatchGT{\bT}{R}{X}}\underbrace{\braket{}{R_m(X)}}_{=1_{\bR}}
=\vert \MatchGT{\bT}{R}{X}\vert\,.
\end{equation}
In both cases of semantics, a candidate match of $R$ into $X$ must satisfy the application condition. In the DPO case, in addition the relevant pushout complement must exist. Combining these facts allows to verify the formulae for $\jcOp{.}$.

\paragraph{Ad~4.:} The proof is straightforward generalization of the corresponding statement for the case of rewriting rules without conditions~\cite{bp2019-ext,nbSqPO2019}. Following standard continuous-time Markov chain (CTMC) theory~\cite{norris}, one may verify that the linear operator $\cH$ has a strictly negative coefficient diagonal contribution $\jcOp{H}$, a non-negative coefficient off-diagonal contribution $H$, thus $\cH$ satisfies $\bra{}\cH=0$. Since in addition a given $X\in\obj{\bfC}_{\cong}$ may be rewritten via direct derivations along the rules of the transition set only in finitely many ways, in summary $\cH$ fulfills all requirements to qualify as a conservative and stable $Q$-matrix (i.e., an infinitesimal generator) of a CTMC (cf.\ \cite{nbSqPO2019} for further  details).

\subsection{Proof of Lemma~\ref{lem:Bridges}}\label{app:proof-lem-bridges}

Consider first the specialization of the algorithmic definition of the $\Shift$ operation as provided in Theorem~\ref{thm:STdefns} the case of a (non-nested) negative application condition, i.e., $\Shift(X{\color{blue}\hookrightarrow Y}, \neg\exists (X{\color{red}\hookrightarrow N}))$, utilizing~\eqref{eq:ShiftAlgo-iter}, \eqref{eq:ShiftAlgo-triv} and~\eqref{eq:ShiftAlgo-neg}:
\begin{equation}
\vcenter{\hbox{\ti{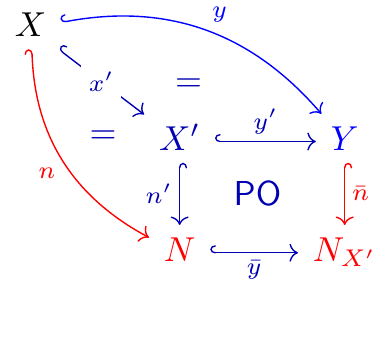}}}\qquad 
\vcenter{\hbox{$\begin{array}{rl}
\Shift(X{\color{blue}\hookrightarrow Y}, \neg\exists (X{\color{red}\hookrightarrow N}))
&:= 
\bigwedge\limits_{\substack{(n',x',y')\in \cM^{\times\:3}\\ n'\circ x'= n\land y'\circ x'= y}}
\neg \exists({\color{blue}Y}{\color{red}\xhookrightarrow{\bar{n}} N_{X'}})\\
&\quad\text{with }{\color{red}N_{X'}}:= \pO{{\color{red}N}{\color{h1color}\xhookleftarrow{n'}X'\xhookrightarrow{y'}}{\color{blue}Y}}
\end{array}$}}
\end{equation}
Specializing this formula further to the setting of the statement of the lemma, i.e., for $X=A_1+A_2$, $N=A$ (with the additional special property that $A$ possesses an edge $e\in E_A$ such that $A\setminus e = A_1+A_2$), and for $Y=B_1+B_2$, one finds
\begin{equation}
\vcenter{\hbox{\ti{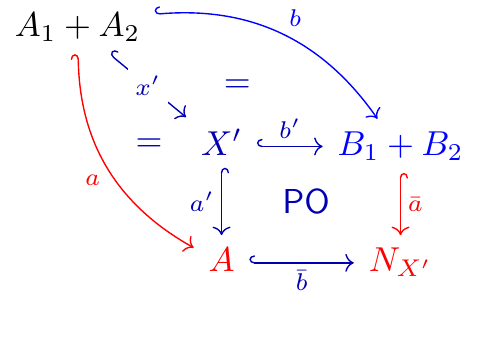}}}\qquad 
\vcenter{\hbox{$\begin{array}{rl}
&\Shift(A_1+A_2{\color{blue}\hookrightarrow B_1+B_2}, \neg\exists (A_1+A_2{\color{red}\hookrightarrow A}))\\
&\quad:= 
\bigwedge\limits_{\substack{(a',x',b')\in \cM^{\times\:3}\\ a'\circ x'= a\land b'\circ x'= b}}
\neg \exists({\color{blue}B_1+B_2}{\color{red}\xhookrightarrow{\bar{a}} N_{X'}})\\
&\qquad\text{with }{\color{red}N_{X'}}:= \pO{{\color{red}A}{\color{h1color}\xhookleftarrow{a'}X'\xhookrightarrow{b'}}{\color{blue}B_1+B_2}}
\end{array}$}}
\end{equation}
But since $A_1+A_2$ and $A$ only differ by a single edge $e$ (i.e., by a ``bridge''), it is straightforward to verify that the (up to span-isomorphisms) unique triple of $\cM$-morphisms $(a',x',b')\in \cM^{\times\:3}$ that satisfies $a'\circ x'= a$ and $b'\circ x'= b$ is $(a,id_{A_1+A_2},b)$, from which the claim follows.

\section{Details on the symbolic solution to the observable average counts in Example~\ref{ex:ugModel}}\label{app:se}

The ODE system of Example~\ref{ex:ugModel} may be solved in closed form as follows:
\begin{equation}
\begin{aligned}
\langle O_{\bullet}\rangle(t)&=\tfrac{\nu_{+}}{\nu_{-}} \left(1-e^{-t \nu_{-}}\right)\\
\langle O_{\bullet\vert\bullet}\rangle(t)&=\tfrac{\nu_{+}^2 e^{-\alpha  t} }{2 \alpha  \beta  \lambda  \nu_{-}^2}\left(\alpha  \beta  \varepsilon_{-} e^{\lambda  t}+2 \varepsilon_{+} \nu_{-}^2-2 \alpha  \kappa  \lambda  e^{\beta  t}+\beta  \lambda  \omega  e^{\alpha  t}\right)\\
\langle O_{\bullet\!-\!\bullet}\rangle(t)&=\tfrac{\varepsilon_{+} \nu_{+}^2 e^{-\alpha  t}}{2 \alpha  \beta  \lambda  \nu_{-}^2}\left(\alpha  \beta  \
e^{\lambda  t}-2 \alpha  \lambda  e^{\beta  t}+\beta  \lambda  \
e^{\alpha  t}-2 \nu_{-}^2\right)\\
\alpha&=\varepsilon_{-}+\varepsilon_{+}+2 \nu_ {-}\,,\;
 \beta=\varepsilon_{-}+\varepsilon_{+}+\nu_ {-}\\
\kappa&=\varepsilon_{-}+\nu_ {-}\,,\;
\lambda=\varepsilon_{-}+\varepsilon_{+}\,,\;
\omega=\varepsilon_{-}+2 \nu_ {-}\,.
\end{aligned}
\end{equation}
In particular, one may provide asymptotic formulae for $t\to\infty$:
\begin{equation}
\begin{aligned}
\langle O_{\bullet}\rangle(t)&\xrightarrow{t\to\infty}\tfrac{\nu_{+}}{\nu_{-}}\\
\langle O_{\bullet\vert\bullet}\rangle(t)&\xrightarrow{t\to\infty}
\tfrac{\nu_ {+}^2 (\varepsilon_{-}+2 \nu_{-})}{2 \nu_ {-}^2 (\varepsilon_{-}+\varepsilon_{+}+2 \nu_{-})}\\
\langle O_{\bullet\!-\!\bullet}\rangle(t)&\xrightarrow{t\to\infty}\tfrac{\varepsilon_{+} \nu_ {+}^2}{2 \nu_ {-}^2 (\varepsilon_{-}+\varepsilon_{+}+2 \nu_{-})}\,.
\end{aligned}
\end{equation}

\section{Technical details of typesetting \MOD{} rules}\label{app:ME}

For the interested readers, the following code may be used in either a standalone instance
or via the \href{https://cheminf.imada.sdu.dk/mod/}{live playground} of \MOD{}~\cite{Andersen_2016} in order to reproduce the graphics in Figure~\ref{fig:formoseRules}.
Note that since \MOD{} employs the traditional ``left-to-right'' convention for rules, the input and output patterns are given as ``left'' and ``right'', respectively.
\begin{python}
aldolAdd = ruleGMLString("""
rule [
        ruleID "Aldol Addition ->"
        left [
                edge [ source 1 target 2 label "=" ]
                edge [ source 2 target 3 label "-" ]
                edge [ source 3 target 4 label "-" ]
                edge [ source 5 target 6 label "=" ]
        ]
        context [
                node [ id 1 label "C" ]
                node [ id 2 label "C" ]
                node [ id 3 label "O" ]
                node [ id 4 label "H" ]
                node [ id 5 label "O" ]
                node [ id 6 label "C" ]
        ]
        right [
                edge [ source 1 target 2 label "-" ]
                edge [ source 2 target 3 label "=" ]
                edge [ source 5 target 6 label "-" ]

                edge [ source 4 target 5 label "-" ]
                edge [ source 6 target 1 label "-" ]
        ]
]
""")
# Printing of the rule:
aldolAdd.print()
\end{python}

\section{Composition Counts}
\label{app:count}

We present in Table~\ref{tab:res} a collection of data generated via performing rule composition operations of the aldol addition rule $r_{+}$ with the rule $r_{id}$ as defined in Figure~\ref{fig:formoseRules}, both in automated fashion via the composition operation $\bullet$ of the \MOD{} framework (cf.\ \ref{app:modcode} for further details) and via a manual computation, the latter employing the restricted DPO-type rule-algebraic composition operation of Definition~\ref{def:RRA}. 

\begin{table}[h!]
\centering
\begin{tabular}{@{}lcccc@{}}
\toprule
Rule & $|r_+ \bullet r_{id}|$ & $|r_{id} \bullet r_+|$ & C1 & C2 \\
\midrule
$r_{+}$       &                      &        1              &   &\\
$r_{2}$       &         1             &       1               &  &\\
$r_{7}$       &         1             &       1               & X &\\
$r_{12}$      &                      &        1              &  & X\\
$r_{18}$      &                      &        1              &  &\\
$r_{29}$      &                      &        1              & X & X\\
$r_{42}$      &                      &        1              &  & X\\
$r_{47}$      &                      &        1              & X & X\\
$r_{68}$      &                      &        1              & X &\\
$r_{73}$      &                      &        1              & X &\\
$r_{79}$      &                      &        1             &  &\\
$r_{89}$      &                      &        1              &  &\\
$r_{94}$      &                      &        1              & X & X\\
$r_{107}$     &        1              &       1               & X &\\
$r_{120}$     &                      &        1              & X &\\
$r_{180}$     &        1              &       1               &  &\\
$r_{185}$     &        1              &       1               & X &\\
$r_{191}$     &        1              &       1               & X & X\\
$r_{234}$     &        1              &                      & X & \\
$r_{252}$     &        1              &                      & X & \\
$r_{258}$     &        1              &                      & X & \\
\arrayrulecolor{red}
  \midrule
\arrayrulecolor{black}
${\color{red}r_{A}}$     &           1           &                      & X &\\
${\color{red}r_{B}}$     &           1           &                     & X &\\
${\color{red}r_{C}}$     &           1          &                      & X &\\
\bottomrule
\end{tabular}
\caption{Overview of rules resulting from composition of $r_{id}$ with $r_+$  (as in Figure~\ref{fig:formoseRules}) and vice versa, and annotated by isomorphism classes of composite rules. Each row represents a unique composed rule up to isomorphism, with the first column being the name of the rule.
The second and third column indicates the number of $\cM$-spans that result in the corresponding rule (with empty entries encoding 0 occurrences). The last three rules listed (highlighted in {\color{red}red}) are rules not found by the \MOD{} operator $\bullet$, yet which were computed manually as contributions to the full DPO-type rule composition operation. An X in the the fourth column (C1) indicates that the composed rule violates convenience constraint 1 that disallows carbons with tow incident double bonds. The fifth column (C2) similarly indicates violation of convenience constraint 2, that the rule contains carbons with more than one adjacent oxygen.
Auto-generated depictions of all the compositions can be found in \ref{app:comp}.}
\label{tab:res}
\end{table}

\section{Automatically Inferred Compositions}
\label{app:comp}
\subsection{Compositions, $r_{id} \bullet r_+$}

\gdef\aligntemp{}%
\foreach \i in {5,6,...,22}{%
\ifnum \i<22
	\xappto\aligntemp{\vcenter{\hbox{\miScaled{\i}{0.7}}}  \noexpand\\ }%
\else
	\xappto\aligntemp{\vcenter{\hbox{\miScaled{\i}{0.7}}}  \noexpand }%
\fi
}%

{\allowdisplaybreaks
\begin{subequations}
\begin{gather}
\aligntemp
\end{gather}
\end{subequations}
}

\subsection{Compositions, $r_+ \bullet r_{id}$}

\gdef\aligntemp{}%
\foreach \i in {23,24,...,31}{%
\ifnum \i<31
	\xappto\aligntemp{\vcenter{\hbox{\miScaled{\i}{0.7}}}  \noexpand\\ }%
\else
	\xappto\aligntemp{\vcenter{\hbox{\miScaled{\i}{0.7}}}  \noexpand }%
\fi
}%

{\allowdisplaybreaks
\begin{subequations}
\begin{gather}
\aligntemp
\end{gather}
\end{subequations}
}

\newpage

\section{\MOD{} Python Code Example}
\label{app:modcode}

In Figure~\ref{fig:code} we present a Python code fragment that illustrates how the composed rules $r_+ \bullet r_{id}$ and $r_{id} \bullet r_+$ were calculated with the \MOD{} framework.
We note that a range of composition operators are supported in \MOD{}.
The operator chosen for the example in this paper is the most generic one currently available.
As empty overlaps are excluded with this operator for purely technical reasons, we explicitly add the rules stemming from the parallel composition to the result containers.
Note, that \MOD{} contains a function to check for isomorphic rules.
Furthermore, \MOD{} supports several other convenience methods, e.g., to print DPO diagrams or rule composition diagrams (cf.\ Figure~\ref{fig:modcomp} and \ref{app:comp}).
\begin{figure}
\begin{python}
[...]

def compose(r1: Rule, r2: Rule, rc: RCEvaluator) -> Tuple[CompRes, CompRes]:
        comp = rcCommon(connected=False, maximum=False)
        res12 = checkRules(
                rc.eval(rcExp([
                        r1 *comp* r2,
                        r1 *rcParallel* r2 
                ])))
        res21 = checkRules(
                rc.eval(rcExp([
                        r2 *comp* r1,
                        r2 *rcParallel* r1
                ])))
        return res12, res21
        
r1 = aldolAdd_F_id
r2 = aldolAdd_F

res = compose(r1, r2, rc)

[...]
\end{python}
\caption{Python code to compute the composed rules with \MOD{}.
\texttt{comp}: defined to be the most general composition operator with non-empty overlap.
\texttt{checkRules}: post-processing to filter our non-chemical rules based on valence constraints.
Note, that \MOD{} has the feature to easily check for rules being isomorphic.
Furthermore, \MOD{} allows to easily print rules, DPO diagrams, as well as rule-composition diagrams (cf.\ Figure~\ref{fig:modcomp} and \ref{app:comp}).
} 
\label{fig:code}   
\end{figure}


\begin{thebibliography}{44}
\expandafter\ifx\csname url\endcsname\relax
  \def\url#1{\texttt{#1}}\fi
\expandafter\ifx\csname urlprefix\endcsname\relax\def\urlprefix{URL }\fi
\expandafter\ifx\csname href\endcsname\relax
  \def\href#1#2{#2} \def\path#1{#1}\fi

\bibitem{BK2020}
N.~Behr, J.~Krivine, {Rewriting theory for the life sciences: A unifying
  framework for CTMC semantics}, in: F.~Gadducci, T.~Kehrer (Eds.), Graph
  Transformation (ICGT 2020), Vol. 12150 of Theoretical Computer Science and
  General Issues, Springer International Publishing, 2020, pp. 185--202.
\newblock \href {https://doi.org/10.1007/978-3-030-51372-6}
  {\path{doi:10.1007/978-3-030-51372-6}}.

\bibitem{Delbr_ck_1940}
M.~Delbr{\"u}ck, {Statistical Fluctuations in Autocatalytic Reactions}, The
  Journal of Chemical Physics 8~(1) (1940) 120--124.
\newblock \href {https://doi.org/10.1063/1.1750549}
  {\path{doi:10.1063/1.1750549}}.

\bibitem{bp2019-ext}
N.~Behr, P.~Sobocinski, \href{https://lmcs.episciences.org/6615}{{Rule Algebras
  for Adhesive Categories (extended journal version)}}, {Logical Methods in
  Computer Science} {Volume 16, Issue 3} (Jul. 2020).
\newline\urlprefix\url{https://lmcs.episciences.org/6615}

\bibitem{Boutillier:2018aa}
P.~Boutillier, M.~Maasha, X.~Li, H.~F. Medina-Abarca, J.~Krivine, J.~Feret,
  I.~Cristescu, A.~G. Forbes, W.~Fontana, The kappa platform for rule-based
  modeling, Bioinformatics 34~(13) (2018) i583--i592.
\newblock \href {https://doi.org/10.1093/bioinformatics/bty272}
  {\path{doi:10.1093/bioinformatics/bty272}}.

\bibitem{Andersen_2016}
J.~L. Andersen, C.~Flamm, D.~Merkle, P.~F. Stadler, {A Software Package for
  Chemically Inspired Graph Transformation}, in: Graph Transformation, Springer
  International Publishing, 2016, pp. 73--88.
\newblock \href {https://doi.org/10.1007/978-3-319-40530-8_5}
  {\path{doi:10.1007/978-3-319-40530-8_5}}.

\bibitem{bp2018}
N.~Behr, P.~Sobocinski, {Rule Algebras for Adhesive Categories}, in: D.~Ghica,
  A.~Jung (Eds.), 27th EACSL Annual Conference on Computer Science Logic (CSL
  2018), Vol. 119 of LIPIcs, Schloss Dagstuhl--Leibniz-Zentrum fuer Informatik,
  Dagstuhl, Germany, 2018, pp. 11:1--11:21.
\newblock \href {https://doi.org/10.4230/LIPIcs.CSL.2018.11}
  {\path{doi:10.4230/LIPIcs.CSL.2018.11}}.

\bibitem{nbSqPO2019}
N.~Behr, {Sesqui-Pushout Rewriting: Concurrency, Associativity and Rule Algebra
  Framework}, in: R.~Echahed, D.~Plump (Eds.), {Proceedings of theTenth
  International Workshop on Graph Computation Models
  (\href{http://gcm2019.imag.fr}{GCM 2019}) in Eindhoven, The Netherlands},
  Vol. 309 of Electronic Proceedings in Theoretical Computer Science, Open
  Publishing Association, 2019, pp. 23--52.
\newblock \href {https://doi.org/10.4204/eptcs.309.2}
  {\path{doi:10.4204/eptcs.309.2}}.

\bibitem{habel2009correctness}
A.~Habel, K.-H. Pennemann, Correctness of high-level transformation systems
  relative to nested conditions, Mathematical Structures in Computer Science
  19~(02) (2009) 245.
\newblock \href {https://doi.org/10.1017/s0960129508007202}
  {\path{doi:10.1017/s0960129508007202}}.

\bibitem{behrRaSiR}
N.~Behr, J.~Krivine, Compositionality of {R}ewriting {R}ules with {C}onditions,
  {Compositionality} 3 (2021).
\newblock \href {https://doi.org/10.32408/compositionality-3-2}
  {\path{doi:10.32408/compositionality-3-2}}.

\bibitem{ehrig:2006fund}
H.~Ehrig, K.~Ehrig, U.~Prange, G.~Taentzer, {Fundamentals of Algebraic Graph
  Transformation}, Monographs in Theoretical Computer Science (An EATCS Series)
  (2006).
\newblock \href {https://doi.org/10.1007/3-540-31188-2}
  {\path{doi:10.1007/3-540-31188-2}}.

\bibitem{Corradini_2006}
A.~Corradini, T.~Heindel, F.~Hermann, B.~K{\"o}nig, {Sesqui-Pushout Rewriting},
  in: A.~Corradini, H.~Ehrig, U.~Montanari, L.~Ribeiro, G.~Rozenberg (Eds.),
  Graph Transformations, Vol. 4178 of Lecture Notes in Computer Science,
  Springer Berlin Heidelberg, Berlin, Heidelberg, 2006, pp. 30--45.

\bibitem{lack2005adhesive}
S.~Lack, P.~Soboci{\'{n}}ski, {Adhesive and quasiadhesive categories}, {RAIRO}
  - Theoretical Informatics and Applications 39~(3) (2005) 511--545.
\newblock \href {https://doi.org/10.1051/ita:2005028}
  {\path{doi:10.1051/ita:2005028}}.

\bibitem{ehrig2004adhesive}
H.~Ehrig, A.~Habel, J.~Padberg, U.~Prange, {Adhesive High-Level Replacement
  Categories and Systems}, in: Lecture Notes in Computer Science, Springer
  Berlin Heidelberg, 2004, pp. 144--160.
\newblock \href {https://doi.org/10.1007/978-3-540-30203-2_12}
  {\path{doi:10.1007/978-3-540-30203-2_12}}.

\bibitem{GABRIEL_2014}
K.~Gabriel, B.~Braatz, H.~Ehrig, U.~Golas, Finitary $\mathcal{M}$-adhesive
  categories, Mathematical Structures in Computer Science 24~(04) (2014).
\newblock \href {https://doi.org/10.1017/S0960129512000321}
  {\path{doi:10.1017/S0960129512000321}}.

\bibitem{ehrig2014mathcal}
H.~Ehrig, U.~Golas, A.~Habel, L.~Lambers, F.~Orejas, {$\mathcal{M}$-adhesive
  transformation systems with nested application conditions. Part 1:
  parallelism, concurrency and amalgamation}, Mathematical Structures in
  Computer Science 24~(04) (2014).
\newblock \href {https://doi.org/10.1017/s0960129512000357}
  {\path{doi:10.1017/s0960129512000357}}.

\bibitem{ehrig2012m}
H.~Ehrig, U.~Golas, A.~Habel, L.~Lambers, F.~Orejas, {$\mathcal{M}$-Adhesive
  Transformation Systems with Nested Application Conditions. Part 2: Embedding,
  Critical Pairs and Local Confluence}, Fundamenta Informaticae 118~(1-2)
  (2012) 35--63.
\newblock \href {https://doi.org/10.3233/FI-2012-705}
  {\path{doi:10.3233/FI-2012-705}}.

\bibitem{padberg2017towards}
J.~Padberg, {Towards M-Adhesive Categories based on Coalgebras and Comma
  Categories}, \href{http://arxiv.org/abs/1702.04650}{arXiv:1702.04650} (2017).
\newblock \href {http://arxiv.org/abs/1702.04650} {\path{arXiv:1702.04650}}.

\bibitem{bdg2016}
N.~Behr, V.~Danos, I.~Garnier, Stochastic mechanics of graph rewriting, in:
  Proceedings of the 31st Annual {ACM}/{IEEE} Symposium on Logic in Computer
  Science - {LICS} {'}16, {ACM} Press, 2016, p. 46–55.
\newblock \href {https://doi.org/10.1145/2933575.2934537}
  {\path{doi:10.1145/2933575.2934537}}.

\bibitem{bdg2019}
N.~Behr, V.~Danos, I.~Garnier,
  \href{https://lmcs.episciences.org/6628}{{Combinatorial Conversion and Moment
  Bisimulation for Stochastic Rewriting Systems}}, {Logical Methods in Computer
  Science} {Volume 16, Issue 3} (Jul. 2020).
\newline\urlprefix\url{https://lmcs.episciences.org/6628}

\bibitem{Danos2014}
V.~Danos, R.~Heckel, P.~Sobocinski, {Transformation and Refinement of Rigid
  Structures}, in: Graph Transformation (ICGT 2014), Vol. 8571 of LNCS,
  Springer International Publishing, 2014, pp. 146--160.
\newblock \href {https://doi.org/10.1007/978-3-319-09108-2_10}
  {\path{doi:10.1007/978-3-319-09108-2_10}}.

\bibitem{Pennemann:aa}
K.-H. Pennemann, {Resolution-Like Theorem Proving for High-Level Conditions},
  in: H.~Ehrig, R.~Heckel, G.~Rozenberg, G.~Taentzer (Eds.), Graph
  Transformations (ICGT 2008), Vol. 5214 of Lecture Notes in Computer Science,
  Springer Berlin Heidelberg, 2008, pp. 289--304.
\newblock \href {https://doi.org/10.1007/978-3-540-87405-8_20}
  {\path{doi:10.1007/978-3-540-87405-8_20}}.

\bibitem{danos2004computational}
V.~Danos, V.~Schachter, {Computational Methods in Systems Biology}, in:
  Conference proceedings CMSB, Springer, 2004, p.~91.
\newblock \href {https://doi.org/10.1007/b107287} {\path{doi:10.1007/b107287}}.

\bibitem{danos2004formal}
V.~Danos, C.~Laneve, Formal molecular biology, Theoretical Computer Science
  325~(1) (2004) 69--110.
\newblock \href {https://doi.org/10.1016/j.tcs.2004.03.065}
  {\path{doi:10.1016/j.tcs.2004.03.065}}.

\bibitem{Feret_2008}
J.~Feret, V.~Danos, J.~Krivine, R.~Harmer, W.~Fontana, Internal coarse-graining
  of molecular systems, Proceedings of the National Academy of Sciences
  106~(16) (2009) 6453--6458.
\newblock \href {https://doi.org/10.1073/pnas.0809908106}
  {\path{doi:10.1073/pnas.0809908106}}.

\bibitem{Danos_2010}
V.~Danos, J.~Feret, W.~Fontana, R.~Harmer, J.~Krivine, Abstracting the
  differential semantics of rule-based models: Exact and automated model
  reduction, in: 2010 25th Annual IEEE Symposium on Logic in Computer Science,
  2010, pp. 362--381.
\newblock \href {https://doi.org/10.1109/LICS.2010.44}
  {\path{doi:10.1109/LICS.2010.44}}.

\bibitem{Harmer_2010}
R.~Harmer, V.~Danos, J.~Feret, J.~Krivine, W.~Fontana, Intrinsic information
  carriers in combinatorial dynamical systems, Chaos: An Interdisciplinary
  Journal of Nonlinear Science 20~(3) (2010) 037108.
\newblock \href {https://doi.org/10.1063/1.3491100}
  {\path{doi:10.1063/1.3491100}}.

\bibitem{danos_et_al:LIPIcs:2012:3866}
V.~Danos, J.~Feret, W.~Fontana, R.~Harmer, J.~Hayman, J.~Krivine,
  C.~Thompson-Walsh, G.~Winskel, {Graphs, Rewriting and Pathway Reconstruction
  for Rule-Based Models}, in: D.~D'Souza, T.~Kavitha, J.~Radhakrishnan (Eds.),
  IARCS Annual Conference on Foundations of Software Technology and Theoretical
  Computer Science (FSTTCS 2012), Vol.~18 of Leibniz International Proceedings
  in Informatics (LIPIcs), Schloss Dagstuhl--Leibniz-Zentrum fuer Informatik,
  Dagstuhl, Germany, 2012, pp. 276--288.
\newblock \href {https://doi.org/10.4230/LIPIcs.FSTTCS.2012.276}
  {\path{doi:10.4230/LIPIcs.FSTTCS.2012.276}}.

\bibitem{danos2015moment}
V.~Danos, T.~Heindel, R.~Honorato-Zimmer, S.~Stucki, Moment semantics for
  reversible rule-based systems, in: International Conference on Reversible
  Computation, Springer, 2015, pp. 3--26.
\newblock \href {https://doi.org/10.1007/978-3-319-20860-2_1}
  {\path{doi:10.1007/978-3-319-20860-2_1}}.

\bibitem{danos2020}
V.~Danos, T.~Heindel, R.~Honorato-Zimmer, S.~Stucki, {Rate Equations for
  Graphs}, in: A.~Abate, T.~Petrov, V.~Wolf (Eds.), Computational Methods in
  Systems Biology, Vol. 12314 of Lecture Notes in Computer Science, Springer
  International Publishing, Cham, 2020, pp. 3--26.
\newblock \href {https://doi.org/10.1007/978-3-030-60327-4_1}
  {\path{doi:10.1007/978-3-030-60327-4_1}}.

\bibitem{Kappa_manual}
P.~Boutillier, J.~Feret, J.~Krivine, W.~Fontana, {The Kappa Language and
  Tools}, Tech. rep.,
  \href{https://kappalanguage.org/sites/kappalanguage.org/files/inline-files/Kappa_Manual.pdf}{Kappalanguage.org}
  (03 2020).

\bibitem{Boutillier17}
P.~Boutillier, T.~Ehrhard, J.~Krivine, Incremental update for graph rewriting,
  in: H.~Yang (Ed.), Programming Languages and Systems, Springer Berlin
  Heidelberg, Berlin, Heidelberg, 2017, pp. 201--228.

\bibitem{ANDREI200867}
O.~Andrei, H.~Kirchner, A rewriting calculus for multigraphs with ports, ENTCS
  219 (2008) 67 -- 82, proceedings of the Eighth International Workshop on Rule
  Based Programming (RULE 2007).
\newblock \href {https://doi.org/https://doi.org/10.1016/j.entcs.2008.10.035}
  {\path{doi:https://doi.org/10.1016/j.entcs.2008.10.035}}.

\bibitem{danos2008rule}
V.~Danos, J.~Feret, W.~Fontana, R.~Harmer, J.~Krivine, Rule-based modelling,
  symmetries, refinements, in: J.~Fisher (Ed.), Formal Methods in Systems
  Biology, Springer Berlin Heidelberg, Berlin, Heidelberg, 2008, pp. 103--122.
\newblock \href {https://doi.org/10.1007/978-3-540-68413-8_8}
  {\path{doi:10.1007/978-3-540-68413-8_8}}.

\bibitem{strat:14}
J.~L. Andersen, C.~Flamm, D.~Merkle, P.~F. Stadler, Generic strategies for
  chemical space exploration, International Journal of Computational Biology
  and Drug Design 7~(2/3) (2014) 225 -- 258.
\newblock \href {https://doi.org/10.1504/IJCBDD.2014.061649}
  {\path{doi:10.1504/IJCBDD.2014.061649}}.

\bibitem{Andersen_2019}
J.~L. Andersen, C.~Flamm, D.~Merkle, P.~F. Stadler, {Chemical Transformation
  Motifs --- Modelling Pathways as Integer Hyperflows}, {IEEE}/{ACM}
  Transactions on Computational Biology and Bioinformatics 16~(2) (2019)
  510--523.
\newblock \href {https://doi.org/10.1109/tcbb.2017.2781724}
  {\path{doi:10.1109/tcbb.2017.2781724}}.

\bibitem{andersen2018rule}
J.~L. Andersen, C.~Flamm, D.~Merkle, P.~F. Stadler, {Rule composition in graph
  transformation models of chemical reactions}, Match 80~(3) (2018) 661--704.

\bibitem{trace:14}
J.~L. Andersen, C.~Flamm, D.~Merkle, P.~F. Stadler, 50 shades of rule
  composition, in: F.~Fages, C.~Piazza (Eds.), Formal Methods in Macro-Biology,
  Vol. 8738 of Lecture Notes in Computer Science, Springer International
  Publishing, 2014, pp. 117--135.
\newblock \href {https://doi.org/10.1007/978-3-319-10398-3_9}
  {\path{doi:10.1007/978-3-319-10398-3_9}}.

\bibitem{Habel_2012}
A.~Habel, D.~Plump, {$\mathcal{M}, \mathcal{N}$ -Adhesive Transformation
  Systems}, in: H.~Ehrig, G.~Engels, H.~Kreowski, G.~Rozenberg (Eds.), Graph
  Transformations (ICGT 2012), Vol. 7562 of Lecture Notes in Computer Science,
  Springer Berlin Heidelberg, 2012, pp. 218--233.
\newblock \href {https://doi.org/10.1007/978-3-642-33654-6_15}
  {\path{doi:10.1007/978-3-642-33654-6_15}}.

\bibitem{quasi-topos-2007}
P.~T. Johnstone, S.~Lack, P.~Soboci{\'{n}}ski, {Quasitoposes, Quasiadhesive
  Categories and Artin Glueing}, in: T.~Mossakowski, U.~Montanari, M.~Haveraaen
  (Eds.), Algebra and Coalgebra in Computer Science, Vol. 4624 of Lecture Notes
  in Computer Science, Springer Berlin Heidelberg, 2007, pp. 312--326.
\newblock \href {https://doi.org/10.1007/978-3-540-73859-6_21}
  {\path{doi:10.1007/978-3-540-73859-6_21}}.

\bibitem{10.5555/504206}
J.~G. Siek, L.~Lee, A.~Lumsdaine, {The Boost Graph Library: User Guide and
  Reference Manual}, Addison-Wesley Longman Publishing Co., Inc., USA, 2002.

\bibitem{behr2019tracelets}
N.~Behr, Tracelets and tracelet analysis of compositional rewriting systems,
  in: J.~Baez, B.~Coecke (Eds.), {\rm Proceedings} Applied Category Theory
  2019, Vol. 323 of EPTCS, Open Publishing Association, 2020, pp. 44--71.
\newblock \href {https://doi.org/10.4204/EPTCS.323.4}
  {\path{doi:10.4204/EPTCS.323.4}}.

\bibitem{ferethal00975861}
J.~Feret, H.~Koeppl, T.~Petrov, {Stochastic fragments: A framework for the
  exact reduction of the stochastic semantics of rule-based models},
  {International Journal of Software and Informatics (IJSI)} 7~(4) (2014)
  527--604.

\bibitem{ehrig2010categorical}
H.~Ehrig, U.~Golas, F.~Hermann, {Categorical frameworks for graph
  transformation and HLR systems based on the DPO approach}, Bulletin of the
  EATCS~(102) (2010) 111--121.

\bibitem{norris}
J.~R. Norris, {Markov Chains}, Cambridge University Press, 1997.
\newblock \href {https://doi.org/10.1017/cbo9780511810633}
  {\path{doi:10.1017/cbo9780511810633}}.

\end{thebibliography}
\end{document}